\documentclass[letterpaper,USenglish,cleveref,numberwithinsect,thm-restate]{no-lipics-v2022}
\usepackage[utf8]{inputenc}
\usepackage{color}
\usepackage{csquotes}
\usepackage{xspace}
\usepackage{amsmath}
\usepackage{amssymb}
\usepackage{xcolor}
\usepackage{mathtools}
\usepackage{bbm}
\usepackage[linesnumbered, noend]{algorithm2e}
\usepackage[draft]{fixme}
\usepackage{tikz}
\usetikzlibrary{arrows,arrows.meta,decorations.pathreplacing,decorations.pathmorphing,shapes,calc,patterns,shapes,matrix,math,quotes}
\fxsetup{mode=multiuser,theme=color, layout=inline}
\FXRegisterAuthor{cj}{acj}{\color{red} {\bf  Ce}}
\FXRegisterAuthor{tk}{atk}{\color{blue} {\bf Tomasz}}

\newcommand{\Oh}{\ensuremath{\mathcal{O}}\xspace}
\newcommand{\Ohtilde}{\ensuremath{\smash{\rlap{\raisebox{-0.2ex}{$\widetilde{\phantom{\Oh}}$}}\Oh}}\xspace}
\newcommand{\Thtilde}{\ensuremath{\smash{\rlap{\raisebox{-0.2ex}{$\widetilde{\phantom{\Theta}}$}}\Theta}}\xspace}
\newcommand{\Omtilde}{\ensuremath{\smash{\rlap{\raisebox{-0.2ex}{$\widetilde{\phantom{\Omega}}$}}\Omega}}\xspace}

\newcommand{\ZZ}{\mathbb{Z}\xspace}

\DeclareMathOperator{\poly}{poly}
\DeclareMathOperator{\polylog}{polylog}

\newcommand{\dd}{.\,.}

\def\emptystring{\ensuremath\varepsilon}

\def\twoheadleadsto{\tikz[baseline=(a.base)]{\draw[%
    decorate,decoration={zigzag,segment length=4, amplitude=.9},%
    ] (0,0) -- (.25, 0);%
    \draw[%
    -{Classical TikZ Rightarrow}.{Classical TikZ Rightarrow},%
    ] (.25, 0) -- (.4, 0);%
    \node (a) at (.4/2,-.03) {\phantom{\(\leadsto\)}};%
}}

\newcommand{\onto}{\twoheadleadsto}
\def\aonto#1{\onto}

\newcommand{\fragmentco}[2]{[#1\dd #2)}
\newcommand{\fragmentoc}[2]{(#1\dd #2]}
\newcommand{\fragmentoo}[2]{(#1\dd #2)}
\newcommand{\fragmentcc}[2]{[#1\dd #2]}
\newcommand{\position}[1]{[#1]}

\newcommand{\floor}[1]{\lfloor{#1}\rfloor}
\newcommand{\ceil}[1]{\lceil{#1}\rceil}

\newcommand{\Zp}{\mathbb{Z}_{>0}}

\newcommand{\cU}{\mathcal{U}}

\newcommand{\eps}{\varepsilon}

\definecolor{darkgreen}{RGB}{0,160,0}
\definecolor{darkred}{RGB}{220,20,60}
\definecolor{darkblue}{RGB}{0,0,160}

\newcommand{\HD}{\mathsf{HD}}
\newcommand{\MM}{\mathsf{MM}}
\newcommand{\Occ}{\mathsf{Occ}}

\newcommand{\Exp}{\mathbf{E}}
\renewcommand{\Pr}{\mathbf{P}}

\newcommand{\per}{\mathrm{per}}

\newcommand{\Ber}{\mathsf{Ber}}
\newcommand{\Random}{\mathsf{Random}}
\newcommand{\Planted}{\mathsf{Planted}}
\newcommand{\Mixed}{\mathsf{Mixed}}

\newcommand{\Runs}{\mathcal{R}}
\newcommand{\Bad}{\mathcal{B}}
\newcommand{\RunsNo}{\Runs_{\textsc{No}}}

\newcommand{\Sparse}{\mathsf{Sparse}}
\newcommand{\Positive}{\mathsf{Equal}}
\newcommand{\Negative}{\mathsf{Independent}}
\newcommand{\Hybrid}{\mathsf{Hybrid}}

\newcommand{\Lang}{\mathcal{L}}

\newcommand{\Yes}{\textsc{Yes}\xspace}
\newcommand{\No}{\textsc{No}\xspace}

\newtheorem{problem}[theorem]{Problem}
\crefname{problem}{Problem}{Problems}
\newcommand{\defproblem}[4]{
  \vspace{2mm}
  \noindent\fbox{
  \begin{minipage}{0.96\textwidth}
\begin{problem}\label{#1}
  \textsf{#2}

  \smallskip
  \noindent
  {\bf{Input:}} #3

  \smallskip
  \noindent
  {\bf{Output:}} #4 \lipicsEnd
\end{problem}
  \end{minipage}
  }
  \vspace{2mm}
}

\nolinenumbers

\title{Near-Optimal Property Testers for Pattern Matching}

\author{Ce Jin}{University of California, Berkeley, CA, USA}{cejin@berkeley.edu}{https://orcid.org/0000-0001-5264-1772}{}%

\author{Tomasz Kociumaka}{Max Planck Institute for Informatics, Saarland Informatics Campus, Germany}{tomasz.kociumaka@mpi-inf.mpg.de}{https://orcid.org/0000-0002-2477-1702}{}%

\authorrunning{C. Jin and T. Kociumaka} %

\Copyright{Ce Jin and Tomasz Kociumaka} %

\begin{document}

\maketitle

\begin{abstract}
The classic \emph{exact pattern matching} problem, given two strings---a pattern $P$ of length $m$ and a text $T$ of length~$n$---asks whether $P$ occurs as a substring of $T$, that is, $P = T\fragmentco{i}{{i+m}}$ holds for some $i \in \fragmentcc{0}{n-m}$.
A \emph{property tester} for the problem needs to distinguish (with high probability) the following two cases for some threshold $k \in [1\dd m)$: the \textsc{Yes} case, where $P$ occurs as a substring of $T$,
and the \textsc{No} case, where $P$ has Hamming distance greater than $k$ from every substring of $T$,  that is, $P$ has no $k$-mismatch occurrence in $T$.

In this work, we provide adaptive and non-adaptive property testers for the exact pattern matching problem, jointly covering the whole spectrum of parameters. 
We further establish unconditional lower bounds demonstrating that the time and query complexities of our algorithms are optimal, up to $\polylog n$ factors hidden within the $\Ohtilde(\cdot)$ notation below.

In the most studied regime of $n = m + \Theta(m)$, our non-adaptive property tester has the time complexity of $\Ohtilde(n/\sqrt{k})$, and a matching lower bound remains valid for the query complexity of adaptive algorithms.
This improves both upon a folklore solution that attains the optimal query complexity but requires $\Omega(n)$ time, and upon the only previously known sublinear-time property tester, by Chan, Golan, Kociumaka, Kopelowitz, and Porat [STOC 2020], with time complexity $\Ohtilde(n/\sqrt[3]{k})$.
The aforementioned results remain valid for $n = m + \Omega(m)$, where our optimal running time $\Ohtilde(\sqrt{nm/k}+n/k)$ improves upon the previously best time complexity of $\Ohtilde(\sqrt[3]{n^2m/k} + n/k)$.
In the regime of $n = m + o(m)$, which has not been targeted in any previous work, we establish a surprising separation between adaptive and non-adaptive algorithms, whose optimal time and query complexities are $\Ohtilde(\sqrt{(n-m+1)m/k} + n/k)$ and $\Ohtilde(\min(n\sqrt{n-m+1}/k, \sqrt{nm/k} + n/k))$, respectively.

Our non-adaptive algorithms answer \textsc{Yes} with high probability not only when $P$ has an exact occurrence in~$T$ but also when $P$ has an occurrence with at most $k'=\Omega(k/\log n)$ mismatches.
The gap $k/k'$ can be reduced by slightly increasing the running time; an arbitrarily small polynomial overhead already suffices to achieve a constant gap.
Moreover, upon request, our algorithms may output a set $A\subseteq [0\dd n-m]$ that contains the starting positions of all $k'$-mismatch occurrences of $P$ in $T$ and no starting position of an occurrence with more than $k$ mismatches.

The key technical innovation behind all our property testers is a novel characterization of the mismatches between the pattern $P$ and the fragments $T\fragmentco{i}{i+m}$ across $i\in \fragmentcc{0}{n-m}$.
We show that one can select $\Ohtilde(k \cdot n/m)$ positions within $P$ and $T$ so that, for every $i \in \fragmentcc{0}{n-m}$, at least $\min(k, k_i)$ of the $k_i$ mismatches between $P$ and $T\fragmentco{i}{i+m}$ involve a selected position.
Previously, such a construction was known for $k=1$ only.%
\end{abstract}%

\section{Introduction}
Exact pattern matching is the most well-known problem in string algorithms, with classic 50-year-old linear-time solutions~\cite{MP70,KMP77,BM77} taught in many undergraduate courses.  
Given two strings, a pattern $P$ of length $m$ and a text $T$ of length $n \ge m$, the task is to decide whether the pattern $P$ is a substring of the text $T$ or, in the reporting version of the problem, output the set
\[\Occ(P,T) = \{i \in \fragmentcc{0}{n-m} : P = T\fragmentco{i}{i+m}\}\]
representing the exact occurrences of $P$ in $T$.
The decision version is formalized as follows:

\defproblem{pr:pm}{Exact Pattern Matching}{A pattern $P\in \Sigma^m$ and a text $T\in \Sigma^n$.}{\Yes if $\Occ(P,T)\ne \emptyset$ and \No if $\Occ(P,T)=\emptyset$.}{}

The enormous growth of available data over recent decades has rendered even linear-time algorithms too costly in many scenarios. 
At the same time, most decision problems, including \cref{pr:pm}, require at least linear time in the worst case, simply because the entire input must be read.
Nevertheless, one can relax decision problems to their approximate variants in hopes of designing randomized sublinear-time algorithms. 
This idea gave rise to \emph{property testing}, which emerged in the late 1990s~\cite{GGR98} and since developed into a mature field; see~\cite{G17,BY22} for dedicated~textbooks.

In complexity theory, decision problems are typically formalized as languages $\Lang\subseteq \Sigma^*$ so that the task of the algorithm is to decide whether a given instance $I\in \Sigma^*$ belongs to the language~$\Lang$.
A property tester is still required to output \Yes if $I\in \Lang$, but it must answer \No only if $I$ is ``\emph{far}'' from every instance $I'\in \Lang$.
This condition is typically formalized using the Hamming metric, where the distance between two instances $I,I'\in \Sigma^n$ is the number of mismatching characters: $\HD(I,I') =|\{i\in \fragmentco{0}{n} : I\position{i} \ne I'\position{i}\}|$.
Now, for a real parameter $\epsilon \in [0,1]$, the instance $I\in \Sigma^n$ is considered ``far'' from $\Lang$ if and only if $\HD(I,I') > \epsilon n$ holds for every $I'\in \Lang \cap \Sigma^{n}$.
As discussed in~\cite[Chapter 3]{BY22}, efficient property testers are known for several string-processing problems such as deciding whether the input is a palindrome ($\Lang$ is the language of all palindromes) or whether it forms a balanced sequence of parentheses ($\Lang$ is the \emph{Dyck} language of well-parenthesized sequences). 

In this work, we consider property testing for the exact pattern matching problem, where $\Lang$ can be interpreted as the set of all substrings of the text $T$. 
The property tester needs to distinguish whether $P$ is a substring of $T$ or is at Hamming distance more than $\epsilon m$ from every substring $T\fragmentco{i}{i+m}$.%
\footnote{Given that both $P$ and $T$ are parts of the input, it might seem more natural to define ``far'' instances as those for which $\HD(P,P')+\HD(T,T')> \epsilon(n+m)$ holds for every text $T'$ and its substring $P'$. 
Nevertheless, it is easy to check that $T=T'$ can be assumed without loss of generality since the only relevant quantity is $\min_{i\in \fragmentcc{0}{n-m}} \HD(P, T\fragmentco{i}{i+m})$.}
We follow the convention from recent approximate pattern matching literature~\cite{CFPSS16,GU18,ChanGKKP20,CJWX23} and parameterize the problem by $k\coloneqq \floor{\epsilon m}$ instead of $\epsilon$. 
The condition for ``far'' instances can be expressed as the lack of \emph{$k$-mismatch occurrences} defined as \mbox{elements of the following set:}
\[\Occ_k(P,T) = \{i\in \fragmentcc{0}{n-m} : \HD(P, T\fragmentco{i}{i+m})\le k\}.\]

\defproblem{pr:tester}{Property Testing for Exact Pattern Matching}{A pattern $P\in \Sigma^m$, a text $T\in \Sigma^n$, and an integer threshold $k\in \fragmentoo{0}{m}$.}{\Yes if $\Occ(P,T)\ne \emptyset$, \No if $\Occ_k(P,T)=\emptyset$, and an arbitrary answer otherwise.}{}

It is easy to solve \cref{pr:tester} using a sublinear number of queries.
\begin{fact}[Folklore]\label{fct:folklore}
\cref{pr:tester} can be solved (correctly with high probability) by a non-adaptive algorithm with query complexity 
$\Ohtilde\Big(\sqrt{\tfrac{nm\vphantom{i}}{k}} + \tfrac{n}{k}\Big)$.
\end{fact}
\begin{proof}[Proof sketch.]
Sample $R_T \subseteq \fragmentco{0}{n}$ and $R_P\subseteq \fragmentco{0}{m}$ uniformly at random with rates $r_T$ and~$r_P$, respectively, such that $r_P \cdot r_T = \min(1,\frac{2\ln n}{k})$, and read the characters at the sampled positions.
For every $i\in \fragmentcc{0}{n-m}$, return \Yes if $P\position{j}=T\position{j+i}$ holds for all $j\in R_P$ such that $j+i\in R_T$.

If $i\in \Occ(P,T)$, then the algorithm clearly outputs \Yes while processing $i$.
If $i\notin \Occ_k(P,T)$, then the probability of the algorithm returning \Yes while processing $i$ is at most $(1-r_P\cdot r_T)^k \le \exp(-2\ln n) = \frac{1}{n^2}$.
By the union bound, if $\Occ_k(P,T)=\emptyset$, then the algorithm returns \No with probability at least $1-\frac{1}{n}$.

The algorithm's query complexity is $\Oh(r_Tn + r_Pm)$, optimized for $r_P = \min\Big(1, \sqrt{\frac{2 n\ln n}{km}}\Big)$. %
\end{proof}

Unfortunately, the algorithm in \cref{fct:folklore} examines each candidate position $i\in\fragmentcc{0}{n-m}$ individually, yielding a total running time of $\Omega(n-m+1)$, which is $\Omega(n)$ unless $n=m+o(m)$.
The task of designing a sublinear-time property tester for exact pattern matching remained open until 2020, when Chan, Golan, Kociumaka, Kopelowitz, and Porat~\cite{ChanGKKP20} proved the following:
\begin{theorem}[{\cite[Theorem 10.6]{ChanGKKP20}}]\label{thm:previous}
\cref{pr:tester} can be solved (correctly with high probability) by a non-adaptive algorithm with time complexity 
$\Ohtilde\Big(\sqrt[3]{\tfrac{n^2m}{k}} + \tfrac{n}{k}\Big)$.\lipicsEnd
\end{theorem}
Although truly sublinear for every $k=n^{\Omega(1)}$, the running time in \cref{thm:previous} falls short of the query complexity achieved in \cref{fct:folklore}. 
For example, when $m = \Theta(n)$, the two complexities reduce to $\Ohtilde(n/\sqrt[3]{k})$ and $\Ohtilde(n/\sqrt{k})$, respectively.
This led to the following question motivating our work:
\begin{quote}
    \itshape Is it possible to solve \cref{pr:tester} in time that matches the query complexity of \cref{fct:folklore}?
\end{quote}
Our first main contribution is a positive answer to this question.

\begin{theorem}\label{thm:main}
    \cref{pr:tester} can be solved (correctly with high probability) by a non-adaptive algorithm with time complexity 
    $\Oh\Big(\Big(\sqrt{\tfrac{nm\log^9 n}{k}} + \tfrac{n}{k}\Big)\log^2 n\Big)$.\lipicsEnd
\end{theorem}
With the time and query complexity matching up to logarithmic factors, we ask the following question:
\begin{quote}
    \itshape Is the query complexity of \cref{fct:folklore,thm:main} optimal up to logarithmic factors?
\end{quote}
Our next result proves such optimality, even among \emph{adaptive} algorithms, for $n=m+\Omega(m)$.

\begin{restatable}{theorem}{thmlbadaptive}\label{thm:lb_adaptive}
Fix integer parameters $1 \le k < m \le n$ such that $k > 4\ln(5n)$.
Every algorithm solving \cref{pr:tester} correctly with probability at least $0.9$ makes at least $\tfrac{1}{55}\left(\sqrt{\tfrac{m(n-m+1)}{k}}+\tfrac{n}{k}\right)$ queries to the input strings $P$ and $T$.
If $k \le \frac{m}{4}$, then the lower bound holds for binary strings. \lipicsEnd
\end{restatable}

\noindent
The bounds of \cref{thm:main,thm:lb_adaptive} diverge for $n=m+o(m)$, which yields the following question:
\begin{quote}
    \itshape Can one improve upon the time and query complexity of \cref{thm:main} for $n=m+o(m)$?
\end{quote}
Surprisingly, we match the lower bound of \cref{thm:lb_adaptive} for the whole spectrum of parameters.
In particular, as long as $n = m + \Oh(m/k)$, we achieve $\Ohtilde(m/k)$ running time. This time bound is already needed to distinguish $i\in\Occ(P,T)$ from $i\notin\Occ_k(P,T)$ for a single candidate $i\in\fragmentcc{0}{n-m}$.
\begin{theorem}\label{thm:adaptive-tester}
    \Cref{pr:tester} can be solved (correctly with high probability)  by an adaptive algorithm with time complexity 
        $\Oh\left(\left(\sqrt{\tfrac{m(n-m+1)\log^{29}n}{k}} + \tfrac{n}{k}\right)\log^{14} n\right)$. \lipicsEnd
\end{theorem}

Unfortunately, unlike all previous solutions, the property tester of \cref{thm:adaptive-tester} is adaptive.
\begin{quote}
    \itshape Is adaptivity necessary to improve upon the query complexity of \cref{thm:main} for $n=m+o(m)$?
\end{quote}
It turns out that there is a separation between adaptive and non-adaptive algorithms.
Among non-adaptive solutions, the query complexity of \cref{thm:main} remains optimal already for $n=m+\Omega(k)$, not merely for $n=m+\Omega(m)$.

\begin{restatable}{theorem}{thmlbnonadaptive}\label{thm:lb_nonadaptive}
    Fix integer parameters $1 \le k < m \le n$ such that $k > 4\ln(5n)$.
    Every non-adaptive algorithm solving \cref{pr:tester} correctly with probability at least $0.9$ makes at least $\tfrac{1}{204}\min\left(\frac{n\sqrt{n-m+1}}{k},\sqrt{\frac{nm}{k}}+\tfrac{n}{k}\right)$ queries to the input strings $P$~and~$T$. If $k \le \frac{m}{4}$, then the lower bound holds for binary strings. \lipicsEnd
\end{restatable}

In order to complete the characterization of \cref{pr:tester}, we note that a simplified version of the approach from~\cite{ChanGKKP20} already yields a near-optimal non-adaptive algorithm for $n = m + \Oh(k)$:
\begin{fact}\label{fct:simplified}
    \cref{pr:tester} can be solved (correctly with high probability) by a non-adaptive algorithm with time complexity 
    $\Oh\Big(\tfrac{n\sqrt{n-m+1}}{k} \log^2 n\Big)$.\lipicsEnd
\end{fact}

\begin{figure}
\begin{tikzpicture}[yscale=1.1,xscale=0.8]
    
    \draw[-latex] (0,0) -- (14.5,0) node[below]{$\Delta = n-m+1$};
    \foreach \x/\t in {0/$1$,3/$\tfrac{m}{k}$,6/$k$,9/$m$,12/$mk$} {
        \draw (\x, 0.1) -- (\x, -0.1) node[below] {\t\vphantom{$\frac{m}{k}$}};
    }
    \draw[thick] (14, 1.7) -- (0, 1.7) -- (0, 0.3) -- (14, 0.3);
    \draw[thick, dotted] (15, 1.7) -- (14, 1.7)  (15, 0.3) -- (14, 0.3);
    \draw[thick] (0, 1) -- (9,1) -- (9, 0.3) (12, 0.3) -- (12, 1.7) (3,0.3) -- (3,1) (6,1) -- (6,1.7);
    \draw (1.5, 0.65) node {$\widetilde \Theta\Big(\tfrac{m}{k}\Big)$};
    \draw (6, 0.65) node {$\widetilde \Theta\Big(\sqrt{\tfrac{m\Delta}{k}}\Big)$};
    \draw (3, 1.35) node {$\widetilde \Theta\Big(\frac{m\sqrt{\Delta}}{k}\Big)$};
    \draw (10.5, 1) node {$\widetilde \Theta\Big(\sqrt{\tfrac{nm\vphantom{\Delta}}{k}}\Big)$};
    \draw (13.5, 1) node {$\widetilde \Theta\Big(\frac{n}{k}\Big)$};

    \draw (0, 0.65) node[left] {Adaptive:};
    \draw (0, 1.35) node[left] {Non-Adaptive:};

\end{tikzpicture}
\caption{The optimal time and query complexities of adaptive (below) and non-adaptive (above) property testers solving \cref{pr:tester} for different ranges of $\Delta\coloneqq n-m+1$. (The breakpoints $\frac{m}{k}$ and $k$ drawn on the horizontal axis may switch order depending on whether $k\ge \sqrt{m}$.)}\label{fig:results}
\end{figure}
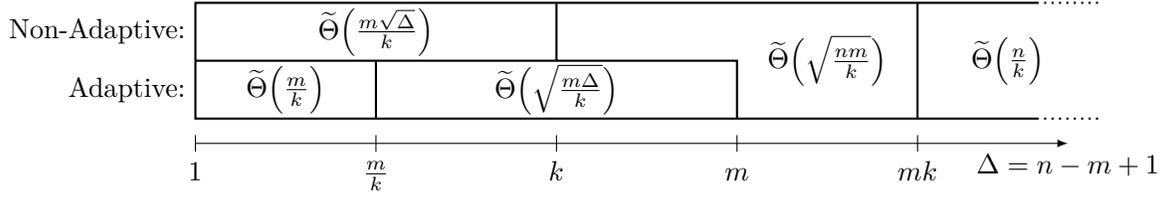

The full characterization of the complexity of \cref{pr:tester} is summarized in \cref{fig:results}.

\paragraph*{Extensions beyond Property Testing}

Having established the time and query complexity of \cref{pr:tester}, we now turn to extensions of this problem previously considered in~\cite{ChanGKKP20}.
The first extension can be interpreted as the tolerant version of \cref{pr:tester}, where the tester must return \Yes not only if $P$ has an exact occurrence in $T$ but also when it has a $k'$-mismatch occurrence for some parameter $k' < k$, typically assumed to be much smaller than $k$.
Setting $k'=0$ recovers the statement of \cref{pr:tester}.

\defproblem{pr:approximation}{Tolerant Property Testing for Pattern Matching}{A pattern $P\in \Sigma^m$, a text $T\in \Sigma^n$, and two integer thresholds $0\le k' < k < m$.}{\Yes if $\Occ_{k'}(P,T)\ne \emptyset$, \No if $\Occ_k(P,T)=\emptyset$, and an arbitrary answer otherwise.}{}

In the second extension, we ask the algorithm to output a set $A$ of candidate occurrences such that $\Occ_{k'}(P,T) \subseteq A \subseteq \Occ_k(P,T)$.
Intuitively, this requires distinguishing $\HD(P,T\fragmentco{i}{{i+m}})\le k'$ and $\HD(P,T\fragmentco{i}{i+m})> k$ for each $i\in \fragmentcc{0}{n-m}$ rather than just for $\min_{i\in \fragmentcc{0}{n-m}}\HD(P,T\fragmentco{i}{i+m})$.

\defproblem{pr:approximationreport}{Tolerant Property Testing for Pattern Matching (reporting version)}{A pattern $P\in \Sigma^m$, a text $T\in \Sigma^n$, and two integer thresholds $0\le k' < k < m$.}{A set $A$ such that $\Occ_{k'}(P,T) \subseteq A \subseteq \Occ_k(P,T)$.}{}

We adjust our non-adaptive algorithms to handle extensions at the price of small overheads in the time and query complexities. 
The multiplicative overhead for the tolerant version is $\Oh(1)$ if $k' = \Oh(k/\log n)$,
and $n^{o(1)}$ if $k' = o(k)$.
\begin{theorem}\label{thm:main-approx}
\Cref{pr:approximation} can be solved (correctly with high probability) by a non-adaptive algorithm with time complexity 
   $\Oh\left(\min\left(\sqrt{\tfrac{nm\log^9 n}{k}}+ \tfrac{n}{k},\tfrac{n\sqrt{n-m+1}}{k}\right)\cdot \log^2 n\right)\cdot n^{\Oh(k'/k)}$.

\Cref{pr:approximationreport} can be solved (correctly with high probability) by a non-adaptive algorithm with time complexity 
    $\Oh\left(\min\left(\sqrt{\tfrac{nm\log^9 n}{k}}+ \tfrac{n}{k},\tfrac{n\sqrt{n-m+1}}{k}\right)\log^2 n+ |\Occ_k(P,T)|\log n\right)\cdot n^{\Oh(k'/k)}$.\lipicsEnd
\end{theorem}

\paragraph*{Open Questions}\label{sec:open}
Our work establishes the near-optimal query and time complexity of \textsf{Property Testing for Pattern Matching} (\cref{pr:tester}).
Nevertheless, the extensions give rise to several interesting open questions.
\begin{itemize}
    \item We are not aware how to generalize our adaptive algorithm (\cref{thm:adaptive-tester}) to the tolerant version (\cref{pr:approximation}) or the reporting version (\cref{pr:approximationreport} with $k'=0$).
    In both cases, we suggest considering $n = m + O(m/k)$ first; in that setting, our algorithm runs in $\Ohtilde(m/k)$ time (see \cref{fig:results}), and many of its steps simplify (see the $\Ohtilde(\Delta + n/k)$-complexity version in \cref{subsec:overview-adaptive}).
    \item In~\cite[Theorem 10.5]{ChanGKKP20}, \cref{pr:approximation} is solved with gap $k/k' = 1+\epsilon$ (for an arbitrarily small constant $\epsilon > 0$) in sublinear running time $\Ohtilde\big(n/k^{\epsilon^{1/3-o(1)}}\big)$.
    Unfortunately, our new techniques increase the gap by a factor of $2$, and thus they offer improvements only for sufficiently large constant gaps. 
    It remains open to improve upon the results of \cite{ChanGKKP20} for smaller gaps.
    \item In general, the time complexity of \cref{pr:approximation} with $k/k' = 1+\epsilon$ is wide open.
    \Cref{fct:folklore} can be easily modified to solve this version at the price of a $\poly(\frac{1}{\epsilon})$-factor overhead in the query complexity, so an improved unconditional lower bound is beyond the scope of the current techniques.
    Still, one can hope for lower bounds conditioned on fine-grained complexity hypotheses.
\end{itemize}
Finally, we remark that the rich structure of \cref{pr:tester} for $n=m+o(m)$ motivates exploring this regime for other approximate pattern matching problems.
To the best of our knowledge, even the time complexity of computing $\Occ_k(P,T)$ (the \textsf{$k$-Mismatch} problem) has not been optimized for $n=m+o(m)$.

\paragraph*{Paper Organization}
In \cref{sec:overview}, we give an overview of our proof techniques.
\cref{sec:prelim} contains useful notations and definitions. 
In \cref{sec:combinatoriallemma}, we state and prove a key combinatorial lemma used by our algorithms.
In \cref{sec:algo}, we present our non-adaptive algorithms, proving \cref{thm:main,fct:simplified,thm:main-approx}.
In \cref{sec:adaptive}, we present our adaptive algorithm, proving \cref{thm:adaptive-tester}.
In \cref{sec:lb_adaptive,sec:lbnonadaptive,sec:lb_large}, we prove our  lower bounds (\cref{thm:lb_adaptive,thm:lb_nonadaptive}).
\section{Technical Overview}
\label{sec:overview}

Henceforth, let $T[0\dd n)\in \Sigma^n$ and $P[0\dd m)\in \Sigma^m$ denote the input text and pattern, respectively.
For $i\in [0\dd n-m]$, let \[M_i \coloneqq  \{j\in [0\dd m): P[j]\neq T[i+j]\}\] 
be the set of mismatching positions between $P$ and $T\fragmentco{i}{i+m}$.
We use $\Ohtilde(\cdot)$, $\Omtilde(\cdot)$, and $\Thtilde(\cdot)$ to hide $\polylog n$ factors.

\subsection{Non-Adaptive Tester}
\label{subsec:overview-nonadap}
Our main non-adaptive property tester (\cref{thm:main}) builds upon the previous work by Chan, Golan, Kociumaka, Kopelowitz, and Porat~\cite{ChanGKKP20}.
In this overview, we first discuss their techniques and limitations.
We then explain how we overcome these limitations, thereby proving \cref{thm:main}.
For simplicity, we focus on the regime $n = \Oh(m)$, in which the desired time complexity in \cref{thm:main} becomes $\Ohtilde\big(\sqrt{nm/k} + n/k\big) = \Ohtilde(n/\sqrt{k})$, whereas the previous algorithm (\cref{thm:previous}) has time complexity $\Ohtilde\big(\sqrt[3]{n^2m/k} +n/k\big) = \Ohtilde(n/\sqrt[3]{k})$.  
The regime of $n = \omega(m)$ uses the same techniques but with appropriately adjusted parameters.

To enable efficient computation, \cite{ChanGKKP20} sampled a highly structured set of positions instead of using i.i.d.\ sampling as in the folklore tester (\cref{fct:folklore}).
The first crucial property is that the sample is $p$-\emph{periodic} for a random prime $p= \Theta(sk\log n)$, where the tunable parameter $s>1$ is specified later.
In other words, the algorithm queries the characters $P[j]$ with $j\bmod p\in R_P$ for some $R_P\subseteq \ZZ_p$ and $T[h]$ with $h\bmod p\in R_T$ for some $R_T\subseteq \ZZ_p$. 
By mimicking the proof of \cref{fct:folklore}, one could pick $R_P$ and $R_T$ uniformly at random with rates $r_P$ and $r_T$, respectively, such that $r_P\cdot r_T = \Theta(\frac{\log n}{k})$.
However, the sample of~\cite{ChanGKKP20} is more structured, derived from a single random subset $B\subseteq \ZZ_p$ sampled at rate $\beta = \Theta(\frac{\log n}{k})$. 
The idea is to pick two subsets $U,V\subseteq \ZZ_p$ of size $\Theta(\sqrt{p})$ such that $\ZZ_p = U+V \coloneqq \{(u+v)\bmod p : u\in U, v\in V\}$ and, based on them, define $R_P = B-U$ and $R_T = B+V$, for a query complexity of $\Oh((|R_P|+|R_T|)\cdot \frac{n}{p}) ={\Oh((|B||U|+|B||V|)\cdot \frac{n}{p})} = \Oh(\beta p\cdot \sqrt{p} \cdot \frac{n}{p})=\Ohtilde(\frac{n\sqrt{p}}{k})= \Ohtilde (\frac{n\sqrt{s}}{\sqrt{k}})$. 
Specifically, the algorithm of~\cite{ChanGKKP20} uses $U = [0\dd  z)$ and $V = z\cdot [0\dd \lceil p/z \rceil)$, where $z=\lceil{\sqrt{p}}\rceil$.
For every candidate position $i \in [0\dd n-m]$, the algorithm picks $u\coloneqq (i\bmod p)\bmod z\in U$ and $v\coloneqq (i\bmod p)-u\in V$ so that $i \equiv u+v \pmod{p}$, and reports \Yes if and only if $T[i+j]=P[j]$ holds for all $j\in [0\dd m)$ such that $j\bmod p \in B-u$ or, equivalently, $(i+j)\bmod p \in B+v$.

\begin{example}[see \cref{fig:cgkkp}]\label{ex:cgkkp}
    Let $p=11$ so that $U=\{0,1,2,3\}$ and $V=\{0,4,8\}$.
    If $B=\{2,8\}$, then the algorithm reads $P[j]$ whenever $j\bmod p \in B-U = \{0,1,2,5,6,7,8,10\}$ and $T[h]$ whenever $h\bmod p \in B+V = \{1,2,5,6,8,10\}$. 
    For $i=17$, we have $u=2$ and $v=4$. The algorithm compares $P[j]$ with $T[i+j]$ for all $j\in [0\dd m)$ such that $j\bmod p \in \{0,6\}$ or, equivalently, $(i+j) \bmod p\in \{1,6\}$.
    \lipicsEnd
\end{example}

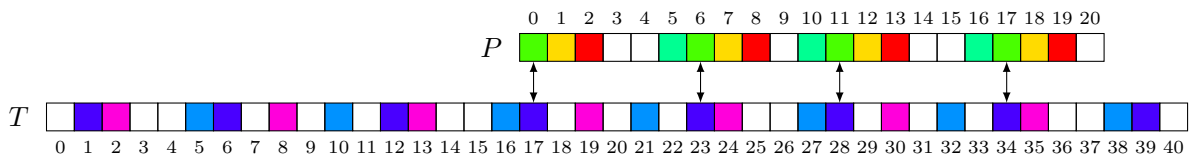
\begin{figure}
    \begin{center}
    \begin{tikzpicture}[scale=0.37]
        \begin{scope}[yshift=2.5cm,xshift=17cm]
        \foreach \x/\y in {0/2,1/1,2/0,5/3,6/2,7/1,8/0,10/3,11/2,12/1,13/0,16/3,17/2,18/1,19/0}{
            \pgfmathsetmacro{\h}{\y/7}    %
            \edef\hval{\h}                 %
            \definecolor{temp}{hsb}{\hval,1,1}
            \fill[temp] (\x,0) rectangle (\x+1,1);
        }
         \foreach \x in {0,1,...,20} {
            \draw (\x,0) rectangle (\x+1,1);
            \draw (\x+.5, 1) node[above]{\scriptsize $\x$};
        }
        \foreach \x in {0,6,11,17} {
            \draw[latex-latex] (\x+.5, 0) -- (\x+.5, -1.5);
        }
        \draw (0,0.5) node[left]{$P\;$};
        \end{scope}
        \foreach \x/\y in {1/1,2/0,5/2,6/1,8/0,10/2,12/1,13/0,16/2,17/1,19/0,21/2,23/1,24/0,27/2,28/1,30/0,32/2,34/1,35/0,38/2,39/1}{
            \pgfmathsetmacro{\h}{6/7-\y/7}    %
            \edef\hval{\h}                 %
            \definecolor{temp}{hsb}{\hval,1,1}
            \fill[temp] (\x,0) rectangle (\x+1,1);
        }
        \foreach \x in {0,1,...,40} {
            \draw (\x,0) rectangle (\x+1,1);
            \draw (\x+.5, 0) node[below]{\scriptsize $\x$};
        }
        \draw (0,0.5) node[left]{$T\;$};
    \end{tikzpicture}
    \end{center}
    \caption{Illustration of the execution of the algorithm of~\cite{ChanGKKP20} for $p = 11$ and $B=\{2,8\}$; see also \cref{ex:cgkkp}.
    The algorithm reads the colorful positions of $P$ and $T$, where each color corresponds to a separate element of the sets $U=\{0,1,2,3\}$ and $V=\{0,4,8\}$.
    When checking the candidate position $i=17$, the algorithm (implicitly) makes comparisons indicated by arrows. 
    In this example, the comparisons are between the green positions in $P$ (corresponding to $u=2$) and the dark blue positions in $T$ (corresponding to $v=4$).
    The subsequences $X_u$ of $P$ and $Y_v$ of $T$ are obtained by concatenating positions (of $P$ or $T$, respectively) of the same color.
    In our example, each position has at most one color; in general, they may have up to $|B|$ colors.}\label{fig:cgkkp}
\end{figure}

Remarkably, \cite{ChanGKKP20} implemented their algorithm to run in \emph{time} $\Ohtilde (\frac{n\sqrt{s}}{\sqrt{k}})$. 
In a nutshell, they compute and compare the Karp--Rabin fingerprints of appropriate subsequences of $P$ and $T$.
We can define these subsequences using the $\bigodot_{r\in R} S_r$ notation, which stands for the concatenation of the strings (or characters) $S_r$ in increasing order of the index $r\in R$.
The algorithm of \cite{ChanGKKP20} defines
\[X_u \coloneqq \bigodot_{j\in [0\dd m)\;:\; j\bmod p \;\in\; B-u}P[j] \quad\text{for }u\in U,\quad\text{and}\quad Y_v \coloneqq \bigodot_{h\in [0\dd n)\;:\; h\bmod p \;\in\; B+v}T[h] \quad\text{for }v\in V.\]
For each $i\in [0\dd n-m]$, the algorithm compares $X_u$ against an appropriate substring of $Y_v$,
where $i\equiv u+v\pmod{p}$.
Crucially, for a fixed $v\in V$, the relevant substrings of $Y_v$ evolve in a sliding-window fashion as $i$ grows from $0$ to $n-m$, and thus their fingerprints can be updated fast.
Furthermore, the interval structure of $U$ allows batched comparisons through range queries and enables a compact representation of the set of \Yes positions $i$; see details in \cref{alg:oneexecution} and the proof of \cref{lem:algo-find-fingerprint-matches}.
To match the complexity $\Ohtilde (\frac{n\sqrt{s}}{\sqrt{k}})$ of this algorithm with the target bound  $\Ohtilde(\frac{n}{\sqrt{k}})$, we would like to \mbox{choose $s = \Ohtilde(1)$}.

On the flip side, the $p$-periodic structure of the sample results in a slightly weaker correctness guarantee for this algorithm: in the \No case, where $|M_i| =\HD(P,T[i\dd i+m))>k$ for all $i\in [0\dd n-m]$, the probability of incorrectly reporting \Yes at position $i$ is $(1-\beta)^{|M_i\bmod p|} \le n^{-\Omega(|M_i\bmod p|/k)}$, which becomes $n^{-\Omega(1)}$ only if the stronger condition $|{M_i \bmod p}|\ge \Omega(k)$ holds as well.
This condition indeed holds for every \textbf{individual} position $i$ such that $|M_i|>k$, with probability at least $1- \frac{1}{s}$ over the random prime $p\in \Theta(sk\log n)$. 
Formally, if the constant hidden in the $\Theta(\cdot)$ notation is appropriate, the following holds by a standard divisor-counting argument with Markov's inequality (see the short proof of \cref{lem:randprime2}):
\begin{equation}
    \label{eqn:keycondition-simple}
\text{If $i\in [0\dd n-m]$ and $|M_i|> k$, then $\Pr_p[|M_i\bmod p| > 0.99k]\ge 1-\tfrac{1}{s}$.}
\end{equation}
However, \eqref{eqn:keycondition-simple} fails to guarantee that $|M_i \bmod p|\ge \Omega(k)$ holds \textbf{simultaneously for all} relevant $i\in[0\dd n-m]$, as the error rate $\frac{1}{s}$ is too high for a union bound over all positions.
There are a few natural attempts to overcome this key issue:
\begin{itemize}
    \item The solution in \cite{ChanGKKP20} is to repeat the aforementioned algorithm $\Oh(\log n)$ times independently. 
    With high probability, every $i$ with $|M_i|>k$ is eliminated at least once, so it is correct to finally return \Yes if some position survives as a \Yes position in all $\Oh(\log n)$ executions. 
    However, determining whether these $\Oh(\log n)$ sets of surviving positions have a nonempty intersection requires time linear in their total size, which is typically $\Omega(\tfrac{n}{s})$ per set given the error rate~$\tfrac{1}{s}$.
    As a result, the total running time is suboptimal: $\Ohtilde \big(\frac{n\sqrt{s}}{\sqrt{k}} + \tfrac{n}{s}\big )$, which becomes $\Ohtilde (n/\sqrt[3]{k})$ \mbox{when $s = \sqrt[3]{k}$.}
    \item One could ask whether a more sophisticated analysis could reduce the error rate below $\frac1s$, which would directly improve \cite{ChanGKKP20}'s time complexity.
    That is,  for a given set $M_i \subseteq [0\dd m)$ of size $|M_i|>k$, we would like a much better upper bound on the number of primes $p\in \Theta(sk\log n)$ such that, say, $|M_i\bmod p|\le 0.1k$.
    This question is related to results from analytical number theory such as the large sieve \cite{bombieri} and larger sieve \cite{largersieve} (see also \cite{helfgott}), which show that a large set cannot occupy too few residue classes modulo many primes. However, the quantitative bounds achieved by these results do not  seem strong enough for our purpose. 
\end{itemize}

At a high level, we address this issue by demonstrating that the bad events for different candidates $i\in \fragmentcc{0}{n-m}$ are highly correlated---a phenomenon reminiscent of the container method in probabilistic combinatorics.
Specifically, we prove the following over the random prime $p\in \Theta(sk\log n)$:%
\begin{equation}
    \label{eqn:keycondition-wishful}
\text{$\Pr_p[\forall_{i\in [0\dd n-m]}\; |M_i|>k \Rightarrow |M_i\bmod p| > 0.49k]\ge 1-\tfrac{\polylog n}{s}$.}
\end{equation}
Crucially, this upper bound applies to the union of bad events across \textbf{all} $i\in [0\dd n-m]$; thus, the total error probability is $\tfrac{\polylog n}{s}$ rather than the much weaker $\frac{n}{s}$ (which results from a naive union bound).
Consequently, we can pick a sufficiently large $s = \polylog n$ and re-use the algorithm of \cite{ChanGKKP20}.
The key difference is that we can now return \Yes already if, for each of the $\Oh(\log n)$ independent iterations, some position survives as a \Yes position. 
Compared to the original setting, the quantifiers are swapped:  the surviving \Yes positions can be different across iterations. 
Thus, although each set of surviving positions is of size $\Omega(\frac{n}{s})=\Omtilde(n)$ on average per iteration, we can halt an iteration as soon as we detect that the set is non-empty.
This yields \cref{thm:main} for $n=\Oh(m)$.

What remains is to prove the purely combinatorial statement of~\eqref{eqn:keycondition-wishful} for an appropriate $s = \polylog n$.
To illustrate our idea, it is instructive to first consider the special case where both the text $T\in \{0,1\}^{n}$ and the pattern $P\in \{0,1\}^m$ are binary strings of Hamming weights $\Ohtilde(k)$.
 Let $M^{(T)}\subseteq [0\dd n)$ and $M^{(P)}\subseteq [0\dd m)$ be the sets of positions containing ones in $T$ and $P$, respectively, where $|M^{(T)}|,|M^{(P)}|= \Ohtilde(k)$.
 Since any mismatch is either between a one from $P$ and a zero from~$T$, or vice versa, the following observation holds: 
\begin{equation}
    \label{eqn:keycondition}
\text{If $i\in [0\dd n-m]$ and $j\in M_i$, then $j\in M^{(P)}$ or $i+j\in M^{(T)}$.}
\end{equation}
By considering which of the two cases in Property~\eqref{eqn:keycondition} occurs more often, we immediately get the following corollary:
\begin{equation}
    \label{eqn:keycondition2}
\text{If $i\in [0\dd n-m]$ and $|M_i|> k$, then $|M_i \cap M^{(P)}|> 0.5k$ or $|M_i \cap (M^{(T)}-i)| > 0.5k$.}
\end{equation}

To simultaneously control the shrinkage of all sets $M_i$ modulo a random prime $p$, we employ a more efficient argument based on Property~\eqref{eqn:keycondition2}, avoiding a naive union bound over all $i\in [0\dd n-m]$.
By the same divisor-counting argument as before, for a sufficiently large random prime $p \in \Theta(\frac1k \cdot (|M^{(T)}|+|M^{(P)}|)^2\cdot \log n) = \Ohtilde(k)$, we have:
\[\Pr_p\Big[|M^{(P)}\bmod p|\ge |M^{(P)}|-0.01k \;\;\text{and}\;\; |M^{(T)}\bmod p|\ge |M^{(T)}|-0.01k\Big] \ge 0.9.\] 
Under these conditions, every subset $M'\subseteq M^{(P)}$ (and every subset $M'\subseteq M^{(T)}$) also satisfies $|M'\bmod p| \ge |M'|-0.01k$. 
In particular, in the first case $|M_i \cap M^{(P)}|> 0.5k$ of Property~\eqref{eqn:keycondition2}, we conclude $|M_i\bmod p|\ge |{(M_i \cap M^{(P)}) \bmod p}| > 0.5k - 0.01k = 0.49k$ (the second case of Property~\eqref{eqn:keycondition2} is similar). In summary, we have established that, with probability at least $0.9$ over a sufficiently large random prime $p=\Thtilde(k)$, we have $|M_i\bmod p|> 0.49k$ simultaneously for all positions $i$ with $|M_i|> k$.
This resolves the key issue in \cite{ChanGKKP20}'s framework discussed earlier and leads to the desired $\Ohtilde(n/\sqrt{k})$-time non-adaptive tester (for the low-Hamming-weight special case under consideration).

For general input instances, which may not be binary with low Hamming weights,
our solution sketched above still works provided that Property~\eqref{eqn:keycondition2} holds for some sets 
$M^{(T)}$ and $M^{(P)}$ of size $\Ohtilde(k)$.
Surprisingly, sets $M^{(T)}$ and $M^{(P)}$ with the desired properties always exist, as we demonstrate in the next subsection.

\subsection{A Novel Characterization of Mismatches}
The following combinatorial lemma is the key technical innovation; it also plays a crucial role for our adaptive tester (see \cref{subsec:overview-adaptive}).

\begin{lemma}[see the full statement in \cref{thm:mismatchcontainer} and an illustration in \cref{fig:container}]
\label{lem:temp-mismatchcontainer}
Let $T\in \Sigma^n$ and $P\in \Sigma^m$, where $n  = \Oh(m)$, and let $k\in [1\dd m]$.
There exist sets $M^{(T)}\subseteq [0\dd n)$ and $M^{(P)}\subseteq [0\dd m)$ of size $\Ohtilde(k)$ such that the following holds for each $i\in [0\dd n-m]$:
\[|\{j\in M_i:  j\in M^{(P)} \text{ or } i+j \in M^{(T)}\}| \ge \min(k, |M_i|).\lipicsEnd\]
\end{lemma}

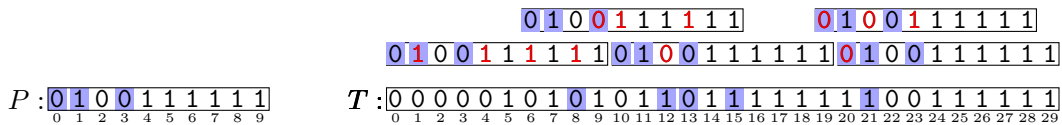
\begin{figure}
    \begin{center}
        \begin{tikzpicture}[scale=.3]
        \begin{scope}
        \foreach \c in {0,1,3} {
         \fill[blue!35] (\c-.4,-.5) rectangle (\c+.4,.5);
        }
        \draw (-.4, -.5) rectangle (9.4,.5);
        \foreach \x [count=\c from 0] in {0,1,0,0,1,1,1,1,1,1}{
            \draw (\c, 0) node{$\mathtt{\x}$};
            \draw (\c, -.2) node[below]{\tiny{\c}};
        }
        \draw (-.1,0) node[left]{$P:$};
        \end{scope}

        \begin{scope}[xshift=15cm]
        \foreach \c in {8,12,13,15,21} {
         \fill[blue!35] (\c-.4,-.5) rectangle (\c+.4,.5);
        }
        \draw (-.4, -.5) rectangle (29.4,.5);
        \foreach \x [count=\c from 0] in {0,0,0,0,0,1,0,1,0,1,0,1,1,0,1,1,1,1,1,1,1,1,0,0,1,1,1,1,1,1}{
            \draw (\c, 0) node{$\mathtt{\x}$};
            \draw (\c, -.2) node[below]{\tiny{\c}};
        }
        \foreach \t/\u in {0/2,6/3.5,10/2,19/3.5,20/2} {
                    \draw (-.4+\t, -.5+\u) rectangle (9.4+\t,.5+\u);
        \foreach \c in {0,1,3} {
            \fill[blue!35] (\c+\t-.4,-.5+\u) rectangle (\c+\t+.4,.5+\u);
            }
                    \draw (-.1,0) node[left]{$T:$};

            \foreach \x [count=\c from 0] in {0,1,0,0,1,1,1,1,1,1}{
            \draw (\c+\t, \u) node{$\mathtt{\x}$};
        }
        }
        \foreach \x/\y/\z in {1/2/1,4/2/1,6/2/1,8/2/1,9/3.5/0,10/3.5/1,13/3.5/1,12/2/0,19/3.5/0,21/3.5/0,23/3.5/1,20/2/0} {
            \draw[red] (\x,\y) node{$\mathtt{\z}$};
        }
        \end{scope}
        \end{tikzpicture}
    \end{center}
    \caption{Illustration of \cref{lem:temp-mismatchcontainer} for sample binary strings $P,T\in\{\mathtt{0},\mathtt{1}\}^{*}$ and $k=2$. In this case, we can set $M^{(P)}=\{0,1,3\}$ and $M^{(T)}=\{8,12,13,15,21\}$; selected positions are highlighted in blue. 
    For each $i\in [0\dd n-m]$, at least $\min(2, |M_i|)$ mismatches involve a position in $M^{(P)}$ or $M^{(T)}$; the figure illustrates this for $i\in \{0,6,10,19,20\}$, with mismatches drawn in red.
    For $i=0$, the four mismatches include $P[1]\ne T[1]$ (with $1\in M^{(P)}$) and $P[8]\ne T[8]$ (with $8\in M^{(T)}$);
    for $i=10$, the only mismatch is $P[2]\ne T[12]$ (with $12\in M^{(T)}$);
    for $i=19$, the three mismatches include $P[0]\ne T[19]$ (with $0\in M^{(P)}$) and $P[2]\ne T[21]$ (with $21\in M^{(T)}$).}\label{fig:container}
\end{figure}

\cref{lem:temp-mismatchcontainer} establishes a $k$-capped version of Property~\eqref{eqn:keycondition}: there exist $\Ohtilde(k)$ positions in $T$ and $P$ which together cover at least $k$ mismatches between $P$ and $T[i\dd i+m)$, for every $i\in [0\dd n-m]$ with $|M_i| > k$.
Note that \cref{lem:temp-mismatchcontainer} still implies Property~\eqref{eqn:keycondition2} as a direct corollary.
Thus, by incorporating \cref{lem:temp-mismatchcontainer} into the discussions in \cref{subsec:overview-nonadap}, we obtain our $\Ohtilde(n/\sqrt{k})$-time non-adaptive property tester in the $n = \Oh(m)$ regime, as claimed in \cref{thm:main}.

We remark that the $k=1$ case of \cref{lem:temp-mismatchcontainer} directly follows from Vishkin's duel method~\cite{Vishkin85}, which has found applications in parallel~\cite{Vishkin85} and quantum~\cite{HariharanV03} algorithms for exact pattern matching.
However, since known proofs for the $k=1$ case do not seem to generalize to large $k$, novel ideas were needed for \cref{lem:temp-mismatchcontainer}. 
Our argument makes interesting use of string synchronizing sets of Kempa and Kociumaka \cite{KempaK19} and heavy-light decompositions in tries.

While sketching the proof of \cref{lem:temp-mismatchcontainer}, we first focus on handling positions $i$ that satisfy $|M_i|\le k$.
In other words, we would like to construct sets $M^{(T)}$ and $M^{(P)}$ of size $\Ohtilde(k)$ such that:
\begin{equation}
    \label{eqn:keycondition-kbounded}
\text{If $i\in [0\dd n-m]$, $|M_i|\le k$, and $j\in M_i$, then $j\in M^{(P)}$ or $i+j\in M^{(T)}$.}
\end{equation}
We then extend this to handle exponentially growing ranges $|M_i|\in \fragmentoc{2^{\ell-1}k}{2^{\ell}k}$.
For this, we construct sets $M^{(P)}_\ell$ and $M^{(T)}_\ell$ of size $\Ohtilde(2^\ell k)$ satisfying Property~\eqref{eqn:keycondition-kbounded} with $2^\ell k$ replacing $k$.
Next, we subsample $\tilde M^{(P)}_\ell\subseteq M^{(P)}_\ell$ and $\tilde M^{(T)}_\ell\subseteq M^{(T)}_\ell$ with rate $\Ohtilde(2^{-\ell})$. 
Finally, we take the unions over all $\ell$ to form $M^{(P)} = \bigcup_\ell \tilde M^{(P)}_\ell$ and $M^{(T)} = \bigcup_\ell \tilde M^{(T)}_\ell$.
By Chernoff bound, the sets $M^{(P)}$ and $M^{(T)}$ satisfy both the ``correctness'' and ``size''  conditions of \cref{lem:temp-mismatchcontainer} with high probability.

To achieve Property~\eqref{eqn:keycondition-kbounded}, the idea is to first isolate the mismatches in $M_i$ modulo a random prime.  (This is analogous to the mod-prime idea from \cite{ChanGKKP20}, but here we are using it in an existential argument via probabilistic method rather than in the actual algorithm.)
Specifically, for $|M_i|\le k$ and each $j\in M_i$, a standard divisor-counting argument shows that, with $\Omega(1)$ probability over a random prime $p\in \Thtilde(k)$, the mismatch $j$ is isolated modulo $p$ from all other mismatches in $M_i$ (that is, $j\not \equiv j'\pmod{p}$ for all $j'\in M_i\setminus \{j\}$). 
Now, for a given prime $p$, we would like to construct sets $M^{(P)}$ and $M^{(T)}$ of size $\Ohtilde(p) = \Ohtilde(k)$ such that, if $j \in M_i$ is isolated modulo $p$, then $j\in M^{(P)}$ or $i+j\in M^{(T)}$. 
Once this is achieved, we take the union over $\Oh(\log n)$ repetitions with independent random primes $p\in \Ohtilde(k)$ so that all  $j\in M_i$ are covered in this way.

Fix a prime $p$ and suppose $j\in M_i$ is isolated from $M_i\setminus \{j\}$ modulo $p$.
Then, the following two equal-length strings, formed by concatenating characters in $P$ and $T$ with spacing $p$,
\[
P_{j\bmod p} \coloneqq \bigodot_{j'\in [0\dd m)\;:\; j'\equiv j\!\!\pmod{p}}P[j'], \qquad\quad
T_{(i+j)\bmod p}^{(i)} \coloneqq 
\bigodot_{j'\in [0\dd m)\;:\; j'\equiv j\!\!\pmod{p}}T[i+j'],
\]
have exactly one mismatch, corresponding to $P[j]\neq T[i+j]$. 
This can also be interpreted as a $1$-mismatch occurrence of $P_{j\bmod p}$ in a string $T_{(i+j)\bmod p}$, where
\[T_{h} \coloneqq \bigodot_{h'\in [0\dd n)\;:\; h'\equiv h\!\!\pmod{p}}T[h'].\]
Thus, our goal is to select $\Ohtilde(p)$ positions within the strings $P_0,\ldots,P_{p-1}$ and $T_0,\ldots,T_{p-1}$ so that,
for each $1$-mismatch occurrence of $P_j$ in $T_h$ (for any $j,h\in \ZZ_p$), at least one of the two mismatching positions is selected.
In order to avoid dealing with multiple strings, we consider their concatenation $S=\bigodot_{j=0}^{p-1} P_j \cdot \bigodot_{h=0}^{p-1} T_h$ (of length $n+m=\Oh(m)$) and, instead, select $\Ohtilde(p)$ positions in $S$ so that the following holds for every $i,i'\in [0\dd |S|-\floor{m/p}]$ (see \cref{lem:onemismatchcontainer}):
\begin{equation}\label{eq:sboth}
\text{If $\HD(S[i\dd i{+}\floor{m/p}),S[i'\dd i'{+}\floor{m/p}))=1$, then one of the mismatching positions is selected.}
\end{equation}
Since a $1$-mismatch occurrence of $P_j$ in $T_h$ corresponds to a pair of length-$\floor{m/p}$ fragments of $S$ with a single mismatch, the modified setting given by \eqref{eq:sboth} is sufficient.
Note that if $\HD(S[i\dd i+\floor{m/p}),\allowbreak S[i'\dd i'+\floor{m/p}))=1$, there are $\Omega(\frac{m}{p})$ matching positions before or after the unique mismatch. 
Focusing on the former case, we can rephrase \eqref{eq:sboth} using the notion of \emph{longest common extensions}
\newcommand{\LCE}{\mathsf{LCE}}
\[\LCE(i,i') \coloneqq \max\{\ell \in \ZZ_{\ge 0} : S[i\dd i+\ell)=S[i'\dd i'+\ell)\}.\]
Namely, we seek a set $M\subseteq [0\dd |S|]$ of size $\Ohtilde(p)$ so that the following holds for every $i,i'\in [0\dd |S|)$:%
\begin{equation}\label{eq:slce}
\text{If $\LCE(i,i')=\Omega(m/p)$, then $i+\LCE(i,i')\in M$ or $i'+\LCE(i,i')\in M$.}
\end{equation}
Our strategy is to proceed in two steps. First, we select a set $\mathsf{S}\subseteq [0\dd |S|)$ of size $\Oh(p)$ so that:
\begin{equation}\label{eq:lces}
\text{If $\LCE(i,i')=\Omega(m/p)$, then $i+\delta\in \mathsf{S}$ and $i'+\delta \in \mathsf{S}$ hold for some $\delta \in [0\dd \LCE(i,i')]$.}
\end{equation}
Next, we select a set $M\subseteq [0\dd |S|]$ of size $\Oh(|\mathsf{S}|\log |\mathsf{S}|)=\Oh(p\log p)$ so that:
\begin{equation}\label{eq:slces}
\text{If $i,i'\in \mathsf{S}$, then $i+\LCE(i,i')\in M$ or $i'+\LCE(i,i')\in M$.}
\end{equation}
These two properties imply \eqref{eq:slce} because $\LCE(i+\delta,i'+\delta)=\LCE(i,i')-\delta$ for every $\delta \in [0\dd \LCE(i,i')]$.

For~\eqref{eq:slces}, we consider a heavy-light decomposition of the trie of all suffixes $S[i\dd |S|)$ with $i\in \mathsf{S}$. 
In the root-to-leaf paths corresponding to $S[i\dd |S|)$ and $S[i'\dd |S|)$, the two nodes representing $S[i+\LCE(i,i')]$ and $S[i'+\LCE(i,i')]$ are siblings, and at least one of them is a light child. 
Hence, the selection rule is to include all positions corresponding to the light nodes on each of the $|\mathsf{S}|$ root-to-leaf paths in the trie; the resulting set $M$ must include one of the required mismatching positions.
Each path has $\Oh(\log|\mathsf{S}|)$ light nodes, so in total we included $\Oh(|\mathsf{S}|\log |\mathsf{S}|)$ positions.

We achieve~\eqref{eq:lces} as a simple application of the $\tau$-synchronizing set \cite{KempaK19} with $\tau=\Theta(\frac{m}{p})$, 
resulting in a set $\mathsf{S}$ of size $|\mathsf{S}|=\Oh(|S|/\tau)=\Oh(p)$.
The only caveat is that the condition \eqref{eq:lces} may fail if $S[i\dd i+\LCE(i,i'))=S[i'\dd i'+\LCE(i,i'))$ is highly periodic with period at most $\frac{\tau}{3}$. 
In that case, however, the period breaks at $S[i+\LCE(i,i')]$ or $S[i'+\LCE(i,i')]$. 
To ensure that \eqref{eq:slce} still holds in this periodic case, we augment $M$ with all positions where a sufficiently short period breaks after running for at least $3\tau-1$ positions.
Periodic structures of this type are well-studied, and the number of breakpoints is known to be $\Oh(|S|/\tau)=\Oh(p)$.
This concludes our proof sketch of \cref{lem:temp-mismatchcontainer}.

\subsection{Extensions of the Non-Adaptive Algorithm}
We now briefly explain how extensions of the algorithm in \cref{subsec:overview-nonadap} allow us to handle the tolerant and reporting variants of our property testing problem (\cref{pr:approximation,pr:approximationreport}), as formalized in \cref{thm:main-approx}.

Recall from \cref{subsec:overview-nonadap} that, assuming the random prime $p$ satisfies the desired conditions
$|M^{(P)}\bmod p|\ge |M^{(P)}|-0.01k$ and
$|M^{(T)}\bmod p|\ge |M^{(T)}|-0.01k$, for any position $i\in [0\dd {n-m}]\setminus \Occ_k(P,T)$, the algorithm reports \Yes with probability $(1-\beta)^{|M_i\bmod p|} \le (1-\beta)^{0.49k}$. 
On the other hand, for a smaller threshold $k'$ and any position $i\in \Occ_{k'}(P,T)$, the algorithm reports \Yes with probability $(1-\beta)^{|M_i\bmod p|} \ge (1-\beta)^{|M_i|} \ge (1-\beta)^{k'}$.
As long as $\frac{k'}{0.49k}<1-\Omega(1)$, there exists a constant probability gap that can be amplified through independent repetitions, enabling reliable distinction between $i\in \Occ_{k'}(P,T)$ and $i\notin \Occ_{k}(P,T)$.%
\footnote{Our method cannot achieve a gap $k/k'<2$ because the transition from \cref{lem:temp-mismatchcontainer} to Property~\eqref{eqn:keycondition2} causes a loss factor of roughly $2$, resulting only in a guarantee of $|M_i\bmod p|> 0.49k$ rather than $|M_i\bmod p|> 0.99k$.
In contrast, \cite{ChanGKKP20} achieves sublinear-time tolerant testing even for gap as small as $k/k'=1+\eps$, using approximate nearest neighbor search; see also the related open question on \cpageref{sec:open}.}

Our actual implementation for \cref{thm:main-approx} is more complicated due to a few technical details
stemming from the non-negligible failure probability of the random prime $p$.
\begin{itemize}
        \item  When $n\in \omega(m)$, we must divide the text into pieces of length $\Theta(m)$ and handle each separately, as a prime $p \in \Thtilde(k)$ may succeed in some pieces while failing in others.
        \item  In the reporting version, the algorithm filters candidate positions generated by one execution of \cref{subsec:overview-nonadap}, reporting only those that appear frequently across all independent repetitions.
        If a prime $p$ fails, it can result in an excessively long list that is infeasible to process.
\end{itemize}

\subsection{Adaptive Tester}\label{subsec:overview-adaptive}
We now overview our adaptive tester, which achieves near-optimal $\Ohtilde(\sqrt{\Delta n/k} + n/k)$ query and time complexity (\cref{thm:adaptive-tester}), focusing on the case when $\Delta := n-m+1 \le 0.1m$; for larger $\Delta$, the non-adaptive tester of \cref{thm:main} suffices.
Notably, our adaptive algorithm uses \cref{lem:temp-mismatchcontainer} for a purpose distinct from that in the non-adaptive setting.
Since our final algorithm is rather complicated, we explain our key techniques in an incremental manner.

\paragraph*{\boldmath Tester with $\Ohtilde(\Delta + n/k)$ Query and Time Complexity}
As a warm-up, we first sketch an adaptive tester with weaker query and time complexity $\Ohtilde(\Delta + \frac{n}{k})$. 
The tester maintains a set $O$ of candidate occurrences, initialized with $O=[0\dd \Delta)$. 
Whenever a mismatch $P[j]\neq T[i+j]$ is observed, certifying that $i\notin \Occ(P,T)$, we eliminate $i$ from $O$. 
The algorithm outputs \No precisely when $O$ becomes empty.
Since exact occurrences are never eliminated, the answer is always correct for \Yes instances. 
For \No instances, where $\Occ_k(P,T) = \emptyset$, the goal is to use few queries to ensure $O$ becomes empty.

We will show shortly that, in the \No case, in $\Ohtilde(\frac{n}{k})$ time and query complexity, with high probability, we can identify $\polylog n$ positions in $P$ and $T$ that collectively eliminate all candidate occurrences $i\in [0\dd\Delta)$ from $O$. (Here, we say that a position $j$ in $P$ eliminates a candidate $i$ if $P[j]\neq T[i+j]$; analogously, a position $h$ in $T$ eliminates $i$ if $T[h]\neq P[h-i]$.)
Once such $\polylog n$ positions are known, we can perform all required eliminations with an extra $\Ohtilde(\Delta)$ query and time complexity by comparing each identified position against $|O|$ positions in the other string.
It remains to show how to identify $\polylog n$ positions achieving such elimination.

In fact, the existence of such $\polylog n$ positions already follows from \cref{lem:temp-mismatchcontainer} applied with parameter $1$ in place of~$k$, but it is very difficult to stumble upon any of these positions.
Hence, we instead apply \cref{lem:temp-mismatchcontainer} with parameter $k$, which shows the existence of size-$\Ohtilde(k)$ sets $M^{(P)}$ and $M^{(T)}$.
Observe that a random position in the disjoint union $M^{(P)}\sqcup M^{(T)}$ eliminates a fixed candidate $i\in [0\dd \Delta)\setminus \Occ_k(P,T)$ with probability at least $\frac{k}{|M^{(P)}|+|M^{(T)}|}\ge 1/\polylog n$.
Hence, drawing $\polylog n$ random samples from $M^{(P)}\sqcup M^{(T)}$ suffices to eliminate all such candidates $i$ with high probability.
Since $|M^{(P)}|+|M^{(T)}| \ge k$, a random sample of $\Ohtilde(n/k)$ positions in $P[0\dd m)$ and $T[0\dd n)$ likely includes at least one random position from $M^{(P)}\sqcup M^{(T)}$. 
If a membership oracle for $M^{(P)}\sqcup M^{(T)}$ existed, we could draw $\polylog n$ samples from $M^{(P)}\sqcup M^{(T)}$ via rejection sampling in $\Ohtilde(n/k)$ query and time complexity, thus obtaining the desired $\polylog n$ eliminating positions.
However, we cannot assume a membership oracle for the unknown set $M^{(P)}\sqcup M^{(T)}$. 

To relax this unrealistic assumption, we modify the aforementioned algorithm as follows. 
Recall that $O\subseteq [0\dd \Delta)$ is the set of remaining candidates.
Instead of trying to sample from $M^{(P)}\sqcup M^{(T)}$, we directly sample a \emph{good position} in $P$ or $T$: one that eliminates at least a $1/\polylog n$ fraction of the candidates in $O$.
The Chernoff bound allows us to approximate the number of eliminated candidates using $\polylog n$ samples drawn from $O$.
Moreover, the reverse Markov inequality shows that $M^{(P)}$ and $M^{(T)}$ collectively contain $\Omega(k)$ good positions.
Hence, by rejection sampling, we can sample a random good position in $\Ohtilde(n/k)$ query and time complexity.
Our algorithm performs $\polylog n$ iterations, where in each iteration we sample a random good position (with respect to the current set of candidates $O$) and use it to eliminate a $1/\polylog n$ fraction of $O$.
In the \No case, with high probability $O$ becomes empty after $\polylog n$ iterations, as desired.
This concludes the description of our warm-up tester with $\Ohtilde(\Delta + \frac{n}{k})$ query and time complexity.

\paragraph*{\boldmath Query-Optimal Tester with $\Ohtilde(\Delta)$ Extra Time}
We now improve the previous tester to achieve near-optimal query complexity $\Ohtilde(\sqrt{\Delta n/k} + n/k)$, though still incurring an additive $\Ohtilde(\Delta)$ term in the time complexity.
As before, we maintain a candidate set $O \subseteq [0\dd \Delta)$ which is iteratively updated as positions are eliminated.
The key idea is to eliminate candidates more efficiently by using the non-adaptive tester from \cref{subsec:overview-nonadap} to operate on \emph{blocks} of positions, thereby avoiding the $\Ohtilde(\Delta)$ queries needed for individual positions.

We partition both text $T$ and pattern $P$ into \emph{blocks} of size~$\Delta$.
Recall that $M^{(P)}$ and $M^{(T)}$ together contain $\Omega(k)$ good positions.
By averaging over the $\Theta(n/\Delta)$ blocks, individual blocks contain $\Omega(k\Delta/n)$ good positions on average.
We call a block \emph{good} if it contains at least $0.1$ times the average number of good positions. Consequently, at least $90\%$ of all good positions belong to good blocks.
Recall that we can sample a random good position with $\Ohtilde(n/k)$ queries and time.
Consequently, with success probability at least $0.9$, we can find a good block with $\Ohtilde(n/k)$ queries and time by sampling a random good position and returning the block it belongs to.

Suppose we have found a good block $P[b\dd b+\Delta)$, assumed in $P$ without loss of generality. 
This block contains $\Omega(k\Delta/n)$ good positions $j$, each satisfying $\Pr_{i\in O}[P[j]\neq T[i+j]] \ge 1/\polylog n$.
Setting $k_0 = \Thtilde(k\Delta/n)$, it follows from linearity of expectation and the reverse Markov inequality that $\Pr_{i\in O}[\HD(P[b\dd b+\Delta), T[i+b\dd i+b+\Delta))>k_0] \ge 1/\polylog n$.
In other words, running the reporting version of our non-adaptive algorithm on $T[b\dd b+2\Delta-1)$ and $P[b\dd b+\Delta)$ with threshold $k_0$ will, with high probability, rule out at least a $1/\polylog n$ fraction of the candidates in~$O$.
The time complexity of this non-adaptive subroutine is $\Ohtilde(\Delta / \sqrt{k_0}) =\Ohtilde(\sqrt{\Delta n/k}) $. 

As before, our algorithm starts with the candidate set $O = [0\dd\Delta)$ and performs $\polylog n$ adaptive iterations. 
In each iteration, we first find a good block using $\Ohtilde(\frac{n}{k})$ query and time complexity. 
Then, we run our non-adaptive algorithm on the $\Theta(\Delta)$-length neighborhood of this good block, as described in the previous paragraph, in $\Ohtilde(\sqrt{\Delta n/k})$ query and time complexity, and obtain an answer set.
We then update $O$ by intersecting it with this answer set.
In the \No case, in each iteration, $|O|$ decreases by a $1-1/\polylog n$ factor with constant probability, so $O$ becomes empty with high probability after $\polylog n$ iterations.
Most steps in this algorithm have time and query complexity within the budget $\Ohtilde(\sqrt{\Delta n/k} + n/k)$, except for computing the narrowed down candidate set~$O$, which takes $\Oh(|O|)\le \Oh(\Delta)$ extra time.

\paragraph*{Achieving Near-optimal Time Complexity}
Now we briefly discuss the ideas in our final algorithm to avoid spending $\Oh(\Delta)$ time iterating over $O\subseteq [0\dd \Delta)$. In the previous algorithm, suppose that the first $d-1$ iterations identified 
the pairs of fragments $(P_1,T_1),\dots,(P_{d-1},T_{d-1})$ for the non-adaptive algorithm to solve.
Then, at the beginning of the $d$-th iteration, the set $O$ is the intersection of answer sets returned by the non-adaptive algorithm run on $\{(P_\ell,T_\ell)\}_{\ell=1}^{d-1}$.
Observe that the algorithm only requires sample access to $O$.

Hence, we can modify the previous algorithm as follows.
To avoid the set intersection, we instead use the answer set $\tilde O$ returned by the non-adaptive algorithm run on a single zipped instance $(P',T')$, where $P'[j]\coloneqq\bigodot_{\ell=1}^{d-1} P_\ell[j]$ is a symbol in an enlarged alphabet $\Sigma^{d-1}$, and $T'$ is defined analogously.
The set $\tilde O$ still excludes all $i\notin \bigcap_{\ell=1}^{d-1} \Occ_{k_0}(P_\ell,T_\ell)$, so it serves the same purpose as taking intersection,  yet it avoids the extra $\Ohtilde (\Delta)$ maintenance time.
The required sampling access to $\tilde O$ is also provided by the succinct representation returned by the non-adaptive algorithm.
However, the downside is that the procedure for sampling from $\tilde O$ now depends on the randomness of the non-adaptive algorithm that generates $\tilde O$.
In other words, instead of sampling a random element from a single set $\tilde O$, we first sample a set from some distribution, and only then a random element of the sampled set.
In particular, the size of $\tilde O$ is a random variable rather than a fixed integer.
Consequently, we no longer can use the size decrement of $\tilde O$ to measure the progress of the algorithm.
Instead, we need to associate the distribution with a potential function and analyze the potential drop in each iteration.

\subsection{Lower Bounds}
We conclude this overview with a brief discussion of our lower bounds.
In all cases, we rely on Yao's principle: we devise a ``\emph{hard}'' input distribution and prove a lower bound on the worst-case number of queries required by a deterministic algorithm to achieve a sufficiently high success probability.
More specifically, we define two input distributions: $D_\Yes$ such that $\Pr_{D_\Yes}[\Occ(P,T)= \emptyset] = 0$ (that is, each sample from $D_\Yes$ contains an exact occurrence), and $D_\No$ such that $\Pr_{D_\No}[\Occ_k(P,T)=\emptyset]$ is close to $1$.
Every solution to \cref{pr:tester} correct on the mixture of these two distributions needs to distinguish them with a sufficiently large success probability.

\paragraph*{\boldmath Warm-Up: Lower Bound for $k=m-1$}
In the case of $k=m-1$, discussed in \cref{sec:lb_large}, we consider strings over an alphabet of sufficiently large size $|\Sigma| = \Theta(n^2)$.
In $D_{\No}$, we independently pick $P$ and $T$ as uniformly random strings in $\Sigma^m$ and~$\Sigma^n$, respectively.
The alphabet size is chosen so that $P\cdot T$ typically has pairwise distinct characters and, in particular, $\HD(P, T\fragmentco{t}{t+m})=m$ typically holds for all $t\in \fragmentcc{0}{n-m}$.
In $D_{\Yes}$, the strings $P$ and $T$ also have uniform marginal distributions, but they are coupled so that $P = T\fragmentco{t}{t+m}$ holds for a uniformly random $t\in \fragmentcc{0}{n-m}$.
In other words, in order to sample $(P,T)\sim D_\Yes$, we pick $t\in \fragmentcc{0}{n-m}$ and
$T\in \Sigma^n$ uniformly at random, and we set $P = T\fragmentco{t}{t+m}$.

For a fixed $t \in \fragmentcc{0}{n-m}$, the algorithm cannot distinguish between $D_\Yes$ and $D_\No$ until it queries both $P[i]$ and $T[i+t]$ for some $i \in \fragmentco{0}{m}$.
A sequence of $q$ queries reveals both $P[i]$ and $T[j]$ for at most $\min(mq, q^2)$ pairs $(i,j)\in\fragmentco{0}{m}\times \fragmentco{0}{n}$,
and thus it may satisfy the aforementioned condition for at most $\min(mq, q^2)$ out of $\Delta$ values $t\in \fragmentcc{0}{n-m}$.
The algorithm needs to be successful on $\Omega(\Delta)$ of values $t\in \fragmentcc{0}{n-m}$, and thus we obtain a lower bound of $\Omega(\max(\sqrt{\Delta}, \frac{\Delta}{m}))$. 
Simple calculations reveal that our adaptive and non-adaptive algorithms achieve this time and query complexity, up to polylogarithmic factors, not only for $k=m-1$ but, in general, as long as $k=\Omtilde(m)$.
Thus, we henceforth consider $k\le \frac{m}{4}$, for which \cref{thm:lb_adaptive,thm:lb_nonadaptive} promise lower bounds for binary strings.

\paragraph*{Lower Bound for Adaptive Algorithms}
The construction behind our binary adaptive lower bound (\cref{thm:lb_adaptive}, \cref{sec:lb_adaptive}) resembles the one described above for $k=m-1$.
However, to derive this bound, the characters of $P$ and $T$ are distributed according to the Bernoulli distribution $\Ber(\frac{2k}{m})$ (where the character is $1$ with probability $\frac{2k}{m}$ and $0$ otherwise), rather than the uniform distribution $\mathcal{U}(\Sigma)$.
In case of $D_\No$, where $P$ and $T$ are independent, we still show that $\Pr_{D_\No}[\Occ_k(P,T)=\emptyset] > 0.8$ (for which we crucially rely on $k \le \frac{m}{4}$).
Using more refined arguments, we show that, for $\Omega(\Delta)$ values of $t \in \fragmentcc{0}{n-m}$, any successful solution to \cref{pr:tester} must not only query the pairs $P[i]$ and $T[i+t]$ but must also encounter a mismatch $P[i] \neq T[i+t]$.
Such a mismatch must involve a 1, which occurs with rate $\Theta(\frac{k}{m})$.
Effectively, the condition $\min(mq, q^2) = \Omega(\Delta)$ is strengthened to $\frac{k}{m}\cdot \min(mq, q^2)  = \Omega(\Delta)$, from which we derive \cref{thm:lb_adaptive}.

\paragraph*{Lower Bound for Non-Adaptive Algorithms}
The construction behind our binary non-adaptive lower bound (\cref{thm:lb_nonadaptive}, \cref{sec:lbnonadaptive}) uses a slightly different approach. 
We fix a parameter $s$ and draw uniformly random values $p\in \fragmentcc{0}{m-s}$ and $t\in \fragmentcc{0}{n-m}$.
We set $P = 0^{p} \cdot S_P \cdot 0^{m-s-p}$ and $T=0^{p+t}\cdot S_T \cdot 0^{n-s-p-t}$, where the characters of $S_P$ and $S_T$ are distributed according to $\Ber(\frac{2k}{s})^s$ (that is, they are binary strings of length $s$ with characters distributed independently according to $\Ber(\frac{2k}{s})$).
In $D_\Yes$, we enforce $S_P=S_T$, whereas in $D_\No$, the strings $S_P$ and $S_T$ are chosen independently.
Intuitively, this construction takes a hard instance $(S_P,S_T)\in \Sigma^{s}\times \Sigma^s$ of the problem of distinguishing $\HD(S_P,S_T)=0$ from $\HD(S_P,S_T)>k$, and then it \emph{hides} this instance at a random, unknown location $(p,t)$ within the full input strings $(P,T)\in \Sigma^m \times \Sigma^n$.
To succeed overall, the algorithm must succeed for a large fraction of locations $(p,t)$, and each such location imposes a constraint on the set of positions that the algorithm reads:
for any successful location $(p,t)$, the algorithm must read $S_P[i] = P[p+i]$ and $S_T[i] = T[p+t+i]$ for $\Omega(\frac{s}{k})$ positions $i\in \fragmentco{0}{s}$.
Crucially, since the algorithm is non-adaptive, the above condition must simultaneously hold across a large fraction of all possible locations $(p,t)$.
We specifically show that, for a large fraction of values $p\in\fragmentcc{0}{m-s}$, the algorithm must read at least $\Omega(\max(\sqrt{s\Delta/k},(s+\Delta)/k))$ characters within $P\fragmentco{p}{p+s}$ and $T\fragmentco{p}{p+s+n-m}$. 
With further calculations based on the appropriate choice of $s=\min(m, \max(4k, \Delta))$, we derive the lower bound of \cref{thm:lb_nonadaptive}.

\section{Preliminaries}
\label{sec:prelim}

\subparagraph*{Generic Notation}
For two real numbers $a \le b$, we use the following integer interval notation: $\fragmentcc{a}{b}\coloneqq \{m\in \ZZ : a \le m \le b\}$, $\fragmentco{a}{b}\coloneqq \{m\in \ZZ : a \le m < b\}$, $\fragmentoc{a}{b}\coloneqq \{m\in \ZZ : a < m \le b\}$, and  $\fragmentoo{a}{b}\coloneqq \{m\in \ZZ : a < m < b\}$.
Typically $a,b\in \ZZ$, and then $\fragmentcc{a}{b} = \fragmentco{a}{b+1}=\fragmentoc{a-1}{b}=\fragmentoo{a-1}{b+1}$.
For $n \in \ZZ_{\ge 0}$, we also write $[n]\coloneqq \fragmentco{0}{n}=\{0,1,\ldots, n-1\}$.

For $A\subseteq \ZZ$ and $b\in \ZZ$, we write $A+b\coloneqq \{a+b: a\in A\}$ and $A-b\coloneqq \{a-b: a\in A\}$.
For a set $A\subseteq \ZZ$ and a modulus $m\in \Zp$, define $A \bmod m \coloneqq \{a\bmod m: a\in A\}$, where $a\bmod m$ is the unique $r\in \fragmentco{0}{m}$ such that $a-r$ is divisible by $m$.

We use $\mathbf{1}[p]$ to denote the \emph{indicator} of a statement $p$, equal to $1$ if $p$ is true, and $0$ if $p$ is false.

\subparagraph*{Strings}
A string $S\in \Sigma^*$ is a finite sequence of characters $S[0]S[1]\cdots S[s-1]$ from an alphabet $\Sigma$; its length is $|S|\coloneqq s$.
In this work, we assume $\Sigma=[0\dd \sigma)$ for some $\sigma\in \Zp$.
Our algorithmic results further require $\sigma \le n^{\Oh(1)}$, where $n$ denotes the input size, so that basic operations on symbols can be performed in constant time on the standard word RAM model with word length $\Theta(\log n)$.

For a \emph{position} $i \in \fragmentco{0}{|S|}$, we
say that $S\position{i}$ is the $i$-th character of $S$.
Given indices $0 \leq i \leq j \leq |S|$, we say that $S\fragmentco{i}{j} \coloneqq S\position{i}\cdots S\position{j-1}$ is a \emph{fragment} of $S$.
The fragment $S\fragmentco{i}{j}$ is called a \emph{prefix} of $S$ if $i=0$ and a \emph{suffix} of $S$ if $j=|S|$.
We may also write $S\fragmentcc{i}{j-1}$, $S\fragmentoc{i-1}{j-1}$,
or $S\fragmentoo{i-1}{j}$ for the fragment $S\fragmentco{i}{j}$.
We say that $S$ is a \emph{substring} of a string $T$ if there is a fragment $T\fragmentco{i}{j}$ of $T$, called an (exact) \emph{occurrence} of $S$ in $T$, such that $S=T\fragmentco{i}{j}$.
Since every occurrence of $S$ in $T$ is uniquely defined by its starting position, we define the set of \emph{exact occurrences} as
\[\Occ(S,T) = \{i\in \fragmentcc{0}{|T|-|S|} : S = T\fragmentco{i}{i+|S|}\}.\]
We remark that the empty string, denoted $\emptystring$, has $|T|+1$ occurrences in $T$ and $\Occ(\emptystring,T)=\fragmentcc{0}{|T|}$.

For two strings $S$ and $S'$ of the same length $|S|=s=|S'|$, we define the set of \emph{mismatches} by
\[\MM(S,S') = \{i\in \fragmentco{0}{s} : S[i] \ne S'[i]\}.\]
The \emph{Hamming distance} of $S$ and $S'$ can then be expressed as $\HD(S,S')=|\MM(S,S')|$.
For a non-negative threshold $k$, we say that $T\fragmentco{i}{j}$ is a \emph{$k$-mismatch occurrence} of $S$ in $T$ if $\HD(S, T\fragmentco{i}{j} )\le k$. 
The set of $k$-mismatch occurrences is uniquely determined by their starting positions, that is, the set 
\[\Occ_k(S,T) = \{i \in \fragmentcc{0}{|T|-|S|} : \HD(S, T\fragmentco{i}{i+|S|}) \le k\}.\]

Let $S\cdot T$ denote the concatenation of two strings $S$ and $T$.
For a finite indexed family $(S_i)_{i\in I}$, let $\bigodot_{i\in I} S_i$ denote concatenation of the strings $S_i$ in increasing order of the running index $i\in I$.

We say that an integer $p \in [1 \dd |S|]$ is a \emph{period} of a string $S\in \Sigma^*$ if $S[i] = S[i + p]$ for all $i \in [0 \dd |S| - p)$; equivalently, $S\fragmentco{0}{|S|-p}=S\fragmentco{p}{|S|}$, that is, the prefix and the suffix of $S$ of length $|S|-p$ \emph{match} (are occurrences of the same string).
The shortest period of a non-empty string $S$ is denoted as $\per(S)$. 
If $\per(S)\le |S|/2$, we say that $S$ is \emph{periodic}.
A \emph{run} in $S$ is a
periodic fragment that cannot be extended (to the left nor to the right) without an increase of its
shortest period.

\subparagraph*{Karp--Rabin Fingerprints}

We use the standard Karp--Rabin fingerprints summarized as follows. 
\begin{lemma}[Fingerprints]
\label{lem:fingerprints}
For prime $q>\sigma$, define the family $\mathcal{F}_q = \{F_{x,q}:x \in \mathbb{Z}_q\setminus \{0\} \}$ where $F_{x,q}\colon \Sigma^* \to \mathbb{Z}_q$ is the function $F_{x,q}(S) = \big ( \sum_{i=0}^{|S|-1} S[i]\cdot x^i\big)\bmod q$.  Then:
\begin{itemize}
    \item $F_{x,q}(S)$ can be evaluated in $O(|S|)$ time.
    \item After preprocessing $x,q,m$ in $O(\log (mq))$ time, the following can be done in constant time:  given  $S[0],S[m]$, and $F_{x,q}(S[0\dd m-1])$, compute $F_{x,q}(S[1\dd m-1])$ and $F_{x,q}(S[0\dd m])$.
    \item For fixed strings $X,Y\in \Sigma^m$ with $X\neq Y$,  we have $\Pr_{F\in \mathcal{F}_q}[F(X)=F(Y)] \le \frac{m-1}{q-1}$. 
\end{itemize}
\end{lemma}
\begin{proof}
   The first property is trivial. The second property can be proved using the identities
   $F_{x,q}(S[0\dd m-1]) - S[0]  \equiv F_{x,q}(S[1\dd m-1]) \cdot x \pmod{q} $
   and 
   $F_{x,q}(S[0\dd m-1]) + S[m]\cdot x^{m}  \equiv F_{x,q}(S[0\dd m]) \pmod{q} $, so after precomputing $x^{m}\bmod q$ and $x^{-1}\bmod q$ in $O(\log (mq))$ time, we can compute $F_{x,q}(S[1\dd m-1])$ and $F_{x,q}(S[0\dd m])$ from the given values in constant time.
   The third property follows from the fact that $F(X)-F(Y)$, viewed as a polynomial in $\mathbb{Z}_q[x]$, is  non-zero and has degree less than $m$, and thus has at most $m-1$ roots.
\end{proof}

\section{Combinatorial Lemma}
\label{sec:combinatoriallemma}
The goal of this section is to prove \cref{thm:mismatchcontainer}.

\begin{theorem}
    \label{thm:mismatchcontainer}
Let $P \in \Sigma^m$, $T\in \Sigma^{n}$, and $k\in [1\dd m]$.
 There exists a set $M\subseteq [0\dd n)$ of size $|M| = \Oh(\frac{n}{m}\cdot k\log^4 m)$ such that the following holds for each $i\in [0\dd n-m]$:
\[|\MM(P,T[i\dd i+m)) \;\cap\; (M\cup (M-i))| \ge \min(k, \HD(P, T[i\dd i+m))).\lipicsEnd\]
\end{theorem}
Note that \cref{lem:temp-mismatchcontainer} directly follows from \cref{thm:mismatchcontainer} by setting $M^{(T)}=M$ and $M^{(P)}=M \cap [0\dd m)$.

We need the following simple lemma (which is similar to \cite[Lemma 3.1]{ChanGKKP20}):
\begin{lemma}
\label{lem:randprime2}
Let $ M\subseteq [0\dd m)$ and $p$ be a random prime from $[\hat p,2\hat p)$.
 Then,
 \begin{enumerate}
     \item For any $d\in M$,
 \[\Pr_{p}\big[\exists_{d'\in M\setminus \{d\}} : d \equiv d' \pmod{p}\big] = \Oh\left(\frac{|M|\log m}{\hat{p}}\right).\]
 \label{item:randprime1}
 \item For any $k>0$, 
 \[ \Pr_{p} \big [|M\bmod p|\le |M| - k \big ]= \Oh\left (\frac{|M|^2\log m}{k\hat p}\right ).\]
 \label{item:randprime2}
 \end{enumerate}
\end{lemma}
\begin{proof}
\begin{enumerate}
    \item 
For each $d'\in M\setminus \{d\}$, we have $\Pr[d \equiv d' \pmod{p}]=\Pr[p \text{ divides } d-d']$. 
Since $0<|d-d'|< m$,  the number $d-d'$ has at most $\log_{\hat p} m$ prime factors in the interval $[\hat p,2\hat p)$.
Moreover, by the prime number theorem, there are $\Omega(\hat p/\log \hat p)$ primes in $[\hat p,2\hat p)$.
Hence, $\Pr[d \equiv d' \pmod{p}] \le \frac{\log_{\hat p} m}{\Omega(\hat p/\log \hat p)} = \Oh(\frac{\log  m}{\hat p})$. 
The lemma follows by the union bound over $|M|-1$ choices for $d'$.
\item Observe that $0\le |M| - |M\bmod p| \le \sum_{d\in M} \mathbf{1}[\exists_{d'\in M\setminus \{d\}} : d \equiv d' \pmod{p}]$. Then, by \cref{item:randprime1} and linearity of expectation, we have $\mathbf{E}_{p}[|M| - |M\bmod p|] = \Oh\left (\frac{|M|^2\log m}{\hat p}\right )$. The claim then follows from Markov's inequality. \qedhere
\end{enumerate}
\end{proof}

We need the string synchronizing sets introduced by \cite{KempaK19}. The following statement follows from \cite[Definition 3.1 and Proposition 8.10]{KempaK19}.
\begin{proposition}[$\tau$-synchronizing sets \cite{KempaK19}]
\label{lem:syncset}
Let $T\in \Sigma^n$ and let $\tau \le \frac{1}{2}n$ be a positive integer. 
There exists a set $\mathsf{S} \subseteq [0 \dd  n - 2\tau]$ of size at most $\frac{30n}{\tau}$ satisfying the following conditions:
\begin{enumerate}
    \item  if $T[i\dd  i + 2\tau) = T [j \dd j + 2\tau)$, then $i \in \mathsf{S}$ holds if and only if $j \in \mathsf{S}$ (for $i, j \in [0 \dd n - 2\tau]$), and
    \item $\mathsf{S} \cap [i \dd i + \tau) = \emptyset$ if and only if $i \in \mathsf{R}$ (for $i \in [0\dd n - 3\tau + 1]$), where 
    \[\mathsf{R} = \{i \in [0 \dd n - 3\tau + 1] : \per(T [i \dd i + 3\tau - 1) )\le \tfrac{1}{3}\tau\}.\lipicsEnd\]
\end{enumerate}
\end{proposition}

\begin{lemma}\label{lem:onemismatchcontainer}
Let $S\in \Sigma^n$ and $m\in [1\dd n]$. There is a set $M \subseteq [0\dd n)$ of size  $\Oh(\frac{n}{m}\log \frac{2n}{m})$ such that
\[\MM(S[i\dd i+m), S[j\dd j+m)) \subseteq (M-i)\cup (M-j)\]
holds for all $i,j\in [0\dd n-m]$ satisfying $\HD(S[i\dd i+m), S[j\dd j+m))=1$.
\end{lemma}
\newcommand{\SSS}{\mathsf{S}}
\begin{proof}
If $m = \Oh(1)$, then $M = [0\dd n)$ clearly satisfies the requirement and has size $|M|= n \in \Oh(\frac{n}{m}\log \frac{2n}{m})$. Thus, we assume $m\ge 6$. Let $\SSS$ be a $\tau$-synchronizing set $\SSS$ of $S$ for $\tau = \lfloor \frac{m}{6}\rfloor $ from \cref{lem:syncset}, with $|\SSS|\le \frac{30n}{\tau} = \Oh(\frac{n}{m})$.
Our construction of $M$ consists of the following three steps:
\begin{enumerate}
    \item
    Define a \emph{$\tau$-run} as a run of length at least $3\tau-1$ with shortest period at most $\frac{\tau}{3}$. 
    For each $\tau$-run $S[\ell\dd r)$, include $\ell-1$ and $r$ in $M$.
    
    Since every two different $\tau$-runs have overlapping length less than $\frac{\tau}{3}$, the total number of $\tau$-runs in $S$ is at most $\frac{|S|}{(3\tau -1 ) - \frac{\tau}{3}} = \Oh(\frac{n}{m})$. Thus, we have added $\Oh(\frac{n}{m})$ elements to $M$.
\item Consider the trie of suffixes $S[i\dd n)$ for $i\in \SSS$.\footnote{The \emph{trie} of a finite set $\mathcal{X}\subseteq \Sigma^*$ of strings is a rooted tree whose nodes are in one-to-one correspondence with the set of prefixes $\{X[0\dd i): X\in \mathcal{X}, i\in [0\dd |X|]\}$, where the root is labeled by the empty string $\emptystring$, and a node labeled by a non-empty string $T$ is a child of the node labeled by $T[0\dd |T|-1)$.}
We say an edge in the trie connecting parent $a$ and child $b$ is \emph{light} if $\mathrm{size}(a)\ge 2\cdot \mathrm{size}(b)$, where $\mathrm{size}(u)$ denotes the number of leaf nodes in the subtree rooted at $u$. 
For each such suffix $S[i\dd n)$, and for each light edge on the corresponding root-to-leaf path (clearly, there can be at most $\log_2|\SSS|$ many), connecting a node with label $S[i\dd i+d)$ with a node with label $S[i\dd i+d]$, add $i+d$ to $M$.

In total, we added at most $|\SSS|\cdot \log_2|\SSS| = \Oh(\frac{n}{m}\log\frac{2n}{m})$ elements to $M$.
\item  Symmetrically, consider the trie of reverse prefixes $\overline{S[0\dd i)} = S[i-1]\cdots S[1]S[0]$ for $i\in \SSS$ and add to $M$ the positions corresponding to light edges on each root-to-leaf path.
\end{enumerate}

Now, consider $i,j\in [0\dd n-m]$ satisfying $\HD(S[i\dd i+m), S[j\dd j+m))=1$.
Let $p$ be the only element of $\MM(S[i\dd i+m), S[j\dd j+m))$, so $S[i+p]\neq S[j+p]$.
By symmetry of our construction of $M$, assume $p\ge \lfloor \frac{m}{2}\rfloor \ge 3\tau - 1$ without loss of generality.
Let $p' = 3\tau - 1\le  p$, and $i'=i+p-p'\ge i$, and $j'=j+p-p' \ge j$.
Consider two cases:
\begin{itemize}
    \item  If $S[i'\dd i+p)=S[j'\dd j+p)$ is periodic with period at most $\frac\tau3$, then both fragments are part of some $\tau$-runs, and the period breaks at position $i+p$ or $j+p$, and thus one of these positions belongs to $M$.
   \item   Otherwise, by the second condition of \cref{lem:syncset}, $i'\notin \mathsf{R}$, and there is a shift $\delta \in [0\dd \tau)$ such that $i'+\delta \in \SSS$. Then, by the first condition of \cref{lem:syncset}, we have $j'+\delta \in \SSS$ as well. The root-to-leaf paths corresponding to the suffixes $S[i'+\delta\dd n)$ and $S[j'+\delta\dd n)$ diverge at depth $p'-\delta$, and one of these paths must use a light edge at that depth.
We then have $i'+p'=i+p\in M$ or $j'+p'=j+p\in M$ depending on which path uses such an edge.
\end{itemize}
In both cases, we have concluded that $i+p\in M$ or $j+p\in M$, which is equivalent to the claimed statement $\{p\}=\MM(S[i\dd i+m), S[j\dd j+m)) \subseteq (M-i)\cup (M-j)$.
\end{proof}

\begin{lemma}\label{lem:kmismatchcontainer}
Let $S\in \Sigma^{n}$, $m\in [1\dd n]$, and $k\in \mathbb{Z}_{+}$.
There exists a randomized set (i.e., a distribution of sets) $M\subseteq [0\dd n)$ of worst-case size $\Oh(\frac{n}{m}\cdot k\log^2 n)$ such that, for each $i,j\in [0\dd n-m]$ satisfying $\HD(S[i\dd i+m), S[j\dd j+m)) \le k$ and each $d\in \MM(S[i\dd i+m), S[j\dd j+m))$, we have 
\[\Pr_M[d \in (M-i)\cup (M-j)] \ge 0.99.\]
\end{lemma}
\begin{proof}
If $k = \Omega(m/\log^2 n)$, then we simply pick $M=[0\dd n)$; this set clearly satisfies the claimed property due to $[0\dd m)\subseteq (M-i)\cup (M-j)$. 
Thus, we henceforth assume $k = o(m/\log^2 n)$.

Let $p$ be a random prime sampled from the interval $[\hat p,2\hat p)$, where $\hat p = \Theta(k\log n)$ is sufficiently large. Since $k = o(m/\log^2 n)$,  we can assume $m\ge 2\hat p>p$.
Define $\hat{S}$ as a permutation of $S$ with characters $S[j]$ ordered according to $(j \bmod p,j)$, that is,
\[\hat{S} = \bigodot_{s\in [0\dd p)} S_{s},\qquad\text{ where }\qquad S_{s} = \bigodot_{j\in [0\dd n) : j \bmod p\;=\; s} S[j].\]
Denote this permutation by $\pi\colon [n]\to [n]$, which satisfies $\hat S[\pi(i)] = S[i]$.
We apply \cref{lem:onemismatchcontainer} for $\hat{S}$ and $\hat{m}\coloneqq \lfloor m/p \rfloor \ge 1$, and we permute the obtained set $\hat{M}$ to derive $M\coloneqq \{i\in [n]: \pi(i)\in \hat M\}$.
Note that $|M|=|\hat{M}| = \Oh(\frac{n}{\hat{m}} \log \frac{2n}{\hat{m}})=\Oh(\frac{n}{m}\cdot p \log n)=\Oh(\frac{n}{m}\cdot k\log^2 n)$.

For $i,j\in [0\dd n-m]$, denote $D_{i,j}\coloneqq \MM(S[i\dd i+m), S[j\dd j+m))$.
Let us pick $i,j\in [0\dd n-m]$ such that $|D_{i,j}| \le k$ and a position $d\in D_{i,j}$.
By \cref{lem:randprime2} (\cref{item:randprime1}), a random prime $p\in [\hat p, 2\hat p)$ fails to isolate $d$ from other elements of $D_{i,j}$ with probability $\Oh((|D_{i,j}| \log m)/\hat p)$.
As long as $\hat p = \Theta(k\log n)$ is sufficiently large, we can bound this probability by $0.01$, that is, $\Pr_p[\exists_{d'\in D_{i,j}\setminus\{d\}} d\equiv d' \pmod{p}]\le 0.01$.
Equivalently, with probability at least $0.99$, every $d'\in D_{i,j}\setminus\{d\}$ satisfies
$d \not\equiv d' \pmod{p}$.
We henceforth assume that this event holds.
Then, the fragments of $S_{(i+d)\bmod p}$ and $S_{(j+d)\bmod p}$ of $\hat{S}$ corresponding to $S[i\dd i+m)$ and $S[j\dd j+m)$, respectively, are at Hamming distance one.
The length of these two fragments is at least $\hat{m}=\floor{m/p}$. 
Thus, by \cref{lem:onemismatchcontainer}, the set $\hat{M}$ contains the position of $\hat{S}$ corresponding to $S[i+d]$ or the position of $\hat{S}$ corresponding to $S[j+d]$; in other words, $\pi(i+d)\in \hat M$ or $\pi(j+d)\in \hat M$ holds.
By construction, this means that $i+d\in M$ or $j+d\in M$, i.e., $d\in (M-i)\cup (M-j)$.
We assumed an event with a probability of at least $0.99$, so we conclude that $\Pr[d\in (M-i)\cup (M-j)] \ge 0.99$ holds as claimed.
\end{proof}

\begin{corollary}\label{cor:kmismatchcontainer_below}
Let $S\in \Sigma^{n}$, $m\in [1\dd n]$, and $k\in \mathbb{Z}_{+}$.
There exists a set $M\subseteq [0\dd n)$ of size $\Oh(\frac{n}{m}\cdot k\log^3 n)$ satisfying the following property for every $i,j\in [0\dd n-m]$ such that $\HD(S[i\dd i+m), S[j\dd j+m))\le k$:
\[\MM(S[i\dd i+m), S[j\dd j+m)) \subseteq  (M-i)\cup (M-j).\]
\end{corollary}
\begin{proof}
 Let us independently sample $\ell \coloneqq \lceil 10\log n\rceil$ sets $(M_t)_{t=1}^\ell$ according to the distribution of \cref{lem:kmismatchcontainer} and pick $M=\bigcup_{t=1}^\ell M_t$.
 Observe that $|M|\le \ell \cdot \Oh(\frac{n}{m}\cdot k\log^2 n)=\Oh(\frac{n}{m}\cdot k\log^3 n)$.
 To complete the proof, it suffices to show that $M$ satisfies the claimed property with non-zero probability.
 For this, let us pick $i,j\in [0\dd n-m]$, and denote $D_{i,j}\coloneqq \MM(S[i\dd i+m), S[j\dd j+m))$.
 If $|D_{i,j}|\le k$, then 
 \cref{lem:kmismatchcontainer} guarantees $\Pr_{M_t}[d \notin (M_t-i)\cup (M_t-j)] \le 0.01$ for every $d\in D_{i,j}$.
 Since the sets $M_t$ are sampled independently of each other, we conclude that
 $\Pr[d \notin (M-i)\cup (M-j)] \le (0.01)^{\ell} = o(n^{-3})$
 and $\Pr[D_{i,j}\subseteq (M-i)\cup (M-j)] \ge 1-o(n^{-2})$.
 We finish by a union bound over all pairs $i,j$.
\end{proof}

\begin{corollary}\label{cor:kmismatchcontainer_above}
Let $S\in \Sigma^{n}$, $m\in [1\dd n]$ and $k\in \mathbb{Z}_{+}$.
There exists a set $M\subseteq [0\dd n)$ of size $\Oh(\frac{n}{m}\cdot k \log^4 n)$ satisfying the following property for every $i,j\in [0\dd n-m]$:
\[|\MM(S[i\dd i+m), S[j\dd j+m)) \;\cap\; ((M-i)\cup (M-j))| \ge \min(k,\HD(S[i\dd i+m), S[j\dd j+m))).\]
\end{corollary}
\begin{proof}
Assume $k\ge 20$; otherwise, setting $k$ to $20$ only affects the constant factor in the upper bound for $|M|$.

To construct $M$, first let $M_{4k}$ be the set of \cref{cor:kmismatchcontainer_below} for threshold $4k$.
Next, for each threshold $K\in \{8k,16k,32k,\dots\} \cap [2m]$, we independently sample $\ell \coloneqq \ceil{10\log n}$ sets $(M_{K,t})_{t=1}^\ell$ according to the distribution of \cref{lem:kmismatchcontainer} with threshold $K$, and we include each $s\in M_{K,t}$ in $\tilde M_{K,t}$ independently with probability $\frac{8k}{K}$. Define $M = M_{4k} \cup \bigcup_{K,t}\tilde M_{K,t}$. It suffices to show that $M$ satisfies the claimed property with non-zero probability.

Each $\tilde M_{K,t}$ has expected size $\frac{8k}{K}\cdot \Oh(\frac{n}{m}\cdot K\log^2 n) = \Oh(\frac{n}{m}\cdot k\log^2 n)$.
Hence, the size $|M|$ is at most $|M_{4k}| +\sum_{K,t}|\tilde M_{K,t}| = \Oh(\frac{n}{m}\cdot k\log^3 n)+  \log \left (\frac{m}{k}\right ) \cdot \ell\cdot\Oh(\frac{n}{m}\cdot k\log^2 n) = \Oh(\frac{n}{m}\cdot k\log^4 n)$ in expectation and, by Markov's inequality, with at least $0.99$ probability.

Now, consider $i,j\in [0\dd n-m]$ and denote $D_{i,j}\coloneqq \MM(S[i\dd i+m), S[j\dd j+m))$.
If $|D_{i,j}|\le 4k$, then \cref{cor:kmismatchcontainer_below} guarantees $D_{i,j}\subseteq (M_{4k}-i)\cup (M_{4k}-j)$, so $|D_{i,j}\;\cap\; ((M-i)\cup (M-j))|= |D_{i,j}|$ as desired.
In the remaining case, we have $|D_{i,j}|\in (K/2,K]$ for some $K\in \{8k,16k,32k,\dots\} \cap [2m]$.
\Cref{lem:kmismatchcontainer} guarantees that, for each $t\in [1\dd \ell]$, the expected number of elements $d\in D_{i,j}$ that are \emph{not} included in $(M_{K,t}-i)\cup (M_{K,t}-j)$ does not exceed $0.01|D_{i,j}|\le 0.01K$.
By Markov's inequality, with probability at least $0.9$, at most $0.1K$ elements $d\in D_{i,j}$ are not included in $(M_{K,t}-i)\cup (M_{K,t}-j)$,
and thus at least $|D_{i,j}|-0.1K > 0.4K$ many elements of $D_{i,j}$ are included in $(M_{K,t}-i)\cup (M_{K,t}-j)$.
This means $|D_{i,j}\cap (M_{K,t} - i)|\ge 0.2K$ or $|D_{i,j}\cap (M_{K,t} - j)|\ge 0.2K$ holds. Assuming the former holds, since $(\tilde M_{K,t}-i)$ is sampled from $(M_{K,t}-i)$ at rate $\frac{8k}{K}$,  by Chebyshev's inequality, $|D_{i,j}\cap (\tilde M_{K,t} - i)|\ge 0.2K \cdot \frac{8k}{K}  - 2 \sqrt{0.2K \cdot \frac{8k}{K}}\ge k$ holds with at least $3/4$ probability (here we used the assumption that $k\ge 20$).  In the latter case we similarly have $|D_{i,j}\cap (\tilde M_{K,t} - j)| \ge k$ with at least $3/4$ probability.
Hence, by a union bound, each $t\in [1\dd \ell]$ succeeds for $D_{i,j}$, namely satisfying $|D_{i,j}\cap ((\tilde M_{K,t} - i) \cup (\tilde M_{K,t} - j))|\ge k$, with $\ge 0.65$ probability. Then, with $1-0.35^{\ell} \ge 1-1/n^{10}$ probability, at least one of $t\in [1\dd\ell]$ succeeds for $D_{i,j}$, in which case we have $|D_{i,j}\cap ((M - i) \cup (M - j))|\ge k$ as desired. We then finish by a union bound over all pairs $i,j$.
\end{proof}

Now, the main theorem immediately follows from \cref{cor:kmismatchcontainer_above}.
\begin{proof}[Proof of \cref{thm:mismatchcontainer}]
Apply \cref{cor:kmismatchcontainer_above} for the concatenated string $S=P\cdot T$ with parameters $m=|P|$ and $k$, and obtain a set $M'\subseteq [0\dd n+m)$ of size $|M'|= \Oh(\frac{|S|}{m}\cdot k\log^4 |S|) = \Oh(\frac{n}{m}\cdot k\log^4 n)$. Transform $M'$ to the set $M\subseteq [0\dd n)$ defined as $M = \big (M' \cup (M' - m)\big ) \cap [0\dd n)$, which still has size $|M| \le 2|M'| = \Oh(\frac{n}{m}\cdot k\log^4 n)$.

It remains to verify that $M$ satisfies the claimed property in \cref{thm:mismatchcontainer} for every $i\in [0\dd n-m]$. Let $D_i = \MM(P,T[i\dd i+m)) = \MM(S[0\dd m), S[i+m\dd i+2m))$. By \cref{cor:kmismatchcontainer_above},  $|D_i \;\cap\; (M'\cup (M'-m-i))| \ge \min(k,|D_i|)$.  Consider any $d\in D_i$. If $d \in M'$, then $d\in M'\cap [0\dd n) \subseteq M$. If $d\in (M'-m-i)$, then $i+d \in (M'-m)\cap [0\dd n) \subseteq M$, and thus $d\in M-i$. Therefore, $D_i \;\cap\; (M'\cup (M'-m-i)) \subseteq M\cup (M-i)$. Then we conclude that $|D_i \cap (M\cup (M-i))| \ge |D_i \cap (M'\cup (M'-m-i))| \ge \min(k,|D_i|)$ as desired.
\end{proof}

\section{Non-Adaptive Algorithm}
\label{sec:algo}
In this section, we present our non-adaptive algorithm for \textsf{Tolerant Property Testing for Pattern Matching}
(\cref{pr:approximation}) and its reporting version
(\cref{pr:approximationreport}), thereby proving \cref{thm:main-approx}. The property testers claimed in \cref{thm:main} and \cref{fct:simplified} immediately follow from \cref{thm:main-approx} by setting $k'=0$.

Our algorithm runs multiple executions of the (slightly adapted) generic algorithm from \cite{ChanGKKP20} and finally performs a post-processing on the set of candidate answers returned by these executions. 
\cref{alg:oneexecution} specifies the behavior of one execution of the generic algorithm, which takes two extra integer parameters $\hat p \ge k$ and $z\in \fragmentcc{1}{\min(\hat{p},\Delta)}$ whose values will be set later.
Naively running the pseudocode in \cref{alg:oneexecution} would be too slow, but later in \cref{lem:algo-find-fingerprint-matches} it is shown that \cref{alg:oneexecution} can be implemented efficiently.

\begin{algorithm}[h]
  \DontPrintSemicolon
\SetKwInput{KwInput}{Input}
\SetKwInput{KwOutput}{Output}
\SetKwComment{Comment}{$\triangleright$\ }{}
\KwInput{$P\in \Sigma^m$, $T\in \Sigma^n$, $k \in \fragmentoc{0}{m}$, $\hat p\ge k$, and $z \in \fragmentoc{0}{\min(\hat p,\Delta)}$, where $\Delta=n-m+1$.}
\KwOutput{A set $A \subseteq \fragmentco{0}{\Delta}$ of candidate answers.}
Denote $z'\coloneqq \lceil\min(2\hat p,\Delta)/z\rceil$\\
Pick a random prime $p\in [\hat p,2\hat p]$ \label{line:samplep}\\
Pick a random subset $B\subseteq \ZZ_p$ with sampling rate $\beta \coloneqq 1-n^{-4/k}$\\
Pick a random fingerprint function $F\in {\cal F}_M$ from \cref{lem:fingerprints}, where $M\in [n^{10},2n^{10}]$ is prime\label{line:samplef}\\
\lForEach{$u\in \fragmentco{0}{z}$}{$\displaystyle X_u\: \coloneqq\: \bigodot_{j\in \fragmentco{0}{m}:\ (j+u) \bmod p\;\in\; B} P[j]$}\label{line:defnXu}
\lForEach{$v\in \fragmentco{0}{z'}$ \KwSty{and} $i\in \fragmentco{0}{\Delta}$}{
$\displaystyle Y_v(i) \:\coloneqq\: \bigodot_{j\in \fragmentco{0}{m}:\ (i+j-vz) \bmod p \;\in\; B } T[i+j]$}\label{line:defnYv}
\lForEach{$i\in \fragmentco{0}{\Delta}$}{write $(i\bmod p)$ as $u_i+v_iz$ with $u_i\in \fragmentco{0}{z}$ and $v_i\in \fragmentco{0}{z'}$}\label{line:defnuivi}
\Return{$A\coloneqq \{i \in \fragmentco{0}{\Delta} : F(X_{u_i}) = F(Y_{v_i}(i)) \}$}\label{line:defnA}
\caption{One execution of \cite{ChanGKKP20}'s generic algorithm} \label{alg:oneexecution}
\end{algorithm}

Recall that we assume a polynomially bounded alphabet $\Sigma$.
For concreteness, assume $\sigma=|\Sigma| < n^{10}$ (otherwise, one should appropriately increase the exponent at Line~\ref{line:samplef}).

\cite{ChanGKKP20} showed the following key property of \cref{alg:oneexecution}:
\begin{lemma}[\cite{ChanGKKP20}]
\label{lem:ineqprob}
In \cref{alg:oneexecution}, after $p$ is chosen, for every $i\in \fragmentco{0}{\Delta}$ it holds that
\[ \Pr_{B}[X_{u_i}\neq Y_{v_i}(i)\mid p] = 1 - (1-\beta)^{|\MM(P,T\fragmentco{i}{i+m})\bmod p|}.\]
\end{lemma}
\begin{proof}
By the definitions in Lines~\ref{line:defnYv}--\ref{line:defnuivi}, we have
\[Y_{v_i}(i) = \bigodot_{j\in \fragmentco{0}{m}:\ (i+j-v_iz) \bmod p \;\in\; B } T[i+j] = \bigodot_{j\in \fragmentco{0}{m}:\ (j+u_i) \bmod p \;\in\; B } T[i+j].\]
By Line~\ref{line:defnXu}, 
\[X_{u_i} = \bigodot_{j\in \fragmentco{0}{m}:\ (j+u_i) \bmod p \;\in\; B } P[j].\]
Comparing these two definitions, we see that $X_{u_i}\neq Y_{v_i}(i)$ holds if and only if
there exists $j\in \MM(P,T\fragmentco{i}{i+m})$ such that $(j+u_i)\bmod p \in B$. 
This is equivalent to the event that $\MM(P,T\fragmentco{i}{i+m}) \bmod p$ has non-empty intersection with $(B-u_i) \bmod p$.
Since $B\subseteq \ZZ_p$ is a random sample with rate $\beta$, the claim immediately follows.
\end{proof}

We would like to use \cref{lem:ineqprob} to eliminate positions $i\notin \Occ_k(P,T)$. 
For this purpose, we would like $|\MM(P,T\fragmentco{i}{i+m})\bmod p|$ to be large compared to $\HD(P,T\fragmentco{i}{i+m})$. 
We achieve this using the following lemma, which is a corollary of the main combinatorial lemma (\cref{thm:mismatchcontainer}):
\begin{lemma}
\label{cor:primegoodforwindow}
Let $P \in \Sigma^m$ and $T\in \Sigma^{n}$ with $m\le n\le 2m$, and $k\in \ZZ_+$.

For a uniformly random prime $p \in [\hat p,2\hat p]$, the following holds:
\[\Pr_p\left[\{i\in \fragmentco{0}{\Delta}: |\MM(P,T\fragmentco{i}{i+m}) \bmod p| < 0.49k\} \subseteq \Occ_{k}(P,T)\right] \ge 1- \Oh\left(\tfrac{k\log^9 n}{\hat p}\right).\]
\end{lemma}
\begin{proof}
Apply \cref{thm:mismatchcontainer} to $P$, $T$, and $k$, and obtain a set $M$ of size $|M|=\Oh(k\log^4 n)$.  
By \cref{lem:randprime2} (\cref{item:randprime2}), $|M\bmod p| \ge |M| - 0.01 k$ holds with probability at least $1 - \Oh(\frac{|M|^2\log n}{0.01 k\hat p})=1 - \Oh(\frac{k\log^9 n}{\hat p})$.
Now assume this event holds,  and thus there exists a subset $M'\subseteq M$ of size $|M'| \le 0.01 k$ such that $M\setminus M'$ have distinct remainders modulo $p$.

For each $i\in \fragmentco{0}{\Delta} \setminus \Occ_{k}(P,T)$, by \cref{thm:mismatchcontainer} we have 
\begin{align*}
k  &\le |\MM(P,T\fragmentco{i}{i+m}) \cap (M\cup (M-i))|\\ &\le |\MM(P,T\fragmentco{i}{i+m}) \cap M| + |\MM(P,T\fragmentco{i}{i+m}) \cap  (M-i)|.
\end{align*}
Then, either $|\MM(P,T\fragmentco{i}{i+m}) \cap M| \ge k/2$ or $|\MM(P,T\fragmentco{i}{i+m}) \cap (M-i)| \ge k/2$ holds. We assume the former holds (the proof for the latter case is similar).
Since $M\setminus M'$ have distinct remainders modulo $p$, we have
\begin{align*}
    |\MM(P,T\fragmentco{i}{i+m}) \bmod p| 
 & \ge|\MM(P,T\fragmentco{i}{i+m}) \cap (M\setminus M')| \\
 & \ge |\MM(P,T\fragmentco{i}{i+m}) \cap M| - |M'| \\
 & \ge 0.5k - 0.01 k, 
\end{align*}
finishing the proof.
\end{proof}

In order to use \cref{cor:primegoodforwindow}, which only applies when $n=\Oh(m)$, we need to divide our possibly long text into pieces each of length $\Oh(m)$.
More specifically, for $P\in \Sigma^m$ and $T\in \Sigma^n$, we partition all candidate positions $\fragmentco{0}{\Delta}$ into $\lfloor n/m\rfloor$ pieces, where the $s$-th piece  is \[I_s\coloneqq\big[sm\dd \min((s+1)m, \Delta)\big ), \;\;\text{ where } s\in [0\dd \lfloor n/m\rfloor),\] with length $m$ (except for the last piece which may be shorter than $m$).
In \cref{alg:oneexecution}, we define the following good event for each piece $I_s$:
\begin{definition}
   \label{prop:eventgoodprime}
   In \cref{alg:oneexecution}, for each $s\in [0\dd \lfloor n/m\rfloor)$, define $\mathcal{E}_{s}$ as the following event:  
   \[\{i\in I_s: |\MM(P,T\fragmentco{i}{i+m}) \bmod p| < 0.49k\} \subseteq \Occ_{k}(P,T).\]

For each $s\in [0\dd \lfloor n/m\rfloor)$, we have $\Pr_{p}[\mathcal{E}_{s}]\ge 1-\Oh(\frac{k\log^9 m}{\hat p})$ by applying \cref{cor:primegoodforwindow} to the text $T\fragmentco{sm}{\min((s+2)m,n)}$ and pattern $P$. \lipicsEnd
\end{definition}

Then, the following lemma summarizes the correctness guarantee of \cref{alg:oneexecution} within each piece $I_s$ (under the event $\mathcal{E}_s$).
\begin{lemma}
   \label{prop:fingerprint-match-probability}
Let $i\in \fragmentco{0}{\Delta}$ and $s= \lfloor i/m\rfloor $ (so that $i\in I_s$). Then, the set $A$ returned by \cref{alg:oneexecution} has the following properties:
\begin{enumerate}
   \item If $i\notin \Occ_k(P,T)$, then $\Pr_{F, B}[i\notin A \mid \mathcal{E}_{s}] \ge 1-n^{-1.9}$.
       \label{item:fingerprint-match-item1}
      \item If $i\in \Occ_{k'}(P,T)$, then $\Pr_{F, B}[i\in A] \ge n^{-4k'/k}$.
       \label{item:fingerprint-match-item2}
\end{enumerate}
\end{lemma}
\begin{proof}
    Recall from Line~\ref{line:defnA} that $i\in A$ if and only if $F(X_{u_i})=F(Y_{v_i}(i))$ so it suffices to analyze the probability that $X_{u_i}$ and $Y_{v_i}(i)$ have the same fingerprint.
   \begin{enumerate}
      \item If $i \notin \Occ_k(P,T)$ then, by definition of event $\mathcal{E}_{s}$ (\cref{prop:eventgoodprime}), we have $|\MM(P,T\fragmentco{i}{i+m}) \bmod p| \ge 0.49k$. Then by \cref{lem:ineqprob}, $\Pr_{B}[X_{u_i} \neq Y_{v_i} (i)] \ge 1-(1-\beta)^{0.49k} = 1-n^{-4/k\cdot 0.49k} = 1-n^{-1.96}$.
Assuming $X_{u_i} \neq Y_{v_i} (i)$ holds, by \cref{lem:fingerprints}, their fingerprints $F(X_{u_i})$ and $F(Y_{v_i}(i))$ are equal with probability at most $\frac{|X_{u_i}|-1}{M-1} \le n^{-9}$. 
By a union bound, the probability of $F(X_{u_i})\neq F(Y_{v_i}(i))$ is at least $1-n^{-1.96}-n^{-9}\ge 1-n^{-1.9}$ as claimed.
\item If $i\in \Occ_{k'}(P,T)$, then 
$|\MM(P,T\fragmentco{i}{i+m})\bmod p|\le k'$.
Hence, by \cref{lem:ineqprob}, with probability at least $(1-\beta)^{k'} = (n^{-4/k})^{k'} = n^{-4k'/k}$ 
it holds that $X_{u_i} = Y_{v_i}(i)$, in which case their fingerprints are also equal. \qedhere
   \end{enumerate}
\end{proof}

\cite{ChanGKKP20} showed that \cref{alg:oneexecution} can be implemented efficiently, and this remains true with our slight adaptations.
\begin{lemma}
   \label{lem:algo-find-fingerprint-matches}
  We can implement \cref{alg:oneexecution} in expected time
   $\Oh\big (\frac{zm+z'n}{k}\log n\big )$  and obtain the output $A\subseteq \fragmentco{0}{\Delta}$ in an implicit representation supporting constant-time random access (specifically, $|A|$ is returned and, given $a\in \fragmentco{0}{|A|}$, the $a$-th smallest element of $A$ can be retrieved in constant time).
\end{lemma}
\begin{proof}
We follow the same implementation as in \cite{ChanGKKP20}. 
Note that Lines~\ref{line:samplep}--\ref{line:samplef} are not the bottleneck:  the random prime $p$ can be sampled in $\poly\log(\hat p)$ expected time, the random subset $B\subseteq \ZZ_p$ can be sampled in $\Oh(\mathbf{E}_B[|B|]) = \Oh(\frac{\hat p\log n}{k})=\Oh(\frac{n}{k}\log n)$ expected time, and the fingerprint function $F$ can be sampled in $\Oh(\log n)$ time. 

We first analyze the complexity of Line~\ref{line:defnXu}. By \cref{lem:fingerprints}, each fingerprint $F(X_u)$ can be evaluated in $\Oh(|X_u|+1)$ time, where $|X_u|$ is the number of positions $j\in \fragmentco{0}{m}$ satisfying $(j+u)\bmod p\in B$.
Since $(j+u)\bmod p\in B$ happens with probability $\beta$ for each $j\in \fragmentco{0}{m}$,  by linearity of expectation, $\mathbf{E}_B[|X_u|]= m\beta =\Oh(\frac{m\log n}{k})$. 
Thus, summing over all $u\in \fragmentco{0}{z}$, the total time of Line~\ref{line:defnXu} is $ \Oh( z\cdot (\frac{m\log n}{k}+1)) = \Oh( \frac{zm\log n}{k})$ in expectation.

Now we implement Line~\ref{line:defnYv}. For each $v \in \fragmentco{0}{z'}$, let $I_v\coloneqq \{i'\in \fragmentco{0}{n}: (i'-vz)\bmod p\in B\}$, which has expected size $\mathbf{E}_B[|I_v|] = \beta n = \Oh(\frac{n\log n}{k})$ because $\beta = 1-n^{-4/k} \le 1-\exp(-\tfrac{4\ln n}{k}) \le \tfrac{4\ln n}{k}$.
Then, $Y_v(i) = \bigodot_{i'\in I_v \cap \fragmentco{i}{i+m}}T[i']$.
Note that the string $Y_v(i)$ differs from $Y_v(i+1)$  only if either $i$ or $i+m$ is in $I_v$.
We maintain a sliding window corresponding to $\fragmentco{i}{i+m} \cap I_v$ (namely, the positions of the characters forming $Y_v(i)$) as $i$ increases, and keep track of the events when elements of $I_v$ enter or leave this window. At the same time, by \cref{lem:fingerprints}, we can maintain the fingerprint $F(Y_v(i))$ upon appending a character or deleting the leftmost character.
The entire process can be done by scanning over the elements of $I_v$ in sorted order in $O(|I_v|)$ time, and it results in the fingerprints $F(Y_v(i))$ computed for all $i\in \fragmentco{0}{\Delta}$, represented as a list of $O(|I_v|)$ tuples $(\fragmentoc{a}{b}; f)$ indicating that $F(Y_v(i))=f$ for all $i\in \fragmentoc{a}{b}$. 
Summing over all $v\in \fragmentco{0}{z'}$, the total time of Line~\ref{line:defnYv} is thus $\Oh(\frac{z'n\log n}{k})$ in expectation.

Now, we  describe how to compute all positions $i$ in the answer set $A$ (Line~\ref{line:defnA}).
We first identify all critical positions $i\in \fragmentco{0}{\Delta-1}$ for which $Y_{v_i}(i)$ may differ from $Y_{v_{i+1}}(i+1)$.
From the definition in Line~\ref{line:defnuivi} we see that $v_i\neq v_{i+1}$ happens at most $\Oh(\frac{\Delta}{p}+\frac{\Delta}{z}+1) = \Oh(\frac{\Delta}{z})$ times.
If $v_i= v_{i+1}=v$, then $Y_{v_i}(i)$ differs from $Y_{v_{i+1}}(i+1)$ only if either $(i-vz)\bmod p$ or $(i+m-vz)\bmod p$ is in $B$, which happens with probability at most $2\beta$. 
Hence, the number of critical positions is $\Oh(\frac{\Delta}{z} + \beta n) = \Oh(\frac{\Delta}{z} +\frac{n\log n}{k})$ in expectation.
(These critical positions $i$ and the fingerprints $F(Y_{v_i}(i))$ at these positions can be determined in linear time, based on the information computed in the previous paragraph.)
Between each pair of consecutive critical positions $a$ and $b$,  the fingerprint $F(Y_{v_i}(i))$ remains the same value $f$, and it suffices to find all $i\in \fragmentoc{a}{b}$ such that $F(X_{u_i})=f$, and include them in the answer~$A$. 
To do this, note that $i\mapsto u_i$ maps $\fragmentoc{a}{b}$ into at most two intervals, so we can find all such positions $i$ by performing predecessor/successor queries on the  sorted list $\{u\in \fragmentco{0}{z}: F(X_{u})= f\}$ (precomputed and preprocessed into a data structure), and the results can be succinctly represented as (concatenation of) contiguous sublists of these precomputed lists.
Finally, the answer $A$, as a sorted list, is the concatenation of all these sublists of positions. We can easily compute $|A|$ by summing up the sizes of these sublists, and can support random access into $A$ by performing predecessor/successor queries on the list of endpoints of these sublists. 
(Here, all predecessor/successor queries can be implemented in constant time using integer data structures in word RAM, such as \cite{PatrascuT14}.)
Therefore, the total time for Line~\ref{line:defnA} is linear in the number of critical positions, which is $\Oh(\frac{\Delta}{z} +\frac{n\log n}{k})$ in expectation. 
Due to $z'\ge 1$, the $\frac{n\log n}{k}$ term is dominated by $\frac{z'n}{k}\log n$.
The same is true for the $\frac{\Delta}{z}$ term because \[\frac{z'n}{k} \ge \frac{\min(2\hat p, \Delta)n}{k}\ge \frac{\min(k, \Delta)\max(k,\Delta)}{kz}=\frac{k\Delta}{kz}=\frac{\Delta}{z}.\]
The overall expected time complexity is thus $\Oh(\frac{zm+z'n}{k}\log n)$.
\end{proof}

Now we are ready to prove the main result:
\begin{proof}[Proof of \cref{thm:main-approx}]
Recall that our goal is to  distinguish between Hamming distance larger than $k$ versus at most $k'$.
In the following, we assume that $k' \le \frac{k}{5}$.
Otherwise, we can solve the problem by brute force in $\Oh(nm)\subseteq \Oh(n^{10 k'/k})$ time.

We set the parameter $\hat p$ in \cref{alg:oneexecution} to be $\hat p = \Theta(n^{4k'/k}\cdot k\log^9 m)$ with a sufficiently large constant factor.
We run $r \coloneqq \lceil 216\cdot n^{4k'/k}\cdot \ln n \rceil$ executions of \cref{alg:oneexecution} (implemented in \cref{lem:algo-find-fingerprint-matches}) independently, and let $\hat A_\ell\subseteq \fragmentco{0}{\Delta}$ denote the set returned by the $\ell$-th execution of \cref{alg:oneexecution}.
Let $\mathcal{T}_{\ell}$ denote the event that execution $\ell$  exceeds $20\cdot n^{4k'/k}$ times the expected time complexity upper bound of \cref{lem:algo-find-fingerprint-matches}, with probability $\Pr[\mathcal{T}_\ell]\le \frac{1}{20}\cdot n^{-4k'/k}$ by Markov's inequality.  
In case of $\mathcal{T}_{\ell}$, we abort execution $\ell$ and let $A_\ell =\fragmentco{0}{\Delta}$; otherwise, execution $\ell$ finishes, and we let $A_\ell = \hat A_\ell$. 
For each execution $\ell \in \fragmentco{0}{r}$ and piece $s\in [0\dd \lfloor n/m\rfloor)$, let $A_{\ell,s} \coloneqq A_\ell \cap I_s$.
We can compute $|A_{\ell,s}|$ for all $\ell \in \fragmentco{0}{r}$ and all pieces $s$.

Let $\alpha \coloneqq \frac{1}{6}\cdot n^{-4k'/k}$. We show that sets $A_{\ell,s}$ have the following two properties:
\begin{claim}
   \label{prop:nocase}
  For each piece $s\in [0\dd \lfloor n/m\rfloor)$, we have
        \[\Pr[| \{ \ell\in \fragmentco{0}{r}: A_{\ell,s} \not\subseteq \Occ_{k}(P,T)\}|\le \alpha r]\ge 1-n^{-6}.\]
\end{claim}
\begin{claimproof}
 Consider each execution $\ell \in \fragmentco{0}{r}$. Let $\mathcal{E}_{\ell,s}$ denote the good event defined in \cref{prop:eventgoodprime} in execution $\ell$ for piece $s$, where $\Pr[\mathcal{E}_{\ell,s}]\ge 1 - \Oh(\frac{k\log^9 m}{\hat p})\ge 1-\frac{1}{100}\cdot n^{-4k'/k}$, provided that the hidden constant in the definition of $\hat p$ is sufficiently large.
 For every $i \in I_{s} \setminus \Occ_{k}(P,T)$, the first statement of \cref{prop:fingerprint-match-probability} implies $\Pr[i\in \hat A_{\ell}\mid \mathcal{E}_{\ell,s}] \le n^{-1.9}$.
 Then, by a union bound over all $i\in I_s\setminus \Occ_{k}(P,T)$, we have \[\Pr\big[\hat A_{\ell} \cap I_s \not\subseteq \Occ_{k}(P,T) \mid \mathcal{E}_{\ell,s} \big] \le n^{-1.9} \cdot |I_s\setminus \Occ_{k}(P,T)| \le  n^{-0.9} \le \tfrac1{100}\cdot n^{-0.8} \le  \tfrac{1}{100}\cdot n^{-4k'/k},\]
 where the two last inequalities require sufficiently large $n$ and $k' \le \frac{k}{5}$.
Hence,
 \[ \Pr\big [\hat A_{\ell} \cap I_s \not\subseteq \Occ_{k}(P,T)  \big ] \le \Pr[\neg \mathcal{E}_{\ell,s}] + \tfrac{1}{100}\cdot n^{-4k'/k} \le \tfrac{2}{100}\cdot n^{-4k'/k}.\]
 Since $A_\ell = \hat A_{\ell}$ when $\neg \mathcal{T}_\ell$, by another union bound,
 \[ \Pr\big [A_{\ell} \cap I_s \not\subseteq \Occ_{k}(P,T)  \big ] \le \Pr[\mathcal{T}_{\ell}] +\tfrac{2}{100}\cdot n^{-4k'/k} \le \tfrac{7}{100}\cdot n^{-4k'/k} < \tfrac{1}{12}\cdot n^{-4k'/k}.\]
 Hence, \[\Pr[A_{\ell,s}\not\subseteq \Occ_k(P,T)] = \Pr\big [A_{\ell} \cap I_s \not\subseteq \Occ_{k}(P,T) \big ]\le \tfrac{1}{12}\cdot n^{-4k'/k}\] for every $\ell \in \fragmentco{0}{r}$.
  
  Since all $r$ executions are independent, by Chernoff bound, we have
 \[ \Pr\big [ | \{ \ell\in \fragmentco{0}{r}: A_{\ell,s} \not\subseteq \Occ_{k}(P,T)\}|\ge 2\cdot \tfrac{1}{12}\cdot n^{-4k'/k} \cdot r \big] \le \exp(- \tfrac13 \cdot \tfrac{1}{12}\cdot n^{-4k'/k} \cdot r) \le n^{-6},  \]
where we used $r \ge 216\cdot n^{4k'/k} \cdot \ln n$.
The claim follows since $2\cdot \frac1{12}\cdot n^{-4k'/k}\cdot r=  \alpha r $.
\end{claimproof}

\begin{claim}
   \label{prop:yescase}
   Let $i\in \Occ_{k'}(P,T)$ and $s= \lfloor i/m\rfloor $ (so that $i\in I_s$).
Then,
\[\Pr[|\{\ell \in \fragmentco{0}{r} : i \in A_{\ell,s} \}| \ge 3\alpha r] \ge 1-n^{-27}.\]
\end{claim}
\begin{claimproof}
Consider each execution $\ell \in \fragmentco{0}{r}$. 
By the second statement of \cref{prop:fingerprint-match-probability}, 
$\Pr[i\in \hat A_\ell] \ge n^{-4k'/k}$.
Since $A_\ell= \hat A_{\ell}$ when $\neg \mathcal{T}_{\ell}$, and $A_{\ell} = \fragmentco{0}{\Delta}\ni i$ when $\mathcal{T}_{\ell}$, we have $\Pr[i\in  A_{\ell,s}]= \Pr[i\in  A_\ell]\ge \Pr[i\in \hat A_\ell]\ge n^{-4k'/k}$.
 Since all $r$ executions are independent, by Chernoff bound, we have
 \[ \Pr\big [ | \{ \ell\in \fragmentco{0}{r}: i\in A_{\ell,s}\}|\le \tfrac{1}{2}\cdot n^{-4k'/k} \cdot r \big] \le \exp(- \tfrac{1}{8}\cdot n^{-4k'/k} \cdot r/8) \le n^{-27},  \]  
where we used $r \ge  216\cdot n^{4k'/k} \cdot \ln n$.
The claim follows since $\frac12 \cdot n^{-4k'/k} \cdot r= 3\alpha r$.
\end{claimproof}

Now we describe the last stage of our algorithm.
For each piece $s$, we sort all executions $\ell\in \fragmentco{0}{r}$ in increasing order of $|A_{\ell,s}|$, and then let $L_{s}\subset \fragmentco{0}{r}$ consist of the first $|L_{s}| \coloneqq\lceil (1-\alpha)r \rceil$ executions in this order.
  If $|A_{\ell,s}|=0$ for all $\ell \in L_s$, then we declare that piece $s$ contains no solution. If none of the pieces has solutions found, we return \No. Otherwise, we return \Yes. If we are solving the reporting version, then in the \Yes case, we report all the solutions in
  \[ A_{out} \coloneqq \bigcup_{s} \big\{ i \in I_s : |\{\ell\in L_s: i\in A_{\ell,s}\}|\ge 2\alpha r \big \},\]
  which can be easily computed after explicitly listing all elements of sets $A_{\ell,s}$ for all pieces $s$ and executions $\ell \in L_s$. This finishes the description of our algorithm.

\paragraph*{Correctness}
 Now we show the correctness of our algorithm. 
 We assume the properties in \cref{prop:nocase} and \cref{prop:yescase} indeed hold, which happens with probability at least $1-(n/m)\cdot n^{-6} - n\cdot n^{-27} \ge 1-n^{-4}$ by union bound. 
 
 Consider any piece $s$.
By \cref{prop:nocase}, for at least $(1-\alpha) r$ many executions $\ell$, it holds that $A_{\ell,s} \subseteq \Occ_{k}(P,T)$ and thus $|A_{\ell,s}| \le |\Occ_k(P,T)\cap I_s|$. 
Since we defined $L_s\subset \fragmentco{0}{r}$ to contain the $\lceil (1-\alpha) r\rceil$ executions with the smallest $|A_{\ell,s}|$, we have
\begin{equation}
\label{eqn:sizebound}
|A_{\ell,s}|\le |\Occ_{k}(P,T) \cap I_s| \text{ for all } \ell \in L_{s}.
\end{equation}
In particular, in the case of $\Occ_k(P,T)=\emptyset$, we have $|A_{\ell,s}|=0$ for all pieces $s$ and all $\ell \in L_s$, so our algorithm indeed returns \No as required.

Now we argue the correctness in the case where $\Occ_{k'}(P,T)\neq \emptyset$. Let $i \in \Occ_{k'}(P,T)$, and suppose $i\in  I_s$. By \cref{prop:yescase},  at least $3\alpha r - (r-|L_s|) \ge 2\alpha r$ many $\ell \in L_s$ satisfy $i\in A_{\ell,s}$.
Hence, our algorithm does not return \No. Moreover, this implies $i\in A_{out}$. Hence, the reporting version of our algorithm satisfies $\Occ_{k'}(P,T) \subseteq A_{out}$ as required.

It remains to prove that $A_{out} \subseteq \Occ_{k}(P,T)$. Take any $i_{bad}\in \fragmentco{0}{\Delta}\setminus \Occ_{k}(P,T)$, and suppose $i_{bad}\in I_s$. \cref{prop:nocase} states that $|\{\ell \in \fragmentco{0}{r}: A_{\ell,s}\setminus \Occ_k(P,T)\neq \emptyset \}| \le \alpha r$, so  $|\{\ell \in L_s: i_{bad}\in A_{\ell,s}\}|\le \alpha r< 2\alpha r$, which means $i_{bad}\notin A_{out}$. This finishes the proof of correctness of our algorithm with high probability.

  \paragraph*{Time complexity}
   We now analyze the time complexity of the algorithm.
   We run $r=\Oh(n^{4k'/k}\log n)$ executions of \cref{alg:oneexecution}, and each execution only runs up to $n^{4k'/k}$ times its expected time complexity upper bound in \cref{lem:algo-find-fingerprint-matches} before being aborted. 
   Recall $\hat p = \Theta(n^{4k'/k} \cdot k\log^9 m)$.
  By \cref{lem:algo-find-fingerprint-matches}, the total time complexity is
  \begin{align*}
   r\cdot n^{4k'/k}\cdot  \Oh\left (\frac{(zm+z'n)\log n}{ k} \right )   = n^{8k'/k} \cdot \log^2 n\cdot  \Oh\left ( \frac{zm+z'n}{ k}\right ) .
  \end{align*}
  Recall that $z$ can be chosen from $[1\dd \min(\hat p,\Delta)]$ and that $z'=\lceil \min(2\hat p,\Delta)/z\rceil$. If $\hat p \le 2n/m$ holds, then we choose $z = \hat p$ and, due to $z'\le 2$, get time complexity
  \[n^{8k'/k} \cdot \log^2 n\cdot  \Oh\left ( \frac{\hat pm + n}{ k}\right ) = \Oh\left (n^{8k'/k} \cdot \log^2 n\cdot  \frac{n }{k} \right ). \]
  If $\hat p> 2n/m$, then we choose $z = \lceil \sqrt{\min(2\hat p,\Delta)n/m}\rceil$. 
  In this case, $n/m \le \hat p/2$, and thus $z \le \hat p$.
  Moreover, we always have $n/m \le \Delta$ and therefore $z \le \Delta \le n$.
  Consequently, the requirement $z\in [1\dd \min(\hat p,\Delta)]$ is satisfied.
  Moreover, $z \le 2\sqrt{\min(2\hat p,\Delta)n/m}$ and, due to $z \le \min(2\hat p,\Delta)$, similarly $z' =  \lceil\min(2\hat p,\Delta)/z\rceil  \le  \lceil\sqrt{\min(2\hat p,\Delta)n/m}\rceil \le 2\sqrt{\min(2\hat p,\Delta)n/m}$.
  Consequently, we get the time complexity of 
  \[ \Oh\left (n^{8k'/k} \cdot \log^2 n\cdot  \frac{\sqrt{\min(\hat p,\Delta)n m} }{k} \right ) = \Oh\left(n^{10k'/k} \cdot \log^{2} n\cdot \min\left(\sqrt{\frac{nm\log^9 m}{k}},\frac{n\sqrt{\Delta}}{k}\right)\right).\]
  Hence, the overall time complexity for the decision version is both in $\Oh\left (\sqrt{\frac{nm}{k}}\log^{6.5}n + \frac{n}{k}\log^2 n\right )\cdot n^{\Oh(k'/k)}$ and in $\Oh\left (\frac{n\sqrt{\Delta}}{k}\log^2 n\right )\cdot n^{\Oh(k'/k)}$  as claimed.
  
  Now we analyze the extra time for solving the reporting version. This requires explicitly listing all elements of sets $A_{\ell,s}$ for all pieces $s$ and executions $\ell \in L_s$, and then computing $A_{out}$ by definition in linear time in the total number of listed elements. Hence, by \cref{eqn:sizebound}, the total extra  time is (up to a constant factor) $\sum_{s}\sum_{\ell \in L_s}|A_{\ell,s}|\le \sum_{s}\sum_{\ell \in L_s}|\Occ_{k}(P,T) \cap I_s|  \le r \cdot  \sum_{s}|\Occ_{k}(P,T) \cap I_s| = r\cdot |\Occ_k(P,T)| = \Oh(n^{4k'/k}\cdot \log n)\cdot |\Occ_k(P,T)|$.
\end{proof}
\section{Adaptive Algorithm}
\label{sec:adaptive}
In this section, we present our adaptive property tester, proving \cref{thm:adaptive-tester}.
Let pattern $P\in \Sigma^m$, text $T\in \Sigma^n$, and threshold $k\in [1\dd m-1]$. We are required to distinguish between the \Yes case where $\Occ(P,T)\neq \emptyset$ and the \No case  where $\Occ_{k}(P,T)=\emptyset$.

Denote $\Delta = n-m+1\ge 1$. 
    If $\Delta > 0.1n$, then we can simply apply our non-adaptive algorithm from \cref{thm:main}  to solve the problem in time $\Oh\Big(\Big(\sqrt{\tfrac{nm\log^9 n}{k}} + \tfrac{n}{k}\Big)\log^2 n\Big)$, which is within the time bound  $\Oh\left(\left(\sqrt{\tfrac{m\Delta\log^{29}n}{k}} + \tfrac{n}{k}\right)\log^{14} n\right)$ claimed in \cref{thm:adaptive-tester}. Therefore, in the following we assume $\Delta \le 0.1n$. In particular, $\Delta \le 0.2m$, and $n<1.2m$.

Throughout this section, let $c\ge 1$ denote the absolute constant hidden by the $\Oh(\cdot )$ notation in \cref{thm:mismatchcontainer}, and define
\begin{equation}
    \label{eqn:defneps}
 \eps \coloneqq \left (20c\log^4 m\right )^{-1}.
\end{equation}

We follow the outline in the technical overview (\cref{subsec:overview-adaptive}).
As mentioned in the last paragraph of \cref{subsec:overview-adaptive}, we need to consider \emph{distributions} over candidate occurrences in $[0\dd \Delta)$.
For a distribution $O$ supported on $[0\dd \Delta)$,
(almost) good  positions are positions in $P$ or $T$ which can eliminate candidate occurrences with total probability mass at least $1/\poly \log n$ in~$O$, formally defined as follows:

\begin{definition}[Good and almost good positions]
\label{defn:goodpos}
Let $O$ be a distribution supported on $[0\dd \Delta)$.
We say that a position $i\in [0\dd n)$ in $T$ is \emph{good with respect to $O$} if 
\begin{equation}
\label{eqn:goodposT}
   \Pr_{x\sim O}\big[ i-x\in [0\dd m) \text{ and } P[i-x]\neq T[i] \big] \ge 4\eps.
\end{equation}
Similarly, we say a position $i\in [0\dd m)$ in $P$ is \emph{good with respect to $O$} if
   \begin{equation}
\label{eqn:goodposP}
   \Pr_{x\sim O}\big[P[i]\neq T[i+x] \big] \ge 4\eps.
   \end{equation}

We say the position $i$ in $T$ (or $P$) is \emph{almost good with respect $O$} if the inequality in \eqref{eqn:goodposT} or \eqref{eqn:goodposP} is satisfied with  the right-hand side replaced by $2\eps$.\lipicsEnd
\end{definition}

The following lemma shows that there exist $\Omega(k)$ good positions in the \No case.

\begin{lemma}
\label{lem:numgoodpos}
Suppose $\Occ_k(P,T) = \emptyset$, and let $O$ be a distribution supported on $[0\dd \Delta)$.
Then, $T$ and $P$ in total have at least $k/4$ good positions with respect to $O$.
\end{lemma}
\begin{proof}
Let $M\subseteq \fragmentco{0}{n}$ be the set obtained from applying \cref{thm:mismatchcontainer} to $P$, $T$, and $k$, with size $|M|\le c\cdot \frac{n}{m}k\log^4 m \le 1.2ck\log^4 m= 0.06k/\eps$.
Since $\Occ_k(P,T)=\emptyset$, the property of $M$ implies that $|\MM(P,T[x\dd x+m)) \cap (M \cup (M-x))| \ge k$ holds for all $x\in [0\dd \Delta)$. 
Then, since $O$ is supported on $[0\dd\Delta)$, we get 
\begin{align*}
 k &\le \mathbf{E}_{x\sim O}[|\MM(P,T[x\dd x+m)) \cap (M \cup (M-x))|] \\
 & \le  \mathbf{E}_{x\sim O}[|\MM(P,T[x\dd x+m)) \cap M |]  + \mathbf{E}_{x\sim O}[|\MM(P,T[x\dd x+m)) \cap (M-x) |].
\end{align*}
Then, at least one of the two terms is no less than $\frac{1}{2}k$. We consider the two cases separately:
\begin{itemize}
    \item Case $\frac{1}{2}k \le \mathbf{E}_{x\sim O}[|\MM(P,T[x\dd x+m)) \cap M|]$:

    By linearity of expectation, this means 
    \begin{align*}
    \tfrac{1}{2} k&\le \sum_{i\in M}\mathbf{E}_{x\sim O}[|\MM(P,T[x\dd x+m)) \cap \{i\}|]\\
    & = \sum_{i\in M \cap [m]}\Pr_{x\sim O}\big [P[i]\neq T[i+x]\big],
    \end{align*}
    where each summand is exactly the left-hand side of \cref{eqn:goodposP}.
 Using $|M\cap [m]|\le |M|\le  0.06k/\eps$, we see that the number of positions $i\in M\cap [m]$ satisfying \cref{eqn:goodposP} is at least $\frac{1}{2}k - 4\eps \cdot |M\cap [m]|\ge 0.26k\ge k/4$. Hence, there exist at least $k/4$ good positions in $P$ with respect to $O$.
    
    \item Case $\frac{1}{2}k\le \mathbf{E}_{x\sim O}[|\MM(P,T[x\dd x+m)) \cap (M-x)|]$:
    
     The proof is similar to the first case. We have
    \begin{align*}
    \tfrac{1}{2} k&\le \sum_{i\in M}\mathbf{E}_{x\sim O}[|\MM(P,T[x\dd x+m)) \cap \{i-x\}|]\\
    & = \sum_{i\in M}\Pr_{x\sim O}\big [i-x\in [0\dd m)\text{ and } P[i-x]\neq T[i]\big],
    \end{align*}
    where each summand is exactly the left-hand side of \cref{eqn:goodposT}.
Then, using $|M|\le  0.06k/\eps$, we see that the number of positions $i\in M$ satisfying \cref{eqn:goodposT} is at least $\frac{1}{2}k - 4\eps \cdot |M|\ge 0.26k\ge k/4$. Hence, there exist at least $k/4$ good positions in $T$ with respect to $O$.
\end{itemize}
Therefore, there exist at least $k/4$ good positions with respect to $O$, in either $P$ or $T$, concluding the proof.
\end{proof}

We now define blocks and good blocks:
\begin{definition}[Blocks and good blocks]
Partition $T\in \Sigma^n$ into $\lceil n/\Delta\rceil$ \emph{blocks}, where the $j$-th block $(j\ge 0)$ is the fragment $T\big [j \Delta \dd  \min((j+1)\Delta, n)\big )$.
 Similarly, partition $P\in \Sigma^m$ into  $\lceil m/\Delta\rceil$ blocks, where
 the $j$-th block $(j\ge 0)$ is the fragment $P\big [j \Delta \dd  \min((j+1)\Delta, m)\big )$.
 All blocks have lengths in $[1\dd \Delta]$.

For a distribution $O$ supported on $[0\dd\Delta)$, we say a block is \emph{good with respect to $O$} if it contains at least $\frac{k\Delta}{64n}$ almost good positions with respect to $O$.\lipicsEnd
\end{definition}

In the \No case, we can find a good block by sampling:

\begin{lemma}
\label{lem:findblock}
Suppose $\Occ_k(P,T) = \emptyset$, and let $O$ be a distribution supported on $[0\dd \Delta)$.
There is an algorithm which draws $L = \lceil 5\eps^{-2}\ln n \rceil =  \Oh(\log^{9}n)$ samples from $O$, spends additional time $\Oh\left (\frac{n}{k}\cdot \log^9 n\right )$, and returns the index of a good block (in $P$ or in $T$) with respect to $O$, with success probability at least $0.6$.
\end{lemma}
\begin{proof}
    Let $x_1,\dots, x_L$ be the independent samples drawn from $O$.
We use empirical means computed from these samples to approximate the left-hand sides of \cref{eqn:goodposT,eqn:goodposP}.
Let $\hat O$ denote the uniform distribution over the samples $x_1,\dots,x_{L}$.
For any fixed position $i\in [n]$ in $T$, by Chernoff bound, 
\begin{align}
\Bigg \lvert  &\Pr_{x\sim O}\Big[ i-x\in [0\dd m) \text{ and } P[i-x]\neq T[i] \Big]   - \Pr_{x\sim \hat O}\Big[ i-x\in [0\dd m) \text{ and } P[i-x]\neq T[i] \Big]\Bigg \rvert \le \eps \label{eqn:estimateprob}
\end{align}
holds with probability at least $1-2\exp(-2L \eps^2 ) \ge 1-2n^{-10}$ over the randomness of $x_1, x_2,\dots, x_{L}$.
 Then, by a union bound, the estimate in \cref{eqn:estimateprob} holds for all $i \in [n]$ with probability at least $1-2n^{-9}$.
Assuming this indeed holds, define a set $I_T\subseteq[n]$ of positions in $T$ as
\begin{equation}
I_T\coloneqq \Big \{i\in [n] \; : \; \Pr_{x\sim \hat O}\Big[ i-x\in [0\dd m) \text{ and } P[i-x]\neq T[i] \Big]  \ge 3\eps \Big \}.
\label{eqn:defnit}
\end{equation}
Then,  by \cref{eqn:estimateprob} and \cref{defn:goodpos}, all good positions in $T$ are contained in $I_T$, and conversely, all $i\in I_T$ are almost good positions in $T$.

Similarly, with probability at least $1-2n^{-9}$, for all positions $i\in [m]$ in $P$, the left-hand side of \cref{eqn:goodposP} is also $\eps$-additively approximated by the empirical means in analogy to \cref{eqn:estimateprob}. Assuming this holds, we can define the set $I_P\subseteq [m]$ in analogy to \cref{eqn:defnit}, so that all good positions in $P$ are contained in $I_P$, and all $i\in I_P$ are almost good positions in $P$.
 Given any position $i\in [n]$ (or $i\in [m]$), one can easily decide whether $i\in I_T$ (or whether $i\in I_P$) by definition in $\Oh(L)= \Oh(\eps^{-2}\log n) = \Oh(\log^9 n)$ time.

We say a block in $T$ (or in $P$) is \emph{light} if it contains at most $\frac{k\Delta}{64n}$ positions that are in $I_T$ (or in~$I_P$); otherwise, we say the block is \emph{heavy}. 
Then, a heavy block is a good block.
All light blocks in $T$ and $P$ in total contain at most $\frac{k\Delta}{64n}\cdot \left (\left \lceil \frac{n}{\Delta} \right \rceil + \left \lceil \frac{m}{\Delta}\right \rceil\right ) \le k/16$ positions in $I_T\sqcup I_P$.\footnote{Note that $I_T$ and $I_P$ are subsets of positions in $T$ and $P$ respectively, and are viewed as disjoint sets (even though they are indexed by overlapping sets $[n]$ and $[m]$).} However, by \cref{lem:numgoodpos}, there are at least $k/4$ good positions in total, which are all in $I_T\sqcup I_P$, so at least $\frac{k/4-k/16}{k/4} =\frac{3}{4}$ fraction of the positions in $I_T\sqcup I_P$ are in heavy blocks.
Hence, for a uniformly random position $i$ drawn from $I_T\sqcup I_P$, with at least $3/4$ probability, the block containing $i$ is a heavy block, and thus is a good block as desired.

Therefore, in order to find a good block, it suffices to draw a uniform sample from $I_T\sqcup I_P$, which can be done by a simple rejection sampling as follows: repeatedly draw uniform samples from positions $i\in [n]$ in $T$ and positions $i\in [m]$ in $P$, and return the first one that is found to be in $I_T\sqcup I_P$.
By \cref{lem:numgoodpos}, there are at least $k/4$ good positions, all of which are in $I_T\sqcup I_P$, so each sample is successful with probability at least $\frac{k/4}{n+m} \ge \frac{k}{8n}$. Hence, the number of required samples exceeds $\frac{80n}{k}$ with probability at most $(1-\frac{k}{8n})^{80n/k} \le e^{-10}$. Recall that testing whether a sampled position is in $I_T$ (or in $I_P$) takes $\Oh(\log^9 n)$ time, so the total time complexity for testing up to $\frac{80n}{k}$ samples is $\Oh\left (\frac{n}{k}\cdot \log^9 n\right )$.

Finally, by a union bound, the probability that we successfully return a position, and the block containing it is indeed a good block, is at least $\frac{3}{4}-2n^{-9}-2n^{-9} - e^{-10} \ge 0.6$.
\end{proof}

Given a good block, we use its length-$\Theta(\Delta)$ neighborhood to define a smaller property testing instance $(P_0,T_0)$ with smaller threshold $k_0$, so that solving this instance would allow us to eliminate a large fraction of candidate occurrences in $O$:

\begin{lemma}
\label{lem:eliminatestep}
Let block $B$ be $T[a\dd b)$ or $P[a\dd b)$.
Define $a_0\coloneqq \min\big (m-2\Delta, \max(a-\Delta,0) \big )$ and fragments $P_0=P[a_0\dd a_0+2\Delta)$, $T_0 = T[a_0\dd a_0+3\Delta-1)$. 
Then,  $\Occ(P,T)\subseteq \Occ(P_0,T_0)$.
Moreover, if $B$ is a good block with respect to a distribution  $O$ supported on $[0\dd \Delta)$, then 
    \begin{equation}
    \label{eqn:prhdbound}
 \Pr_{x\sim O} [x \notin \Occ_{k_0}(P_0,T_0) ] \ge \eps,
    \end{equation}
 where 
 \begin{align}
     \label{eqn:defnk0}
 k_0 &\coloneqq \left \lceil \tfrac{k\Delta}{64n}\cdot \eps\right \rceil - 1\ge 0.
 \end{align}
\end{lemma}
\begin{proof}
First, note that $a_0\ge \min(m-2\Delta,0) \ge 0$, and $a_0\le m-2\Delta = n-3\Delta+1$, so $P_0$ and $T_0$ are well-defined fragments of $P[0\dd m)$ and $T[0\dd n)$ respectively, with $|T_0|-|P_0|+1= \Delta$. For any $x\in [0\dd \Delta)$, clearly $\MM(P_0,T_0[x\dd x+2\Delta)) + a_0 \subseteq \MM(P,T[x\dd x+m))$, so $\Occ(P,T)\subseteq \Occ(P_0,T_0)$ as claimed. It remains to show \cref{eqn:prhdbound} assuming $B$ is a good block.

We will establish
\begin{equation}
\label{eqn:toprove}
 \Pr_{x\sim O}\Big [\big \lvert\MM(P,T[x\dd x+m) )  \, \cap \, [a-\Delta \dd b) \big \rvert  > k_0  \Big ] \ge \eps.
\end{equation}
To derive \cref{eqn:prhdbound} from \cref{eqn:toprove}, first note that $\MM(P,T[x\dd x+m) ) \cap [a-\Delta \dd b) \subseteq [0\dd m) \cap [a-\Delta \dd b) \subseteq [a_0\dd a_0+2\Delta)$, where the last step follows from the definition of $a_0$ and that the given block has length $b-a\le \Delta$. Then, 
$\big \lvert\MM(P,T[x\dd x+m) )  \, \cap \, [a-\Delta \dd b) \big \rvert \le \big \lvert\MM(P,T[x\dd x+m) )  \, \cap \, [a_0\dd a_0 + 2\Delta) \big \rvert = \HD(P_0,T_0[x\dd x+2\Delta))$
holds for all $x\in [0\dd \Delta)$. Hence, \cref{eqn:toprove}
implies $\Pr_{x\sim O} \Big[\HD(P_0,T_0[x\dd x+2\Delta)) > k_0 \Big ] \ge \eps$, which is equivalent to \cref{eqn:prhdbound}.

To prove \cref{eqn:toprove}, we first consider the case  $B=T[a\dd b)$.
Let $G$ be the set of almost good positions in the good block  $T[a\dd b)$ with respect to $O$, where $|G|\ge \frac{k\Delta}{64n}$. Then, by \cref{defn:goodpos}, 
    \[ \sum_{i\in G} \Pr_{x\sim O}\big[ i-x\in [0\dd m) \text{ and } P[i-x]\neq T[i] \big] \ge 
     2\eps |G|.\]
    By changing the summation order, we get
    \begin{align*}
 2\eps|G| &\le \mathbf{E}_{x\sim O}\sum_{i\in G} \mathbf{1}\big [ i-x \in [0\dd m) \text{ and } P[i-x]\neq T[i]\big ]\\
        & = \mathbf{E}_{x\sim O} \big [|\MM(P,T[x\dd x+m) )  \, \cap \, (G-x) |\big ].
    \end{align*}
    Then, since $0 \le |\MM(P,T[x\dd x+m) )  \, \cap \, (G-x) | \le |G|$ always holds, we must have 
    \[\Pr_{x\sim O}\Big [|\MM(P,T[x\dd x+m) )  \, \cap \, (G-x) | \ge \eps |G|\Big ] \ge  \eps.\]
    Since $G-x \subseteq [a\dd b) - [0\dd \Delta)  \subset [a-\Delta\dd b)$ always holds, this means
    \[\Pr_{x\sim O}\Big [|\MM(P,T[x\dd x+m) )  \, \cap \, [a-\Delta\dd b) | \ge \eps |G|\Big ] \ge  \eps,\]
    which immediately implies \cref{eqn:toprove} since $\lceil \eps|G|\rceil  \ge \left \lceil \tfrac{k\Delta}{64n}\cdot \eps\right \rceil  = k_0+1$.

    Now we prove \cref{eqn:toprove} in the case $B=P[a\dd b)$. The proof is analogous to the previous paragraph. Let $G$ denote the set of almost good positions in the block $P[a\dd b)$ with respect to $O$. Then, by \cref{defn:goodpos} and changing the summation order, we get
    \begin{align*}
2\eps |G| \le      \sum_{i\in G} \Pr_{x\sim O}\big[P[i]\neq T[i+x] \big]  =  \mathbf{E}_{x\sim O}  \big [|\MM(P,T[x\dd x+m) )  \cap  G |\big ],
    \end{align*}
    which implies 
    \[\Pr_{x\sim O}\Big [|\MM(P,T[x\dd x+m) )   \cap  G | \ge \eps |G|\Big ] \ge  \eps.\]
    Then, \cref{eqn:toprove} follows from $G\subseteq[a\dd b) \subset [a-\Delta\dd b)$, and 
    $\lceil \eps|G|\rceil  \ge \left \lceil \tfrac{k\Delta}{64n}\cdot \eps\right \rceil  = k_0+1$.
\end{proof}

Now we introduce the following technical \cref{lem:algomultiple}, which builds on our non-adaptive property tester described in \cref{sec:algo}. To motivate,
suppose we have obtained multiple small property testing instances
$(P_1,T_1),\dots, (P_d,T_d) \in \Sigma^{2\Delta}\times \Sigma^{3\Delta-1}$ via \cref{lem:eliminatestep} (in particular, $\Occ(P,T) \subseteq \bigcap_{j=1}^d\Occ(P_j,T_j) $). Then, we may use \cref{lem:algomultiple} to solve these instances in a combined fashion, and obtain a candidate set $A$ that satisfies $\bigcap_{j=1}^d\Occ(P_j,T_j) \subseteq A \subseteq \bigcap_{j=1}^d\Occ_{k_0}(P_j,T_j)$.

We will apply \cref{lem:algomultiple} with  parameters $D,\delta^{-1}\le (\log n)^{\Oh(1)}$, so its time complexity will be $\Oh(\frac{\Delta}{\sqrt{k_0+1}})\cdot (\log n)^{\Oh(1)}$ in this case. For later analysis,  we need the technical condition (\cref{eqn:coupling}) that, with good probability under some coupling, adding an extra input instance $(P_d,T_d)$  would not introduce new elements to the answer set returned by the algorithm.

\begin{lemma}
\label{lem:algomultiple}
Let $D\ge 1,  k_0\in [0\dd 2\Delta]$, and $\delta \in (0,1)$. There is a randomized algorithm $\mathcal{A}$ 
with expected time complexity 
\[\Oh\left ( \frac{\Delta}{\sqrt{k_0+1}}\cdot \left (D\delta^{-0.5}\log^{5.5}(2\Delta) + D^{1.1}\delta^{-1}\right )\right ),\]
which takes $0\le d\le D$ pairs $(P_1,T_1),\dots, (P_d,T_d) \in \Sigma^{2\Delta}\times \Sigma^{3\Delta-1}$ as input, uses randomness $r$, and returns
     a set $A = \mathcal{A}_{r}(P_1,T_1,\dots,P_d,T_d) \subseteq [0\dd \Delta)$ in an implicit representation supporting constant-time random access, such that:
    \begin{enumerate}
        \item $[0\dd \Delta)\cap \bigcap_{j=1}^d \Occ(P_j,T_j) \subseteq A$ (deterministically), and
            \label{item:multiple-condition1}
        \item $\Pr_{r}\left [A \subseteq \bigcap_{j=1}^d \Occ_{k_0}(P_j,T_j)\right ] \ge 1-\delta$.
            \label{item:multiple-condition2}
    \end{enumerate}

    Additionally, for any $(P_1,T_1),\dots, (P_d,T_d) \in \Sigma^{2\Delta}\times \Sigma^{3\Delta-1}$ $(d\ge 1)$, it holds that
    \begin{equation}
        \label{eqn:coupling}
     \Pr_{r}[\mathcal{A}_{r}(P_1,T_1,\dots,P_{d-1},T_{d-1})  \supseteq \mathcal{A}_{r}(P_1,T_1,\dots,P_{d},T_{d})] \ge 1-\delta/\Delta. 
    \end{equation}
\end{lemma}
\begin{proof}
First, we assume $\Delta \ge 10D^{0.2}\delta^{-2}$;
otherwise, 
the task can be solved by a brute force algorithm which deterministically computes all $\Occ(P_j,T_j)$ by exact pattern matching~\cite{MP70,KMP77,BM77} and returns $[0\dd \Delta) \cap \bigcap_{j=1}^d \Occ(P_j,T_j)$, in total time complexity $\Oh(\Delta D) = \Oh(\sqrt{\Delta} D\cdot D^{0.1}\delta^{-1})  = \Oh(\frac{\Delta}{\sqrt{k_0+1}}\cdot D^{1.1}\delta^{-1})$, which is within the desired time bound. One can verify that this brute force algorithm satisfies both conditions \cref{item:multiple-condition1,item:multiple-condition2} and \cref{eqn:coupling}.
By a similar argument, we can also assume $\Delta$ is larger than some sufficiently large constant, and assume that $k_0\ge 1$.

\paragraph*{Description of Algorithm $\mathcal{A}$}

First, define a zipped pattern $\hat P \in (\Sigma^d)^{2\Delta}$, where for each $i\in [0\dd 2\Delta)$,  $\hat P[i]$ is the concatenation $P_1[i]P_2[i]\cdots P_d[i]$ viewed as a symbol in $\Sigma^d$. Similarly, define a zipped text $\hat T \in (\Sigma^d)^{3\Delta-1}$.
 
 Sample a fingerprint function $G\colon \Sigma^* \to \mathbb{F}_{q}$ from \cref{lem:fingerprints} with prime $q \in [\Delta^9,2\Delta^9]$.  Define a hashed pattern $ P \in \mathbb{F}_q^{2\Delta}$ and a hashed text  $ T \in \mathbb{F}_q^{3\Delta-1}$, where $ P[i] = G(\hat P[i])$ and $ T[j] = G(\hat T[j])$.
 Note that every symbol in $ P$ and $ T$ can be accessed in $\Oh(d)=\Oh(D)$ time by evaluating the fingerprint of a  string in $\Sigma^d$.

 Then, we run \cref{alg:oneexecution} on $P$, $T$, with $n\coloneqq3\Delta-1$, $m\coloneqq 2\Delta$ (where $n$ is indeed larger than some sufficiently large constant, as required in \cref{sec:algo}), $k\coloneqq k_0$, and $\hat p \coloneqq \Theta (\delta^{-1}k_0\log ^9(2\Delta))$. The parameter $z$ of \cref{alg:oneexecution} will be chosen later. We return the set $A\subseteq [0\dd \Delta)$ returned by \cref{alg:oneexecution} as our answer.

\paragraph*{Proof of \cref{item:multiple-condition1,item:multiple-condition2}}

First,  by definition of $\hat P$ and $\hat T$, it holds that 
$
\Occ(\hat P,\hat T) = [0\dd\Delta) \cap \bigcap_{j=1}^d \Occ(P_j,T_j),
$
 and 
$ \Occ_{k_0}(\hat P,\hat T) \subseteq \bigcap_{j=1}^d \Occ_{k_0}(P_j,T_j)$.

By definition of $ P$ and $ T$, $ \Occ(\hat P,\hat T)\subseteq \Occ( P, T)$ always holds.
For all $i\in \Occ(P, T)$, by \cref{prop:fingerprint-match-probability} (\cref{item:fingerprint-match-item2}) with $k'=0$, we have $\Pr[i\in A] =1$, so $\Occ( P, T)\subseteq A$ holds deterministically.
Hence, $ [0\dd \Delta)\cap \bigcap_{j=1}^d \Occ(P_j,T_j)= \Occ(\hat P,\hat T) \subseteq \Occ( P, T) \subseteq A$,  which proves \cref{item:multiple-condition1}.

Since $ P[i]=G(\hat P[i])$ and $T[j] = G(\hat T[j])$, by the property of the fingerprint function $G$,
 we have $ P[j]\neq  T[i]$ for all $(i,j)\in [2\Delta]\times [3\Delta -1]$ with $\hat P[j]\neq \hat T[i]$,  with probability at least $1-2\Delta\cdot(3\Delta-1)\cdot \frac{d-1}{q-1} \ge 1-\Delta^{-2}$ (where we used $\Delta\ge 10D^{0.2}\ge 10d^{0.2}$ and $q\ge \Delta^9$) by a union bound over all $(i,j)$.  
Hence, we have $\Occ_{k_0}( P, T) = \Occ_{k_0}(\hat P,\hat T)$ with at least $1-\Delta^{-2}$ probability.

The good event $\mathcal{E}_0$ in \cref{prop:eventgoodprime} (here we only have $\lfloor n/m\rfloor = 1$ piece) happens with 
$\Pr[\mathcal{E}_{0}]\ge 1-\Oh(\frac{k_0\log^9 (2\Delta)}{\hat p}) = 1-\delta/2$ (assuming $\hat p$ is defined with a sufficiently large constant factor).
For any $i\notin \Occ_{k_0}( P, T)$,
by \cref{prop:fingerprint-match-probability} (\cref{item:fingerprint-match-item1}), we have $\Pr[i\notin A \mid \mathcal{E}_{0}] \ge 1-\Delta^{-1.9}$.
By a union bound over all $i\in [0\dd \Delta)\setminus \Occ_{k_0}( P, T)$, we then get $\Pr[A \subseteq \Occ_{k_0}( P, T)] \ge
\Pr[\mathcal{E}_0] - \Delta \cdot \Delta^{-1.9}\ge 
1- \delta/2-\Delta^{-0.9} $.
By the previous paragraph and another union bound,
$\Pr[A \subseteq \Occ_{k_0}( \hat P, \hat T)] \ge \Pr[A \subseteq \Occ_{k_0}(  P,  T)] - \Delta^{-2} \ge 1- \delta/2-\Delta^{-0.9} - \Delta^{-2} \ge 1-\delta$, where we used $\Delta\ge 10\delta^{-2}$ in the last step.
Since $ \Occ_{k_0}(\hat P,\hat T) \subseteq \bigcap_{j=1}^d \Occ_{k_0}(P_j,T_j)$, this proves \cref{item:multiple-condition2}.

\paragraph*{Time Complexity}
 By \cref{lem:algo-find-fingerprint-matches}, the time complexity of the algorithm is 
$\Oh\big (D\cdot \frac{z\cdot 2\Delta+z'\cdot (3\Delta-1)}{k_0 }\log (3\Delta-1)\big ) = \Oh\big (\frac{(z+z')\Delta}{k_0}D\log(2\Delta)\big )$ in expectation, where the extra $\Oh(D)$ factor is the time complexity for accessing each symbol of $ P$ and $ T$.
Recall that $z$ can be chosen from $[1\dd \min(\hat p,\Delta)]$ and that $z'=\lceil \min(2\hat p,\Delta)/z\rceil$. 
  We choose $z = \lceil \sqrt{\min(2\hat p,\Delta)}\rceil$, which satisfies the requirement $z\in [1\dd \min(\hat p,\Delta)]$ since $\hat p\ge 2$ and  $\Delta \ge 1$.
  Then, $z' =  \lceil\min(2\hat p,\Delta)/z\rceil \le \lceil\sqrt{\min(2\hat p,\Delta)}\rceil $.
  Consequently, $z+z' \le 2\lceil\sqrt{\min(2\hat p,\Delta)}\rceil= \Oh(\sqrt{\hat p})$, and we get the expected time complexity
  \[ \Oh\left ( \frac{\Delta\sqrt{\hat p}}{k_0} D\log(2\Delta) \right ) =\Oh\left ( \frac{\Delta\sqrt{\delta^{-1}k_0\log ^9(2\Delta)}}{k_0} D\log(2\Delta) \right ) = \Oh\left( \frac{\Delta}{\sqrt{k_0}}D\delta^{-0.5}\log^{5.5} (2\Delta)\right), \]
  which is within the claimed time bound (as we assumed $k_0\ge 1$).

  \paragraph*{Proof of \cref{eqn:coupling}}
Given $(P_1,T_1),\dots,(P_d,T_d)$, consider the zipped pattern $\hat P \in (\Sigma^d)^{2\Delta}$ and zipped text $\hat T \in (\Sigma^d)^{3\Delta-1}$ defined as before, as well as the zipped pattern $\hat P' \in (\Sigma^{d-1})^{2\Delta}$ and zipped text $\hat T' \in (\Sigma^{d-1})^{3\Delta-1}$ defined based on $(P_1,T_1),\dots,(P_{d-1},T_{d-1})$ only. Note that
$(\hat P[i]=\hat T[j]) \Rightarrow (\hat P'[i]=\hat T'[j])$ holds for all $(i,j)\in [2\Delta]\times [3\Delta-1]$.
Let $ P, P'\in \mathbb{F}_q^{2\Delta}$ and $T, T'\in \mathbb{F}_q^{3\Delta-1}$ be obtained by applying the fingerprint function $G\colon \Sigma^*\to \mathbb{F}_q$ entrywise to $\hat P,\hat P',\hat  T,\hat T'$. 
Then, by a union bound,  we get
\begin{equation}
( P[i]= T[j]) \Rightarrow ( P'[i]= T'[j])
\label{eqn:symbolcondition}
\end{equation}
for all $(i,j)\in [2\Delta]\times [3\Delta-1]$, with probability at least $1-2\Delta\cdot(3\Delta-1)\cdot \frac{d-1}{q-1} \ge 1-\Delta^{-2}$ (where we used $\Delta\ge 10D^{0.2}\ge 10d^{0.2}$ and $q\ge \Delta^9$) over the choice of $G$.

  Let $r=(G,p,B,F)$ denote the randomness used by our algorithm $\mathcal{A}$, which consists of the fingerprint function $G$, as well as the randomness used by \cref{alg:oneexecution} (which we denote by $\mathcal{A}^1$), namely
prime $p \in [\hat p,2\hat p]$,  subset $B\subseteq \ZZ_{p}$, and fingerprint function $F\colon \mathbb{F}_q^* \to \mathbb{F}_M$ (with $M\ge (3\Delta-1)^{10}$).

 By Line~\ref{line:defnA} of \cref{alg:oneexecution}, the returned set is
\[\mathcal{A}_{r}(P_1,T_1,\dots,P_{d},T_{d}) = \mathcal{A}^1_{p,B,F}(P,T) = \{i \in [0\dd \Delta) : F(X_{u_i}) = F(Y_{v_i}(i))\},\] where $X_{u_i}$ and $Y_{v_i}(i)$ were defined in Lines~\ref{line:defnXu} and \ref{line:defnYv} (based on $p$ and $B$, but not $F$). 
 For any $i\in [0\dd \Delta)$ with $X_{u_i}\neq Y_{v_i}(i)$, by \cref{lem:fingerprints}, $F(X_{u_i}) = F(Y_{v_i}(i))$ happens with probability  at most $\frac{|X_{u_i}|-1}{M-1}\le \Delta^{-9}$ over the choice of $F$. Hence, by a union bound over all $i\in [0\dd\Delta)$, with at least $1-\Delta^{-8}$ probability, we have
\[\mathcal{A}_{r}(P_1,T_1,\dots,P_{d},T_{d}) = \mathcal{A}^1_{p,B,F}(P,T) = \{i \in [0\dd \Delta) : X_{u_i} = Y_{v_i}(i)\}.\]
Similarly, with at least $1-\Delta^{-8}$ probability over the choice of $F$, we have
\[\mathcal{A}_{r}(P_1,T_1,\dots,P_{d-1},T_{d-1}) = \mathcal{A}^1_{p,B,F}(P',T') = \{i \in [0\dd \Delta) : X'_{u_i} = Y'_{v_i}(i)\},\]
where $X'_{u_i}$ and $Y'_{v_i}(i)$ consist of symbols from $P'$ and $T'$ (instead of $P$ and $T$) at the same positions as in the definitions of $X_{u_i}$ and $Y_{v_i}(i)$.
Notably, for $i\in [0\dd \Delta)$, if $X_{u_i} = Y_{v_i}(i)$, then \cref{eqn:symbolcondition} implies $X'_{u_i} = Y'_{v_i}(i)$ as well. 
Recall \cref{eqn:symbolcondition} holds with probability $\ge 1-\Delta^{-2}$ over the choice of $G$.
Therefore, by a union bound, we conclude that
\[\mathcal{A}_{r}(P_1,\dots,T_{d-1}) = \{i \in [0\dd \Delta) : X'_{u_i} = Y'_{v_i}(i)\} \supseteq  \{i \in [0\dd \Delta) : X_{u_i} = Y_{v_i}(i)\} = \mathcal{A}_{r}(P_1,\dots,T_{d}),\]
with probability at least $1 - \Delta^{-2}- \Delta^{-8} - \Delta^{-8} \ge 1 - \delta\Delta^{-1}$ (where we used $\Delta \ge 10\delta^{-2} >3\delta^{-1}$) over the random choice of $r=(G,p,B,F)$. This establishes \cref{eqn:coupling}, and finishes the proof of the lemma.
\end{proof}

The description of our main adaptive tester is given in \cref{alg:adaptive-tester}. In the following, let $\delta\coloneqq \eps/16 =  (320c\log^4 m )^{-1}$.
 \begin{algorithm}[h]
  \DontPrintSemicolon
\SetKwInput{KwInput}{Input}
\SetKwInput{KwOutput}{Output}
\SetKwComment{Comment}{$\triangleright$\ }{}
\For{iterations $d\gets 1,2,\dots, D\coloneqq \lceil 90\eps^{-1}\ln(1+\Delta)\rceil$}{
         Define $O_{d-1}^*$ as the output distribution of the following sampling process (which depends on $P_1,T_1,\dots,P_{d-1},T_{d-1}$):  \texttt{``Let $A\subseteq [0\dd \Delta)$ be returned by the randomized algorithm $\mathcal{A}(P_1,T_1,\dots,P_{d-1},T_{d-1})$ from \cref{lem:algomultiple} (with $\delta \coloneqq\eps/16$, and $k_0$ defined in \cref{eqn:defnk0}). If $A=\emptyset$, output $\bot$; otherwise, output a uniformly random element of  $A$.''}  \label{line:defnOd}\\
        Define $O_{d-1}$ as the distribution of $x\sim O^*_{d-1}$ conditioned on $x\neq \bot$.\label{line:defndistribO}\\
        \Comment{By definition, both $O_0^*$ and $O_0$ are the uniform distribution over $[0\dd \Delta)$}
        Draw $L = \lceil 5\eps^{-2}\ln n \rceil = \Oh(\log^{9}n)$ independent samples $x_1,\dots,x_{L}$ from $O^*_{d-1}$. (If the time spent at this line exceeds $10$ times its expected time upper bound from \cref{lem:algomultiple}, then abort this line and assign $x_1,\dots,x_L$ with all zeros.) \label{line:drawsamples}
        \\
        \lIf{$\bot \in \{x_1,\dots,x_L\}$}{\Return{\No}}
        Obtain a block $B$ by applying \cref{lem:findblock} to the distribution $O_{d-1}$, where the required $L$ samples from $O_{d-1}$ are supplied by $x_1,\dots,x_L$.\label{line:findblock}\\
        
        Let $P_d,T_d$ be the fragments $P_0,T_0$ defined in \cref{lem:eliminatestep} based on block $B$. \label{line:pdtd}
}
\Return{\Yes}
\caption{Adaptive property tester for \cref{thm:adaptive-tester}} \label{alg:adaptive-tester}
\end{algorithm}

We first observe that our tester is always correct in the \Yes case: 
\begin{proposition}
   If $\Occ(P,T)\neq\emptyset$, then \cref{alg:adaptive-tester} always returns \Yes.
 \label{prop:noonly}
\end{proposition}
\begin{proof}
It suffices to show that, if \cref{alg:adaptive-tester} returns \No, then $\Occ(P,T)=\emptyset$.
   Note that \cref{alg:adaptive-tester} returns \No only if the samples drawn at Line~\ref{line:drawsamples} include $\bot$, which means that some execution of $\mathcal{A}(P_1,T_1,\dots,P_{d-1},T_{d-1})$ returned $A=\emptyset$. By \cref{lem:algomultiple} (\cref{item:fingerprint-match-item1}), this indicates $\bigcap_{j=1}^{d-1} \Occ(P_j,T_j) = \emptyset$. 
   From Line~\ref{line:pdtd} we see that $P_j,T_j$ are fragments defined in \cref{lem:eliminatestep}, which satisfy $\Occ(P,T) \subseteq \Occ(P_j,T_j)$. Hence, $\Occ(P,T) \subseteq \bigcap_{j=1}^{d-1} \Occ(P_j,T_j) = \emptyset$.
\end{proof}

We use the following  potential function $\Phi_d$ to measure the progress of our algorithm towards eliminating all candidate answers. For $d\in [0\dd D]$, the potential function $\Phi_d$ is a quantity determined by the state of \cref{alg:adaptive-tester} at the end of the $d$-th iteration.

\begin{definition}[$\Phi_d$]
       \label{defn:potential-func}
    For $d\in \fragmentcc{0}{D}$, suppose \cref{alg:adaptive-tester} has finished the first $d$ iterations (without returning \No) and has chosen the fragments $P_1,T_1,\dots,P_{d},T_{d}$ at Line~\ref{line:pdtd}. Then, define the potential function $\Phi_d= \Phi_d(P_1,T_1,\dots,P_d,T_d)$ as
       \begin{equation}
       \Phi_d\coloneqq \mathbf{E}_{A \sim\mathcal{A}(P_1,T_1,\dots,P_{d},T_d)}[\ln (1+|A|)] \in [0, \ln(1+\Delta)],
       \end{equation}
       where $\mathcal{A}$ is the randomized algorithm from \cref{lem:algomultiple}.
In the case where \cref{alg:adaptive-tester} has returned \No by the end of $d$-th iteration (i.e., $P_d,T_d$ are undefined), we artificially define $\Phi_d = -d$ for convenience. \lipicsEnd
\end{definition}
We first observe that $\Phi_d$ is almost non-increasing:
\begin{lemma}
    \label{lem:potential-non-increase}
    For $d\in \fragmentcc{1}{D}$, $\Phi_d \le \Phi_{d-1} + \eps/16$ always holds.
\end{lemma}
\begin{proof}
    If \cref{alg:adaptive-tester} has returned \No by the end of $d$-th iteration, we have $\Phi_d=-d$ and $\Phi_{d-1}\ge \min(0,-(d-1))=-d+1$, so the claim holds. Now assume \cref{alg:adaptive-tester} has finished the first $d$ iterations without returning \No.
    Denote $A^r_{d} = \mathcal{A}_r(P_1,T_1,\dots,P_{d},T_d)$ for short, where $r$ is the randomness used by the algorithm $\mathcal{A}$.
   By \cref{eqn:coupling} in \cref{lem:algomultiple}, $\Pr_r[A^r_{d-1} \supseteq A^r_{d}]\ge 1-\delta/\Delta$. Then,
  \begin{align*}
   \Phi_{d-1}-\Phi_{d} &= \mathbf{E}_r[\ln(1+|A_{d-1}^r|)] - \mathbf{E}_r[\ln(1+|A_{d}^r|)]\\
   &\ge \mathbf{E}_r\big [\ln(1+|A_{d-1}^r|)\cdot  \mathbf{1}[A_{d-1}^r \supseteq A_d^r]\big] - \mathbf{E}_r[\ln(1+|A_{d}^r|)]\\
   & = \mathbf{E}_r\big [\ln\big (\tfrac{1+|A_{d-1}^r|}{1+|A_d^r|}\big )\cdot  \mathbf{1}[A_{d-1}^r \supseteq A_d^r]\big] - \mathbf{E}_r\big[\ln(1+|A_{d}^r|)\cdot  \mathbf{1}[A_{d-1}^r \not \supseteq A_d^r]\big]\\
   & \ge 0 - \ln(1+\Delta)\Pr_r[A_{d-1}^r \not \supseteq A_d^r]\\
   &\ge - \tfrac{\delta \ln(1+\Delta)}{\Delta} \ge -\delta = -\eps/16,
  \end{align*} 
  which finishes the proof.
\end{proof}

The following lemma shows that the potential decreases rapidly in the \No case: 
\begin{lemma}
\label{lem:terminate-fast}
    Suppose $\Occ_k(P,T) = \emptyset$. 
Let $d\in \fragmentcc{1}{D}$. Then, with at least $1/2$ probability over the randomness in iteration $d$, it holds that $\Phi_{d-1}-\Phi_d \ge \eps/8$.\lipicsEnd
\end{lemma}

Before proving \cref{lem:terminate-fast}, we first  use it to prove the main theorem:
\begin{proof}[Proof of \cref{thm:adaptive-tester} assuming \cref{lem:terminate-fast}]
    We first show the correctness of \cref{alg:adaptive-tester}. In the \Yes case where $\Occ(P,T)\neq \emptyset$, \cref{alg:adaptive-tester} is correct by \cref{prop:noonly}.
    Now we consider the \No case where $\Occ_k(P,T)=\emptyset$.
The initial potential is $\Phi_0 \le  \ln(1+\Delta)$.
    In iteration $d\in \fragmentcc{1}{D}$, by \cref{lem:terminate-fast}, with at least $1/2$ probability we have potential decrement $\Phi_{d-1}-\Phi_d \ge \eps/8$. 
    By Chernoff bound, among all  $D= \lceil 90\eps^{-1}\ln(1+\Delta)\rceil$ iterations, with at least $1-\exp(-2\cdot D\cdot 0.1^2) \ge 1-\exp(-32c\log^4 m\ln(1+\Delta))\ge 1-n^{-10}$ probability, at least $0.4D$ iterations have potential decrement $\ge \eps/8$. In this case, by \cref{lem:potential-non-increase}, the total potential decrement is $\Phi_0-\Phi_D \ge 0.4D\cdot (\eps/8) + (D-0.4D) \cdot (-\eps/16) = D\eps/80 \ge \frac{9}{8}\ln(1+\Delta)>\Phi_0$, so $\Phi_D<0$, which indicates our algorithm has returned \No by \cref{defn:potential-func}.
  Therefore, we have showed that \cref{alg:adaptive-tester} is correct with high probability in the \No case, finishing the proof of correctness.

  It remains to analyze the time complexity.
The time complexity of \cref{alg:adaptive-tester} is dominated by the total time for drawing samples $x_1,\dots,x_L$ from $O_{d-1}^*$ at Line~\ref{line:drawsamples}, as well as the extra time complexity of \cref{lem:findblock} at Line~\ref{line:findblock}.
By the definition at Line~\ref{line:defnOd}, each sample from $O_{d-1}^*$ is drawn by running the algorithm $\mathcal{A}$ of \cref{lem:algomultiple} with $\delta = \eps/16$, $D =\lceil 90\eps^{-1}\ln(1+\Delta)\rceil$, and $k_0+1 = \Omega(\eps k\Delta/n)$ (by \cref{eqn:defnk0}), where $\eps^{-1} = \Theta(\log^4 m) $ (by \cref{eqn:defneps}).
Then, by \cref{lem:algomultiple}, drawing one sample from $O_{d-1}^*$ takes expected time 
$\Oh\left ( \frac{\Delta}{\sqrt{k_0+1}}\cdot \left (D\delta^{-0.5}\log^{5.5}(2\Delta) + D^{1.1}\delta^{-1}\right )\right ) = \Oh\left ( \sqrt{\frac{\Delta n}{k}}\cdot \log^{14.5}n\right ) $.
Over all $D$ iterations, we draw at most $DL = \Oh(D \log^9 n) = \Oh(\log^{14}n) $ samples, in total time $\Oh\left ( \sqrt{\frac{\Delta n}{k}}\cdot \log^{28.5}n\right ) $ (this is worst-case time rather than expected time, due to the abortion at Line~\ref{line:drawsamples}).
The extra time complexity of \cref{lem:findblock} is $\Oh(\frac{n}{k}\cdot \log^9 n)$ per execution, so the total time over all $D$ iterations is $\Oh(D\cdot \frac{n}{k}\cdot \log^9 n) = \Oh(\frac{n}{k}\log^{14} n)$.
In summary, the total time complexity of \cref{alg:adaptive-tester} is 
$ \Oh\left (\sqrt{\frac{\Delta n}{k}} \log^{28.5} n + \frac{n}{k} \log^{14} n\right )$ as claimed. 
\end{proof}

It remains to prove \cref{lem:terminate-fast}:
\begin{proof}[Proof of \cref{lem:terminate-fast}]
    First, if \cref{alg:adaptive-tester} already returned \No in any of the first $d-1$ iterations,  then the desired potential decrement automatically holds since
    $\Phi_{d-1}-\Phi_d = -(d-1) -(-d)= 1>\eps/8$.
Thus, in the following, we assume \cref{alg:adaptive-tester} has finished the first $d-1$ iterations without returning \No. Now we focus on the $d$-th iteration. For the same reason, if \cref{alg:adaptive-tester} returns \No in the $d$-th iteration, then we also have the desired potential decrement $\Phi_{d-1}-\Phi_d \ge 0 - (-d)>\eps/8$.
    
  If $\Pr_{x\sim O_{d-1}^*}[x \neq \bot] < 1/2$, then,  at Line~\ref{line:drawsamples}, the $L = \lceil 5\eps^{-2}\ln n \rceil$ samples drawn from $O_{d-1}^*$ include $\bot$ with at least $1-2^{-L}>1/2$ probability, and thus we return \No in the $d$-th iteration, achieving the desired potential decrement. 
  Hence, in the following, we assume
$\Pr_{x\sim O_{d-1}^*}[x \neq \bot] \ge 1/2$.

    Denote $A^r_{d-1} = \mathcal{A}_r(P_1,T_1,\dots,P_{d-1},T_{d-1})$ for short, where $r$ is the randomness used by the algorithm $\mathcal{A}$.
   By definition of $O^*_{d-1}$ and the assumption that $\Pr_{x\sim O_{d-1}^*}[x \neq \bot] \ge 1/2$, we have 
   \begin{equation}
   \label{eqn:premptyset}
   \Pr_{r}[A_{d-1}^r\neq \emptyset] \ge 1/2.
   \end{equation}

At Line~\ref{line:drawsamples}, we draw $L$ samples
$x_1,\dots,x_L$ from $O^*_{d-1}$, where the failure probability due to abortion is at most $0.1$, by Markov's inequality.
 If the samples $x_1,\dots,x_L$ contain $\bot$,
  then we return \No and achieve the desired potential decrement. 
   Otherwise, we have successfully drawn $L$ independent samples from $x\sim O^*_{d-1}$ conditioned on $x\neq \bot$, which are equivalent to $L$ independent samples from the distribution $O_{d-1}$ defined at Line~\ref{line:defndistribO} (note that $O_{d-1}$ is supported on $[0\dd \Delta)$). Then, by \cref{lem:findblock}, at Line~\ref{line:findblock} we find a good block which contains at least $\frac{k\Delta}{64n}$ almost good positions with respect to $O_{d-1}$, with at least $0.6$ success probability. Assuming this block is indeed good, by \cref{lem:eliminatestep}, at Line~\ref{line:pdtd} we obtain fragments
$P_d,T_d$ of lengths $|P_d|=2\Delta,|T_d|=3\Delta-1$ which satisfy $\Pr_{x\sim O_{d-1}} [x \in B_d  ] \ge \eps$,
where we denote 
\[B_d:= [0\dd \Delta)\setminus  \Occ_{k_0}(P_d,T_d)\]
for short; by the definition of $O_{d-1}$, this means
\begin{equation}
    \mathbf{E}_{r} \left [\frac{|A_{d-1}^r \cap B_d|}{ |A_{d-1}^r|} \, \Bigg \vert \, A_{d-1}^r \neq \emptyset \right ] \ge \eps.
    \label{eqn:temptemp}
\end{equation}
We use \cref{eqn:premptyset} to derive the following slight variant of \cref{eqn:temptemp}:
   \begin{align}
    \mathbf{E}_{r} \left [\frac{|A_{d-1}^r \cap B_d|}{ 1+|A_{d-1}^r|} \right ] &=     \mathbf{E}_{r}\left [\frac{|A_{d-1}^r \cap B_d|}{ 1+|A_{d-1}^r|} \, \Bigg \vert \, A_{d-1}^r \neq \emptyset\right ]\cdot \Pr_r[A_{d-1}^r \neq \emptyset  ] +  0\cdot \Pr_r [A_{d-1}^r = \emptyset  ]\nonumber \\
    & \ge \frac{1}{2}\mathbf{E}_{r}\left [\frac{|A_{d-1}^r \cap B_d|}{ 1+|A_{d-1}^r|} \, \Bigg \vert \, A_{d-1}^r \neq \emptyset\right ] \tag{by \cref{eqn:premptyset}}\\
    & \ge\frac{1}{2}\mathbf{E}_{r}\left [\frac{|A_{d-1}^r \cap B_d|}{2|A_{d-1}^r|} \, \Bigg \vert \, A_{d-1}^r \neq \emptyset\right ]\nonumber \\
      & \ge \eps/4, \label{eqn:expectkill}
   \end{align}
   where the last step used \cref{eqn:temptemp}.
   
Now we consider the behavior of $A^r_{d} = \mathcal{A}_r(P_1,T_1,\dots,P_{d},T_{d})$, where $r$ is the randomness used by the algorithm $\mathcal{A}$.
By \cref{lem:algomultiple} (\cref{item:multiple-condition2}), with at least $1-\delta$ probability, it holds that
$   A^r_d \subseteq \bigcap_{j=1}^d \Occ_{k_0}(P_j,T_j) \subseteq \Occ_{k_0}(P_d,T_d) = [0\dd \Delta)\setminus B_d$; in other words,
   \begin{equation}
   \label{eqn:nonintersect}
   \Pr_{r}[A_{d}^r \cap B_d \neq \emptyset]\le \delta. 
   \end{equation}

 By \cref{eqn:coupling} in \cref{lem:algomultiple}, we have
   \begin{equation}
       \label{eqn:coupleprob}
   \Pr_{r}[A_{d-1}^r\not \supseteq A_d^r] \le \delta/\Delta.
   \end{equation}

   Now we are ready to bound the potential decrement: 
   \begin{align*}
      \Phi_{d-1} - \Phi_d &=   \mathbf{E}_r[\ln(1+|A_{d-1}^r|)] - \mathbf{E}_r[\ln(1+|A_{d}^r|)] \\
      & = - \mathbf{E}_r \left [\ln \left (1 - \frac{|A_{d-1}^r| - |A_{d}^r|}{1+|A_{d-1}^r|}\right )\right ]\\
      & \ge \mathbf{E}_r \left [\frac{|A_{d-1}^r| - |A_{d}^r|}{1+|A_{d-1}^r|}\right ].\tag{by $-\ln(1-x)\ge x$}
   \end{align*}
   We rewrite the numerator using
  $|A_{d-1}^r|= |A_{d-1}^r \setminus B_d|+|A_{d-1}^r \cap B_d|$, and $|A_{d}^r|= |A_{d}^r \setminus Y|+|A_{d}^r \cap Y|$ with $Y = B_d\cap A_{d-1}^r$. After rearranging, we get
  \begin{equation}
      \Phi_{d-1} - \Phi_d \ge \mathbf{E}_r\left [\frac{|A_{d-1}^r\setminus B_d| - |A_{d}^r\setminus (B_d\cap A_{d-1}^r)| }{1+|A_{d-1}^r|}\right ] + \mathbf{E}_r\left [\frac{|A_{d-1}^r\cap B_d|}{1+|A_{d-1}^r|}\right ] - \mathbf{E}_r\left [\frac{|A_{d}^r\cap B_d\cap A_{d-1}^r| }{1+|A_{d-1}^r|}\right ]. \label{eqn:boundthis}
  \end{equation}

   The second term of \cref{eqn:boundthis} is at least $\eps/4$ due to \cref{eqn:expectkill}.
   The third term of \cref{eqn:boundthis} can be bounded as
   \begin{align*}
\mathbf{E}_r\left [\frac{|A_{d}^r\cap B_d\cap A_{d-1}^r| }{1+|A_{d-1}^r|}\right ] \le \mathbf{E}_r\left [\mathbf{1}[A_d^r \cap B_d\neq \emptyset]\cdot \frac{| A_{d-1}^r| }{1+|A_{d-1}^r|} \right ]
      \le \Pr_{r}[A_d^r \cap B_d\neq \emptyset] \le  \delta, 
   \end{align*}
where the last step follows from  \cref{eqn:nonintersect}.
  
   To bound the first term of \cref{eqn:boundthis}, observe that if $A_{d-1}^r \supseteq A_d^r$, then the numerator is 
   \[
 |A_{d-1}^r\setminus B_d| - |A_{d}^r\setminus (B_d\cap A_{d-1}^r)| = |A_{d-1}^r\setminus B_d| - |A_{d}^r\setminus B_d| = |(A_{d-1}^r \setminus A_d^r) \setminus B_d| \ge 0.\]
 Hence, the first term of \cref{eqn:boundthis} is
 \begin{align*}
\mathbf{E}_r\left [\frac{|A_{d-1}^r\setminus B_d| - |A_{d}^r\setminus (B_d\cap A_{d-1}^r)| }{1+|A_{d-1}^r|} \right ] &\ge \mathbf{E}_r\left [\frac{|A_{d-1}^r\setminus B_d| - |A_{d}^r\setminus (B_d\cap A_{d-1}^r)| }{1+|A_{d-1}^r|} \cdot \mathbf{1}[A_{d-1}^r \not \supseteq A_d^r]\right ] \\
& \ge  \mathbf{E}_r\left [-\Delta \cdot \mathbf{1}[A_{d-1}^r \not \supseteq A_d^r]\right ]\\
& \ge  -\Delta \cdot \delta/\Delta = -\delta,
 \end{align*}
where the last step follows from \cref{eqn:coupleprob}.

 Summing up the estimates for the three terms of \cref{eqn:boundthis}, we get
 \[ \Phi_{d-1}-\Phi_d \ge -\delta + \eps/4 - \delta = \eps/8\]
 as desired. This potential decrement holds under the condition that the sampling at Line~\ref{line:drawsamples} did not abort, and the block found by \cref{lem:findblock} at Line~\ref{line:findblock} is good. By a union bound, this holds with probability at least $0.6-0.1=1/2$.
\end{proof}

\section{Adaptive Lower Bound}\label{sec:lb_adaptive}
In this section, we prove the following theorem:
\thmlbadaptive*

We henceforth fix integers $1 \le k \le m \le n$ such that $4\ln(5n) < k \le \frac{m}{4}$ and denote $\Delta \coloneqq n-m+1$ (the case of $k>\frac{m}{4}$ will be handled at the end of this section).

Recall that, for a real parameter $p\in [0,1]$, the \emph{Bernoulli} distribution $\Ber(p)$ is defined so that $\Pr[X=1]=p$ and $\Pr[X=0] = 1-p$ if $X\sim \Ber(p)$. 
We define the following distributions on $(P,T)\in \{0,1\}^{m}\times \{0,1\}^n$ so that $P[i]\sim \Ber(\frac{2k}{m})$ for $i\in \fragmentco{0}{m}$ and $T[j] \sim \Ber(\frac{2k}{m})$ for $j\in \fragmentco{0}{n}$:
\begin{itemize}
    \item $\Random$, where the variables $\{P[i] : i\in \fragmentco{0}{m}\}$ and $\{T[j] : j\in \fragmentco{0}{n}\}$ are all independent;
    \item $\Planted_t$ for $t\in \fragmentco{0}{\Delta}$, where the variables $\{T[j] : j\in \fragmentco{0}{n}\}$ are all independent and $P=T[t\dd t+m)$.
    \item $\Planted$, which is a mixture of distributions $\Planted_t$ for $t\in \fragmentco{0}{\Delta}$ with weights $\frac{1}{\Delta}$ each.
    \item $\Mixed$, which is a mixture of distributions $\Random$ and $\Planted$ with weights $\frac12$ each.
\end{itemize}

By construction, the only correct answer for every instance drawn according to $\Planted$ is \textsc{Yes}.

\begin{observation}\label{obs:planted}
If $(P,T)\sim \Planted$, then $\Occ(P,T)\ne \emptyset$.\lipicsEnd
\end{observation}

Moreover, the only correct answer for most instances drawn according to $\Random$ is \textsc{No}.

\begin{lemma}\label{lem:random}
If $(P,T)\sim \Random$, then $\Pr[\Occ_k(P,T)=\emptyset]> 0.8$.
\end{lemma}
\begin{proof}
Let us fix $t\in \fragmentco{0}{\Delta}$. 
Due to the assumption $k \le \frac{m}{4}$, for each $i\in \fragmentco{0}{m}$, we have 
\[\Pr[P[i]\ne T[i+t]] = 2\cdot \tfrac{2k}{m}\cdot (1-\tfrac{2k}{m}) \ge \tfrac{2k}{m},\]
and these events are independent across $i\in \fragmentco{0}{m}$.
Hence $\Exp[\HD(P, T\fragmentco{t}{t+m})]\ge 2k$ and, by the Chernoff bound, 
\[\Pr[\HD(P, T\fragmentco{t}{t+m}) \le k] \le \exp\left(-\tfrac12 \cdot (\tfrac{1}{2})^2 \cdot 2k\right) = \exp(-\tfrac{k}{4}) < \tfrac{1}{5n}.\]
Here, the last inequality is due to $k \ge 4\ln(5n)$.
By the union bound, $\Pr[\Occ_k(P,T)\ne\emptyset] < 0.2$.
\end{proof}

A combination of \cref{obs:planted} and \cref{lem:random} implies that any solution to \cref{pr:tester} must distinguish input distributions $\Planted$ and $\Random$.

We henceforth fix a \emph{deterministic} algorithm $A$ that makes exactly $q$ distinct queries (we append dummy queries when necessary to make the number of queries uniform across all executions).
We define a \emph{run} as the sequence of queries the algorithm makes and the answers it receives. 
Formally, a run is a sequence of $q$ triples $(S,i,a)$ such that $S\in \{P,T\}$, $i\in \fragmentco{0}{|S|}$, and $a=S[i]$.
We denote by $\Runs$ the family of all \emph{admissible} runs, that is, runs describing the execution of $A$ on some input.
Every input distribution $D$ yields a distribution on $\Runs$ so that $\Pr_D[r]$ is the probability that the execution of $A$ is described by $r$ provided that the input is drawn according~to~$D$.

The following lemma says that, unless a mismatch between $P$ and $T\fragmentco{t}{t+m}$ is discovered, it is more likely that the input comes from $\Planted_t$ than from $\Random$. 
The intuitive interpretation of the immediate corollary is that many such mismatches are needed across all $t\in \fragmentco{0}{\Delta}$ in order to confirm the hypothesis that the input comes from $\Random$ rather than $\Planted$.  
\begin{lemma}\label{lem:mismatch}
Consider a run $r\in \Runs$ and a position $t\in \fragmentco{0}{\Delta}$. 
If $r$ contains triples $(P,i,a)$ and $(T,i+t,b)$ for some $i\in \fragmentco{0}{m}$ and $a\ne b$, then $\Pr_{\Planted_t}[r]=0$.
Otherwise, $\Pr_{\Planted_t}[r]\ge \Pr_{\Random}[r]$.
\end{lemma}
\begin{proof}
Since $A$ is deterministic and $r$ is admissible, we have $\Pr_D[r] = \Pr_D[\bigwedge_{(S,i,a)\in r} S[i]=a]$.
By definition of $\Planted_t$, we have $\Pr_{\Planted_t}[P[i]=a \wedge T[i+t]=b]=0$ for every $i\in \fragmentco{0}{m}$ and $a\ne b$.
On the other hand, \[\Pr_{\Planted_t}[P[i]=0 \wedge T[i+t]=0]=(1-\tfrac{2k}{m}) \ge (1-\tfrac{2k}{m})^2 = \Pr_{\Random}[P[i]=0 \wedge T[i+t]=0]\]
and 
\[\Pr_{\Planted_t}[P[i]=1 \wedge T[i+t]=1]=\tfrac{2k}{m} \ge (\tfrac{2k}{m})^2 = \Pr_{\Random}[P[i]=1 \wedge T[i+t]=1].\]
The remaining characters are independent and distributed identically under both $\Planted_t$ and $\Random$, and thus they do not affect the ratio between $\Pr_{\Planted_t}[r]$ and $\Pr_{\Random}[r]$.
\end{proof}
\begin{corollary}\label{cor:mismatch}
Every run $r\in \Runs$ satisfies
\[\tfrac{\Pr_{\Planted}[r]}{\Pr_{\Random}[r] } \ge 1-\tfrac{1}{\Delta}\left(|\{i: (P,i,0)\in r\}|\cdot |\{j: (T,j,1)\in r\}|+|\{i: (P,i,1)\in r\}|\cdot |\{j: (T,j,0)\in r\}|\right).\]
\end{corollary}
\begin{proof}
Recall that $\Pr_{\Planted}[r] =\tfrac{1}{\Delta}\sum_{t=0}^{\Delta-1} \Pr_{\Planted_t}[r]$.
By \cref{lem:mismatch}, $\Pr_{\Planted_t}[r] \ge \Pr_{\Random}[r]$ holds unless $(P,i,a),(T,i+t,b)$ belong to $r$ for some $i\in \fragmentco{0}{m}$ and $a\ne b$. The number of such indices $t\in \fragmentco{0}{\Delta}$ is upper-bounded by the number of pairs $(P,i,a),(T,j,b)\in r$ with $a\ne b$.
\end{proof}

In the next lemma, we consider the set $\Bad\subseteq \Runs$ of runs which are significantly more likely under $\Random$ than under $\Planted$. 
Intuitively, these are the only runs for which the $\textsc{No}$ answer does not contribute too much to the overall error probability.
Our strategy is to show that, with insufficiently many queries and $\Random$ input distribution, the algorithm is unlikely to encounter a run in $\Bad$.
\begin{lemma}\label{lem:bad}
Let $\Bad=\{r \in \Runs : \Pr_\Planted[r] < \tfrac12 \Pr_\Random[r]\}$. If $q < \tfrac{1}{55}\left(\sqrt{\tfrac{m\Delta}{k}}+\tfrac{n}{k}\right)$, then $\Pr_\Random[\Bad] \le 0.4$.
\end{lemma}
\begin{proof}
Define $\Bad'\subseteq \Runs$ be the set of admissible runs that read at least $\frac{10kq}{m}$ ones in total or at least $10k$ ones in $P$.
If the input distribution is $\Random$, then $A$ makes $q$ queries and each answer is distributed according to $\Ber(\frac{2k}{m})$ independently from the previous answers. 
Thus, the expected number of ones read is $\frac{2kq}{m}$ and, by Markov's inequality, the probability that $A$ reads at least $\frac{10kq}{m}$ ones does not exceed $0.2$.
Similarly, the expected total number of ones in $P$ is $2k$, and the probability that there are at least $10k$ ones in $P$ does not exceed $0.2$.
The union bound implies $\Pr_\Random[\Bad']\le 0.4$.

Consequently, it suffices to prove that the upper bound on $q$ implies $\Bad\subseteq \Bad'$.
For a proof by contradiction, we fix a run $r\in \Bad \setminus \Bad'$ that reads $q_P$ characters from $P$, including $q'_P$ ones, and $q_T$ characters from $T$, including $q'_T$ ones.
By \cref{cor:mismatch}, the assumption $r\in \Bad$ implies 
\begin{equation}\label{eq:mismatch}
\frac{1}{2} \ge \frac{\Pr_{\Planted}[r]}{\Pr_{\Random}[r] } \ge 1-\frac{(q_P-q'_P)\cdot q'_T + q'_P \cdot (q_T-q'_T)}{\Delta} \ge 1-\frac{q_P\cdot q'_T + q'_P \cdot q_T}{\Delta}.\end{equation}
Denote $q'\coloneqq q'_P+q'_T$ and recall that $q=q_P+q_T$. Now, \eqref{eq:mismatch} implies
\begin{equation}\label{eq:qqprim}
    q\cdot q' \ge q_P\cdot q'_T + q'_P \cdot q_T \ge \tfrac{\Delta}{2}.
\end{equation}
The assumption $r\notin \Bad'$ yields $q' \le \frac{10kq}{m}$, and thus \eqref{eq:qqprim} implies
\begin{equation}\label{eq:qsquared}
    q \cdot \tfrac{10kq}{m} \ge \tfrac{\Delta}{2},\qquad\text{that is,}\qquad q  \ge \sqrt{\tfrac{m\Delta}{20k}}.
\end{equation}
Equation \eqref{eq:qqprim} also implies $q'>0$. Since $q'$ is an integer and  $q' \le \frac{10kq}{m}$, we derive
\begin{equation}\label{eq:qmk}
    q  \ge \tfrac{m}{10k}.
\end{equation}
The assumption $r\notin \Bad'$ yields $q'_P \le 10k$ and, due to $q_P \le m$, $q'_T \le q' \le \frac{10kq}{m}$, $q_T\le q$, \eqref{eq:qqprim} also implies
\begin{equation}\label{eq:20kq}
    20kq = m \cdot \tfrac{10kq}{m} + 10k\cdot q\ge q_P\cdot q'_T + q'_P \cdot q_T \ge \tfrac{\Delta}{2},\qquad\text{that is,}\qquad q \ge \tfrac{\Delta}{40k}.
\end{equation}
A convex combination of \eqref{eq:qmk} and \eqref{eq:20kq} yields
\begin{equation}\label{eq:nk}
    q \ge \tfrac15 \cdot \tfrac{m}{10k}+\tfrac45 \cdot \tfrac{\Delta}{40k} = \tfrac{n+1}{50k} > \tfrac{n}{50k}.
\end{equation}
With a convex combination of \eqref{eq:qsquared} and \eqref{eq:nk}, we obtain 
\begin{equation}
    q \ge \tfrac1{11} \cdot \tfrac{1}{\sqrt{20}}\cdot \sqrt{\tfrac{m\Delta}{k}} + \tfrac{10}{11} \cdot \tfrac{1}{50}\cdot \tfrac{n}{k}\ge \tfrac1{55}\left(\sqrt{\tfrac{m\Delta}{k}}+\tfrac{n}{k}\right),
\end{equation}
which is the desired contradiction.
\end{proof}

We are now ready to prove that the algorithm $A$ needs to make many queries to achieve low error probability on $\Mixed$.

\begin{proposition}\label{prp:lb_adaptive}
    Every deterministic algorithm solving \cref{pr:tester} correctly with probability at least $0.9$ for the input distribution $\Mixed$ makes at least $\tfrac{1}{55}\left(\sqrt{\tfrac{m\Delta}{k}}+\tfrac{n}{k}\right)$ queries to the strings~$P$~and~$T$.
\end{proposition}
\begin{proof}
Fix a deterministic algorithm $A$ with error probabilities $e_{R}$ and $e_{P}$ under $\Random$ and $\Planted$, respectively.
For a proof by contradiction, suppose that $A$ makes $q < \tfrac{1}{55}\left(\sqrt{\tfrac{m\Delta}{k}}+\tfrac{n}{k}\right)$ queries.

Moreover, let $\RunsNo\subseteq \Runs$ be the set of runs on which $A$ returns \textsc{No}.
By \cref{lem:random}, the answer \textsc{No} is required on $\Random$ with probability strictly larger than $0.8$,
so $\Pr_\Random[\RunsNo] > 0.8-e_R$.
By \cref{obs:planted}, all \textsc{No} answers are incorrect under $\Planted$, so $e_P\ge \Pr_\Planted[\RunsNo]$.

Due to \cref{lem:bad}, we further have $\Pr_\Random[\RunsNo\setminus \Bad] > 0.8-e_R - 0.4 = 0.4-e_R$.
By definition of~$\Bad$, this means that
\[e_P \ge \Pr_{\Planted}[\RunsNo]\ge \Pr_{\Planted}[\RunsNo\setminus \Bad] \ge \tfrac12 \Pr_\Random[\RunsNo\setminus \Bad]> 0.2-\tfrac{e_R}{2} \ge 0.2-e_R.\]
Hence, the error probability of $A$ on $\Mixed$, which equals $\tfrac12(e_P+e_R)$, is strictly larger than $0.1$.
\end{proof}

\cref{thm:lb_adaptive} restricted to $k \le \frac{m}{4}$ follows from \cref{prp:lb_adaptive} by Yao's principle.
For $\frac{m}{4} < k < m$, we derive \cref{thm:lb_adaptive} from the following result proved in \cref{sec:lb_large} using a similar approach:
\begin{restatable}{theorem}{thmlbadaptivelargek}\label{thm:lb_adaptive-largek}
    Fix integer parameters $1 \le m \le n$. 
    Every algorithm solving \cref{pr:tester} with $k=m-1$ on strings from the alphabet $\Sigma = [1\dd 10n^2]$ correctly with probability at least $0.8$ makes at least $\frac1{6}\left(\sqrt{n-m+1}+\frac{n}{m}\right)$ queries to the input strings $P$ and $T$.\lipicsEnd
\end{restatable}

We note that every algorithm solving \cref{pr:tester} for any $k\in \fragmentco{1}{m}$ also solves the problem for $k=m-1$.
Thus, \cref{thm:lb_adaptive-largek} yields a lower bound of $\frac1{6}\left(\sqrt{n-m+1}+\frac{n}{m}\right)$.
For $k > \frac{m}{4}$, we have 
\[\tfrac{1}{55}\!\left(\sqrt{\tfrac{m(n-m+1)}{k}}+\tfrac{n}{k}\right) < \tfrac{1}{55}\!\left(2\sqrt{n-m+1}+\tfrac{4n}{m}\right) < \tfrac{4}{55}\left(\sqrt{n-m+1}+\tfrac{n}{m}\right) < \tfrac{1}{6}\left(\sqrt{n-m+1}+\tfrac{n}{m}\right),\]
and thus \cref{thm:lb_adaptive} follows.
\section{Non-Adaptive Lower Bound}\label{sec:lbnonadaptive}
In this section, we prove the following theorem:
\thmlbnonadaptive*

We henceforth fix integers $1 \le k \le m \le n$ such that $4\ln(5n) < k \le \frac{m}{4}$   (the case of $k>\frac{m}{4}$ will be handled at the end of this section) and denote 
\[\Delta \coloneqq n-m+1\qquad\text{and}\qquad s\coloneqq\min(m,\max(4k,\Delta)).\]
Define a distribution $\Sparse$ on $S\in \{0,1\}^s$ so that $S[i] \sim \Ber(\frac{2k}{s})$ and the variables $\{S[i] : i\in \fragmentco{0}{s}\}$ are independent.
For integer parameters $p\in \fragmentcc{0}{m-s}$ and $t\in \fragmentco{0}{\Delta}$, define three distributions on $(P,T)\in \{0,1\}^{m}\times \{0,1\}^n$, where $P=0^p \cdot S_P \cdot 0^{m-s-p}$ and $T=0^{p+t} \cdot S_T \cdot 0^{n-s-p-t}$:
\begin{itemize}
    \item $\Positive_{p,t}$, where $S_P=S_T\sim \Sparse$;
    \item $\Negative_{p,t}$, where $S_P,S_T\sim \Sparse$ are independent;
    \item $\Hybrid_{p,t}$ which is the mixture of distributions $\Positive_{p,t}$ and $\Negative_{p,t}$ with equal weights $\frac12$.
\end{itemize}
Moreover, let $\Hybrid_p$ be the mixture of distributions $\Hybrid_{p,t}$ for $t\in \fragmentco{0}{\Delta}$ with equal weights $\frac{1}{\Delta}$, and let $\Hybrid$ be the mixture of distributions $\Hybrid_{p}$ for $p\in \fragmentcc{0}{m-s}$ with equal weights $\frac{1}{m-s+1}$.

By construction, the only correct answer for every instance drawn according to $\Positive_{p,t}$ is \textsc{Yes}.
   
\begin{observation}\label{obs:positive}
If $(P,T)\sim \Positive_{p,t}$ for $p\in \fragmentcc{0}{m-s}$ and $t\in \fragmentco{0}{\Delta}$, then $\Occ(P,T)\ne \emptyset$.\lipicsEnd
\end{observation}

Moreover, the only correct answer for most instances drawn according to $\Negative_{p,t}$ is \textsc{No}.

\begin{lemma}\label{lem:negative}
If $(P,T)\sim \Negative_{p,t}$ for $p\in \fragmentcc{0}{m-s}$ and $t\in \fragmentco{0}{\Delta}$, then \[\Pr[\Occ_k(P,T)=\emptyset] > 0.8.\]
\end{lemma}
\begin{proof}
    Let us fix $u \in \fragmentco{0}{\Delta}$.
    Observe that the characters $P[i]$ for $i\in \fragmentco{p}{p+s}$ are independent and distributed according to $\Ber(\frac{2k}{s})$.
    The characters $T[i+u]$ for $i\in \fragmentco{p}{p+s}$ are also independent and distributed according to $\Ber(\frac{2k}{s})$ or $\Ber(0)$, depending on whether $i+u\in \fragmentco{p+t}{p+t+s}$ or not.
    Thus, the events $P[i]\ne T[i+u]$ are independent across $i\in \fragmentco{p}{p+s}$.
    The probability of each such event is $\frac{2k}{s}$ (if $T[i+u]\sim \Ber(0)$)
    or $2\cdot \frac{2k}{s}\cdot (1-\frac{2k}{s}) \ge \frac{2k}{s}$ (otherwise, the inequality is due to $s \ge 4k$).
    In either case, the probability is at least $\frac{2k}{s}$ and hence $\Exp[\HD(P\fragmentco{p}{p+s},T\fragmentco{p+u}{p+u+s})] \ge 2k$.
    By the Chernoff bound,
    \[\Pr[\HD(P\fragmentco{p}{p+s},T\fragmentco{p+s}{p+s+s}) \le k] \le \exp\left(-\tfrac12 \cdot\! (\tfrac{1}{2})^2 \!\cdot 2k\right) = \exp(-\tfrac{k}{4}) < \tfrac{1}{5n},\]
    where the last inequality is due to $k > 4\ln(5n)$.
    Hence, $\Pr[\HD(P, T\fragmentco{u}{u+m})\le k] < \tfrac{1}{5n}$  and, by the union bound, $\Pr[\Occ_k(P,T)\ne \emptyset] < 0.2$.
\end{proof}

We henceforth fix a \emph{deterministic} \emph{non-adaptive} algorithm $A$ that makes exactly $q$ distinct queries.
Recall from the previous section that the execution of $A$ on some input can be described by a \emph{run}: a sequence of $q$ triples $(S,i,a)$ such that $S\in \{P,T\}$, $i\in \fragmentco{0}{|S|}$, $a=S[i]$. 
Moreover, $\Runs$ denotes the family of all \emph{admissible} runs.
Since $A$ is non-adaptive, it can be described by two sets $R_P\subseteq \fragmentco{0}{m}$ and $R_T\subseteq \fragmentco{0}{n}$ of $|R_P|+|R_T|\le q$ positions that it reads.
Thus, for each input $(P,T)\in \{0,1\}^m\times \{0,1\}^n$, the corresponding run is 
$r = \{(P,i,P[i]) : i\in R_P\}\cup \{(T,j,T[j]) : j\in R_T\}$.
Every input distribution $D$ yields a distribution on $\Runs$ so that $\Pr_D[r]$ is the probability that the execution of $A$ is described by $r$ provided that the input is drawn according~to~$D$.

The following lemma formalizes the intuition that, unless a mismatch between $P$ and $T\fragmentco{t}{t+m}$ is discovered, it is more likely that the input comes from $\Positive_{p,t}$ than from $\Negative_{p,t}$. 
\begin{lemma}\label{lem:mismatch2}
    Consider a run $r\in \Runs$ as well as positions $t\in \fragmentco{0}{\Delta}$ and $p\in \fragmentcc{0}{m-s}$.
    If $r$ contains triples $(P,i,a)$ and $(T,i+t,b)$ for some $i\in \fragmentco{p}{p+s}$ and $a\ne b$, then $\Pr_{\Positive_{p,t}}[r]=0$. Otherwise, $\Pr_{\Positive_{p,t}}[r]\ge \Pr_{\Negative_{p,t}}[r]$.
 \end{lemma}
    \begin{proof}
    Since $A$ is deterministic and $r$ is admissible, we have $\Pr_D[r] = \Pr_D[\bigwedge_{(S,i,a)\in r} S[i]=a]$.
    By definition of $\Positive_{p,t}$, we have $\Pr_{\Positive_{p,t}}[P[i]=a \wedge T[i+t]=b]=0$ for every $i\in \fragmentco{p}{p+s}$ and $a\ne b$.
    On the other hand, \[\Pr_{\Positive_{p,t}}[P[i]=0 \wedge T[i+t]=0]=(1-\tfrac{2k}{m}) \ge (1-\tfrac{2k}{m})^2 = \Pr_{\Negative_{p,t}}[P[i]=0 \wedge T[i+t]=0]\]
    and 
    \[\Pr_{\Positive_{p,t}}[P[i]=1 \wedge T[i+t]=1]=\tfrac{2k}{m} \ge (\tfrac{2k}{m})^2 = \Pr_{\Negative_{p,t}}[P[i]=1 \wedge T[i+t]=1].\]
    The remaining characters are independent and distributed identically under both $\Positive_{p,t}$ and $\Negative_{p,t}$, and thus they do not affect the ratio between $\Pr_{\Positive_{p,t}}[r]$ and $\Pr_{\Negative_{p,t}}[r]$.
\end{proof}

Next, we prove that, in order to succeed on $\Positive_{p,t}$ and $\Negative_{p,t}$, the algorithm needs to read sufficiently many aligned characters of $S_P=P\fragmentco{p}{p+s}$ and $S_T=T\fragmentco{p+t}{p+t+s}$.
\begin{lemma}\label{lem:hpt}
If a deterministic non-adaptive algorithm solves \cref{pr:tester} correctly with probability at least $0.7$ under the input distribution $\Hybrid_{p,t}$ for some $p\in \fragmentcc{0}{m-s}$ and $t\in \fragmentco{0}{\Delta}$, then the algorithm reads both $P[p+j]$ and $T[t+p+j]$ for at least $\frac{s}{20k}$ positions $j\in \fragmentco{0}{s}$.
\end{lemma}
\begin{proof}
Let $R \coloneqq (R_P-p)\cap (R_T-(p+t))\cap \fragmentco{0}{s}$ be the set of positions $j\in \fragmentco{0}{s}$ such that the algorithm reads both $P[p+j]$ and $T[t+p+j]$.
For a proof by contradiction, suppose that $|R| \le \frac{s}{20k}$.

In the following claim, we consider the set $\Bad\subseteq \Runs$ of runs which are more likely under $\Negative_{p,t}$ than under $\Positive_{p,t}$, and we prove that the algorithm is unlikely to encounter a run in $\Bad$.
\begin{claim}\label{clm:mismatch}
The set $\Bad = \{r \in \Runs : \Pr_{\Positive_{p,t}}[r] <  \Pr_{\Negative_{p,t}}[r]\}$
satisfies $\Pr_{\Negative_{p,t}}[\Bad] \le 0.2$.
\end{claim}
\begin{claimproof}
Observe that, for every $j\in R$, both $P[p+j]$ and $T[t+p+j]$ are distributed according to $\Ber(\frac{2k}{s})$.
Hence, the total expected number of ones at these positions does not exceed $2|R|\cdot \frac{2k}{s} \le 0.2$.
By Markov's inequality, with probability at least $0.8$, we have that $P[p+j]=T[t+p+j]=0$ holds for all $j\in R$.
For all runs satisfying this property, \cref{lem:mismatch2} implies $\Pr_{\Positive_{p,t}}[r]\ge \Pr_{\Negative_{p,t}}[r]$, and thus none of them belongs to $\Bad$.
Hence,  $\Pr_{\Negative_{p,t}}[\Bad] \le 0.2$ holds as claimed.
\end{claimproof}

Suppose that the algorithm errs with probability $e_E$ on $\Positive_{p,t}$ and $e_I$ on $\Negative_{p,t}$.
By \cref{lem:negative}, the answer \textsc{No} is required on $\Negative_{p,t}$ with probability strictly larger than $0.8$.
Hence, the set of $\RunsNo\subseteq \Runs$ on which the algorithm answers \textsc{No} satisfies $\Pr_{\Negative_{p,t}}[\RunsNo] > 0.8-e_I$.
By \cref{obs:positive}, all \textsc{No} answers are incorrect under $\Positive_{p,t}$, so $e_E\ge \Pr_{\Positive_{p,t}}[\RunsNo]$.
By \cref{clm:mismatch},
\begin{multline*}e_E \ge \Pr_{\Positive_{p,t}}[\RunsNo] \ge \Pr_{\Positive_{p,t}}[\RunsNo\setminus \Bad] \ge \Pr_{\Negative_{p,t}}[\RunsNo\setminus \Bad] \\ \ge \Pr_{\Negative_{p,t}}[\RunsNo]- \Pr_{\Negative_{p,t}}[\Bad] > 0.8-e_I-0.2=0.6-e_I.\end{multline*}
The error probability on $\Hybrid_{p,t}$ is thus equal to $\tfrac12(e_E+e_I) > 0.3$; a contradiction.
\end{proof}

The following lemma builds upon \cref{lem:hpt} to characterize non-adaptive algorithms that succeed under $\Hybrid_p$.
Although the statement resembles \cref{prp:lb_adaptive} (if $p=0$ and $s=m$, it quantifies the total number of characters that the algorithm reads), it crucially requires the restriction to non-adaptive algorithms.
This is because the zeroes surrounding $S_T$ may reveal its location within $T$.

\begin{lemma}\label{lem:hp}
If a deterministic non-adaptive algorithm solves \cref{pr:tester} correctly with probability at least $0.8$ under the input distribution $\Hybrid_{p}$ for some $p\in \fragmentcc{0}{m-s}$, then the total number of characters that the algorithm reads  from $P\fragmentco{p}{p+s}$ and $T\fragmentco{p}{p+s+n-m}$ is at least $\max\left(\frac{1}{4}\sqrt{\frac{s\Delta}{k}},\frac{s+\Delta}{60k}\right)$.
\end{lemma}
\begin{proof}
On average across $t\in \fragmentco{0}{\Delta}$, the error probability on $\Hybrid_{p,t}$ does not exceed $0.2$.
Hence, by Markov's inequality, the error probability exceeds $0.3$ for at most $\frac{2\Delta}{3}$ values $t\in \fragmentco{0}{\Delta}$.
In other words, the algorithm is correct on $\Hybrid_{p,t}$ with probability at least $0.7$ for at least $\frac{\Delta}{3}$ values $t\in \fragmentco{0}{\Delta}$.

Let $R_P^p \coloneqq R_P \cap \fragmentco{p}{p+s}$ and $R_T^p \coloneqq R_T \cap \fragmentco{p}{p+s+n-m}$ denote the relevant indices of the characters that the algorithm reads.
By \cref{lem:hpt}, if the algorithm is correct on $\Hybrid_{p,t}$ with probability at least $0.7$, then $|R_P^p\cap (R_T^p-t)| \ge \frac{s}{20k}$; note that every pair in $R_P^p\times R_T^p$ contributes to $R_P^p\cap (R_T^p-t)$ for at most one value $t\in \fragmentco{0}{\Delta}$.
Hence, $|R_P^p\cap (R_T^p-t)| \ge \frac{s}{20k}$ holds for at least $\frac{\Delta}{3}$ values $t\in \fragmentco{0}{\Delta}$. We draw three conclusions from this fact:
\begin{enumerate}
    \item $|R_P^p\times R_T^p| \ge \frac{\Delta}{3}\cdot \frac{s}{20k}$, and thus $|R_P^p|+|R_T^p|\ge 2\sqrt{|R_P^p\times R_T^p|} \ge 2\sqrt{\frac{s\Delta}{60k}} > \frac{1}{4}\sqrt{\frac{s\Delta}{k}}$.
    \item $|R_P^p\times R_T^p| \ge \frac{\Delta}{3}\cdot \frac{s}{20k}$ and, due to $|R_P^p|\le s$, we must have
    $|R_T^p|\ge \frac{\Delta}{60k}$.
    \item $|R_P^p| \ge \frac{s}{20k} > \frac{s}{60k}$.
    \qedhere
\end{enumerate}
\end{proof}

We are now ready to prove that a non-adaptive algorithm needs to make many queries to achieve low error probability on $\Hybrid$.
\begin{proposition}\label{prp:lb_nonadaptive}
    Every non-adaptive deterministic algorithm solving \cref{pr:tester} correctly with probability at least $0.9$ for the input distribution $\Hybrid$ makes at least $\tfrac{1}{204}\min\left(\frac{n\sqrt{\Delta}}{k},\sqrt{\tfrac{nm}{k}}+\tfrac{n}{k}\right)$ queries to the strings $P$~and~$T$.
\end{proposition}
\begin{proof}
    On average across $p\in \fragmentcc{0}{m-s}$, the error probability under $\Hybrid_{p}$ does not exceed~$0.1$.
    Hence, by Markov's inequality, the error probability exceeds $0.2$ for at most $\frac{m-s+1}{2}$ values $p\in \fragmentcc{0}{m-s}$
    In other words, the algorithm must be correct on $\Hybrid_{p}$ with probability at least $0.8$ for at least $\frac{m-s+1}{2}$ values $p\in \fragmentcc{0}{m-s}$.
    There is at least one such value, and thus \cref{lem:hp} yields the following lower bound on the total number $q=|R_P|+|R_T|$ of characters read:
    \begin{equation}\label{eq:one}
         q\ge \max\left(\frac{1}{4}\sqrt{\frac{s\Delta}{k}},\frac{s+\Delta}{60k}\right).
    \end{equation}

    Each character of $P$ belongs to $P\fragmentco{p}{p+s}$ for at most $s$ values $p\in \fragmentcc{0}{m-s}$, whereas each character of $T$ belongs to $T\fragmentco{p}{p+s+n-m}$ for at most $s+\Delta$ values $p\in \fragmentcc{0}{m-s}$.
    Hence, by \cref{lem:hp}, the total number of characters read must be at least 
    \begin{equation}\label{eq:many}
    q\ge \frac{m-s+1}{2(s+\Delta)}\cdot \max \left(\frac{1}{4}\sqrt{\frac{s\Delta}{k}},\frac{s+\Delta}{60k}\right).
    \end{equation}
    A convex combination of \eqref{eq:one} and \eqref{eq:many} with weights $\frac13$ and $\frac23$, respectively, yields a lower bound 
    \begin{equation}\label{eq:combined}
    q\ge \left(\frac13 +\frac23\cdot \frac{m-s+1}{2(s+\Delta)}\right)\cdot \max \left(\frac{1}{4}\sqrt{\frac{s\Delta}{k}},\frac{s+\Delta}{60k}\right) >
    \frac{n}{3(s+\Delta)}\cdot \max \left(\frac{1}{4}\sqrt{\frac{s\Delta}{k}},\frac{s+\Delta}{60k}\right) \ge \frac{n}{180k}.
    \end{equation}

    We consider three cases depending on the value of $s$.
    \begin{enumerate}
        \item $s=\Delta$: In this case, \eqref{eq:combined} yields
        \[q\ge  \frac{n}{3(s+\Delta)}\cdot\frac{1}{4}\sqrt{\frac{s\Delta}{k}} = \frac{n}{6\Delta}\cdot \frac{\Delta}{4\sqrt{k}} = \frac{n}{24\sqrt{k}} \ge \frac{1}{24}\sqrt{\frac{nm}{k}}.\]
        A convex combination with $q\ge \frac{n}{180k}$ gives the claimed lower bound.
        \item $s=4k>\Delta$: In this case, \eqref{eq:combined} implies
        \[q\ge  \frac{n}{3(s+\Delta)}\cdot\frac{1}{4}\sqrt{\frac{s\Delta}{k}} = \frac{n}{3(4k+\Delta)}\cdot \frac{\sqrt{\Delta}}{2} \ge \frac{1}{48}\cdot \frac{n\sqrt{\Delta}}{k} > \frac{1}{204}\cdot \frac{n\sqrt{\Delta}}{k}.\]
        \item $s=m <\Delta$: In this case, $\Delta-1=n-m \ge m$, so $\Delta > \max(n-m,m)\ge \frac{n}{2}$ and  \eqref{eq:one} yields
        \[q\ge \frac{1}{4}\cdot \sqrt{\frac{s\Delta}{k}} \ge \frac{1}{4\sqrt{2}}\cdot \sqrt{\frac{mn}{k}} > \frac{1}{24}\sqrt{\frac{nm}{k}}.\]
        A convex combination with $q\ge \frac{n}{180k}$ gives the claimed lower bound.\qedhere
    \end{enumerate}
\end{proof}

\cref{thm:lb_nonadaptive} restricted to $k \le \frac{m}{4}$ follows from \cref{prp:lb_nonadaptive} by Yao's principle.
For $\frac{m}{4} < k < m$, we derive \cref{thm:lb_nonadaptive} from \cref{thm:lb_adaptive-largek}.
We note that every algorithm solving \cref{pr:tester} for any $k\in \fragmentco{1}{m}$ also solves the problem for $k=m-1$.
Thus, \cref{thm:lb_adaptive-largek} yields a lower bound of $\frac1{6}\left(\sqrt{n-m+1}+\frac{n}{m}\right)$.
For $k > \frac{m}{4}$, we have 
\[\tfrac{1}{204}\min\left(\tfrac{n\sqrt{n-m+1}}{k},\sqrt{\tfrac{nm}{k}}+\tfrac{n}{k}\right)
    < \tfrac{1}{204}\min\left(\tfrac{4n\sqrt{n-m+1}}{m},2\sqrt{n}+\tfrac{4n}{m}\right).
\]
If $n \le 2m$, then $\tfrac{4n\sqrt{n-m+1}}{m}<8\sqrt{n-m+1}$.
Otherwise, $\sqrt{n} < \sqrt{2(n-m+1)} < 8\sqrt{n-m+1}$.
In both cases, we have 
\[\tfrac{1}{204}\min\left(\tfrac{n\sqrt{n-m+1}}{k},\sqrt{\tfrac{nm}{k}}+\tfrac{n}{k}\right)
    < \tfrac{8}{204}\left(\sqrt{n-m+1}+\tfrac{n}{m}\right)
    < \tfrac{1}{6}\left(\sqrt{n-m+1}+\tfrac{n}{m}\right),
\]
and thus \cref{thm:lb_nonadaptive} follows.

\section{\boldmath Lower Bound for $k=m-1$}\label{sec:lb_large}
In this section, we prove the following theorem:
\thmlbadaptivelargek*
Note that \cref{thm:lb_adaptive-largek} is false for binary strings, in contrast to the $k\le m/4$ case in \cref{sec:lb_adaptive}. To see this, observe that when $n\ge m+1$, for any $(P,T)\in \{0,1\}^m \times \{0,1\}^n$, the \No case $\Occ_{m-1}(P,T)= \emptyset$ holds only if $(P,T)\in \{(0^m,1^n),(1^m,0^n)\}$. Denote the fraction of ones in the string $S\in \{P,T\}$ as $\delta_S=\frac{1}{|S|}\sum_{i\in [0\dd |S|)}S[i]$. Then, the \No case implies $|\delta_P-\delta_T|=1$, whereas in the \Yes case with $i\in \Occ(P,T)$, we have
\[n\cdot |\delta_P-\delta_T|\le (n-m)\delta_P + |m\delta_P-n\delta_T|   \le(n-m) + \textstyle \sum_{j\in [0\dd n)\setminus[i\dd i+m)}T[j]   \le 2(n-m),\]
which yields $|\delta_P-\delta_T| \le 2-\frac{2m}{n}\le 0.8$ provided $m\ge 0.6n$. Therefore, a tester can estimate $\delta_P,\delta_T$ up to additive error $\tfrac{1}{100}$ with high success probability in $\Oh(\log n)$ query complexity, and obtain an estimate for $|\delta_P-\delta_T|$ which is sufficient to distinguish the \Yes and \No  cases.

The proof of \cref{thm:lb_adaptive-largek} follows the same structure as in \cref{sec:lb_adaptive}, but gets much simplified in the $k=m-1$ case.

We henceforth fix integers $1 \le  m \le n$ and denote $\Delta \coloneqq n-m+1$. Let $\sigma =|\Sigma|=10n^2$. Let $\cU(\Sigma)$ denote the uniform distribution on $\Sigma$. 

We define the following distributions on $(P,T)\in \Sigma^{m}\times \Sigma^n$ so that $P[i]\sim \cU(\Sigma)$ for $i\in \fragmentco{0}{m}$ and $T[j] \sim \cU(\Sigma)$ for $j\in \fragmentco{0}{n}$:
\begin{itemize}
    \item $\Random$, where the variables $\{P[i] : i\in \fragmentco{0}{m}\}$ and $\{T[j] : j\in \fragmentco{0}{n}\}$ are all independent;
    \item $\Planted_t$ for $t\in \fragmentco{0}{\Delta}$, where the variables $\{T[j] : j\in \fragmentco{0}{n}\}$ are all independent and $P=T[t\dd t+m)$.
    \item $\Planted$, which is a mixture of distributions $\Planted_t$ for $t\in \fragmentco{0}{\Delta}$ with weights $\frac{1}{\Delta}$ each.
    \item $\Mixed$, which is a mixture of distributions $\Random$ and $\Planted$ with weights $\frac12$ each.
\end{itemize}

By construction, the only correct answer for every instance drawn according to $\Planted$ is \textsc{Yes}.

\begin{observation}\label{obs:planted-largek}
If $(P,T)\sim \Planted$, then $\Occ(P,T)\ne \emptyset$.\lipicsEnd
\end{observation}

Moreover, the only correct answer for most instances drawn according to $\Random$ is \textsc{No}.

\begin{lemma}\label{lem:random-largek}
If $(P,T)\sim \Random$, then $\Pr[\Occ_k(P,T)=\emptyset]> 0.8$.
\end{lemma}
\begin{proof}
Let us fix $t\in \fragmentco{0}{\Delta}$. 
For each $i\in \fragmentco{0}{m}$, we have 
\[\Pr[P[i]= T[i+t]] = \sigma^{-1},\]
and these events are independent across $i\in \fragmentco{0}{m}$.
Hence $\Exp[m - \HD(P, T\fragmentco{t}{t+m})]\le m\cdot \sigma^{-1} \le \tfrac{1}{10n}$ and, by the Markov's inequality, 
\[\Pr[\HD(P, T\fragmentco{t}{t+m}) \le k] = \Pr[m-\HD(P, T\fragmentco{t}{t+m}) \ge 1]  \le \tfrac{1}{10n}.\]
By the union bound, $\Pr[\Occ_k(P,T)\ne\emptyset] \le 0.1< 0.2$.
\end{proof}

A combination of \cref{obs:planted-largek} and \cref{lem:random-largek} implies that any solution to \cref{pr:tester} must distinguish input distributions $\Planted$ and $\Random$.

We henceforth fix a \emph{deterministic} algorithm $A$ that makes exactly $q$ distinct queries (we append dummy queries when necessary to make the number of queries uniform across all executions).
We define a \emph{run} as the sequence of queries the algorithm makes and the answers it receives. 
Formally, a run is a sequence of $q$ triples $(S,i,a)$ such that $S\in \{P,T\}$, $i\in \fragmentco{0}{|S|}$, and $a=S[i]$.
We denote by $\Runs$ the family of all \emph{admissible} runs, that is, runs describing the execution of $A$ on some input.
Every input distribution $D$ yields a distribution on $\Runs$ so that $\Pr_D[r]$ is the probability that the execution of $A$ is described by $r$ provided that the input is drawn according~to~$D$.

The following lemma says that, unless a mismatch between $P$ and $T\fragmentco{t}{t+m}$ is discovered, it is more likely that the input comes from $\Planted_t$ than from $\Random$. 
The intuitive interpretation of the immediate corollary is that many such mismatches are needed across all $t\in \fragmentco{0}{\Delta}$ in order to confirm the hypothesis that the input comes from $\Random$ rather than $\Planted$.  
\begin{lemma}\label{lem:mismatch-largek}
Consider a run $r\in \Runs$ and a position $t\in \fragmentco{0}{\Delta}$. 
If $r$ contains triples $(P,i,a)$ and $(T,i+t,b)$ for some $i\in \fragmentco{0}{m}$ and $a\ne b$, then $\Pr_{\Planted_t}[r]=0$.
Otherwise, $\Pr_{\Planted_t}[r]\ge \Pr_{\Random}[r]$.
\end{lemma}
\begin{proof}
Since $A$ is deterministic and $r$ is admissible, we have $\Pr_D[r] = \Pr_D[\bigwedge_{(S,i,a)\in r} S[i]=a]$.
By definition of $\Planted_t$, we have $\Pr_{\Planted_t}[P[i]=a \wedge T[i+t]=b]=0$ for every $i\in \fragmentco{0}{m}$ and $a\ne b$.
On the other hand, for every $a\in \Sigma$, \[\Pr_{\Planted_t}[P[i]=a \wedge T[i+t]=a]=\sigma^{-1} \ge \sigma^{-2} = \Pr_{\Random}[P[i]=a \wedge T[i+t]=a].\]
The remaining characters are independent and distributed identically under both $\Planted_t$ and $\Random$, and thus they do not affect the ratio between $\Pr_{\Planted_t}[r]$ and $\Pr_{\Random}[r]$.
\end{proof}
\begin{corollary}\label{cor:mismatch-largek}
Every run $r\in \Runs$ satisfies
$\tfrac{\Pr_{\Planted}[r]}{\Pr_{\Random}[r] } \ge 1-\tfrac{\min(mq, q^2)}{\Delta}$.
\end{corollary}
\begin{proof}
Recall that $\Pr_{\Planted}[r] =\tfrac{1}{\Delta}\sum_{t=0}^{\Delta-1} \Pr_{\Planted_t}[r]$.
By \cref{lem:mismatch-largek}, $\Pr_{\Planted_t}[r] \ge \Pr_{\Random}[r]$ holds unless $(P,i,a),(T,i+t,b)$ belong to $r$ for some $i\in \fragmentco{0}{m}$ and $a\ne b$. The number of such indices $t\in \fragmentco{0}{\Delta}$ is upper-bounded by the number of pairs $(P,i,a),(T,j,b)\in r$ with $a\ne b$, which is at most $\min(m|r|,|r|^2) = \min(mq, q^2)$.
\end{proof}

We are now ready to prove that the algorithm $A$ needs to make many queries to achieve low error probability on $\Mixed$.

\begin{proposition}\label{prp:lb_adaptive-largek}
    Every deterministic algorithm solving \cref{pr:tester}  with $k=m-1$ correctly with probability at least $0.8$ for the input distribution $\Mixed$ makes at least $\frac{1}{6}\left(\sqrt{\Delta}+\frac{n}{m}\right)$ queries to the strings~$P$~and~$T$.
\end{proposition}
\begin{proof}
Fix a deterministic algorithm $A$ with error probabilities $e_{R}$ and $e_{P}$ under $\Random$ and $\Planted$, respectively.
For a proof by contradiction, suppose that $A$ makes $q < \frac{1}{6}\left(\sqrt{\Delta}+\frac{n}{m}\right)$ queries.
Note that
\[q < \tfrac{1}{6}\left(\sqrt{\Delta}+\tfrac{n}{m}\right) < \tfrac{1}{6}\left(\sqrt{\Delta}+\tfrac{\Delta}{m}+1\right) \le \tfrac12\max\left(\sqrt{\Delta},\tfrac{\Delta}{m},1\right).\]
In the first branch, we have $q^2 < \frac14\Delta$, in the second branch, we have $mq < \frac12\Delta$,
whereas in the third branch, we must have $q=0$.
By \cref{cor:mismatch-largek}, every $r\in \Runs$ satisfies $\tfrac{\Pr_{\Planted}[r]}{\Pr_{\Random}[r] } \ge 1-\tfrac{\min(mq,q^2)}{\Delta}$. In all three branches, we conclude $\tfrac{\Pr_{\Planted}[r]}{\Pr_{\Random}[r] } > \frac{1}{2}$.

Let $\RunsNo\subseteq \Runs$ be the set of runs on which $A$ returns \textsc{No}.
By \cref{lem:random-largek}, the answer \textsc{No} is required on $\Random$ with probability strictly larger than $0.8$,
so $\Pr_\Random[\RunsNo] > 0.8-e_R$.
By \cref{obs:planted-largek}, all \textsc{No} answers are incorrect under $\Planted$, so $e_P\ge \Pr_\Planted[\RunsNo]$.
 Then, we have
\[e_P \ge \Pr_{\Planted}[\RunsNo]\ge  \tfrac12 \Pr_\Random[\RunsNo]> 0.4-\tfrac{e_R}{2} \ge 0.4-e_R.\]
Hence, the error probability of $A$ on $\Mixed$, which equals $\tfrac12(e_P+e_R)$, is strictly larger than $0.2$.
\end{proof}

\cref{thm:lb_adaptive-largek} follows from \cref{prp:lb_adaptive-largek} by Yao's principle.

\bibliography{refs}

\end{document}